\documentclass[11pt]{article}

\usepackage{amssymb, amsfonts, amsthm, mathtools}
\mathtoolsset{showonlyrefs}
\usepackage[margin=2.5cm]{geometry}

\usepackage{booktabs, multirow}
\usepackage{makecell}
\usepackage{rotating}
\usepackage{subcaption}
\usepackage{threeparttable}

\newtheorem{thm}{Theorem}
\newtheorem{prop}{Proposition}

\newtheorem{lem}{Lemma}
\newtheorem{asm}{Assumption}
\newtheorem{rem}{Remark}

\newtheorem*{algo*}{Algorithm}
\allowdisplaybreaks

\renewenvironment{proof}[1][\proofname]{{\bfseries #1. }}{\qed}

\DeclareMathOperator*{\argmin}{\arg\!\min}

\usepackage{xcolor}
\definecolor{gold}{rgb}{1.0, 0.84, 0.0}
\usepackage{tikz}
\DeclareRobustCommand\full  {\tikz[baseline=-0.6ex]\draw[thick] (0,0)--(0.5,0);}
\DeclareRobustCommand\dotted{\tikz[baseline=-0.6ex]\draw[thick,dotted] (0,0)--(0.54,0);}
\DeclareRobustCommand\dashed{\tikz[baseline=-0.6ex]\draw[thick,dashed] (0,0)--(0.54,0);}
\DeclareRobustCommand\longdash{\tikz[baseline=-0.6ex]\draw[thick,dash pattern=on 8 off 4] (0,0)--(0.54,0);}

\usepackage{setspace}
\usepackage[style = authoryear, backend = biber, url = false, doi = false, isbn = false, eprint = true, maxcitenames=99, maxbibnames = 99, uniquename = false]{biblatex} 
\title{Inference methods for unit-specific coefficients in panel data models with latent group structure\thanks{Nishi acknowledges financial support from the Japan Society for the Promotion of Science under KAKENHI Grant No. 25KJ0041. Okui acknowledges financial support from the Japan Society for the Promotion of Science under KAKENHI Grant Nos. 23K25501. All errors are our own.}}

\author{Mikihito Nishi\thanks{Graduate School of Economics, University of Tokyo, 7-3-1 Hongo, Bunkyou-ku, Tokyo, 113-0033, Japan. Email: mnishi@g.ecc.u-tokyo.ac.jp} \ and Ryo Okui\thanks{Faculty of Economics, the University of Tokyo, Bunkyo-ku, Tokyo 113-0033, Japan. Email: okuiryo@e.u-tokyo.ac.jp}}

\date{\today}

\begin{document}

\maketitle

\begin{abstract}

  This paper introduces statistical inference procedures for unit-specific coefficients in panel data models, where the coefficients exhibit a latent group structure. The proposed methods achieve efficiency gains by clustering units into a small number of groups, while explicitly accounting for the statistical uncertainty of group assignments. The core idea is to integrate standard inference procedures, such as the $t$-test and Wald tests, with confidence sets for group membership. Two methods are proposed: the first takes the minimum of the test statistics over the confidence set for group membership, and the second corrects for bias caused by possible group misassignment. The former can produce shorter but possibly disconnected sets, while the latter guarantees connected, interpretable intervals at some cost in length. We also develop standard errors that are adjusted for possible group misassignment and valid even with short time periods, which may be of independent interest. Monte Carlo simulations demonstrate that our approach yields narrower confidence sets for units with relatively large error variances than unit-by-unit time-series methods. In contrast, ignoring statistical uncertainty in the group membership estimation leads to distortions in size and coverage. We illustrate the method with an empirical example that estimates the effect of the minimum wage in each U.S. state.

\medskip

\textbf{Keywords}: Panel data, latent group structure, confidence set, unit-specific coefficients

\textbf{JEL classifications}: C23, C33, C38

\end{abstract}

\section{Introduction}

Heterogeneity plays a crucial role in panel data analysis, yet modeling and accounting for it in statistical inference pose numerous challenges. In the literature, various methods have been developed to address heterogeneity across units. One prominent approach is the latent group structure model, which posits that units can be partitioned into a few groups, with units within each group homogeneous and those across groups exhibiting distinct characteristics \parencite[see e.g.,][for a review]{pionati2025latent}. This framework is both flexible enough to capture complex heterogeneity and parsimonious enough to remain interpretable. However, inference methods for heterogeneous coefficients are still limited. Existing works mostly focus on group-level inferences.

In many applications, the objects of ultimate interest are the coefficients of individual units rather than those of the groups to which units are assigned. Researchers and policymakers typically need to draw conclusions about a particular unit, such as the effect of a policy in a specific state. In our empirical application, for instance, each U.S. state can, to some extent, set its own minimum wage, so the employment effect relevant to a given state's policy debate is that state's own coefficient. From this perspective, a concrete decision is guided by the parameter attached to the unit itself, and the group structure is primarily an estimation device that improves precision by pooling information across similar units. Reporting only group-level coefficients would fulfill these purposes if the group structure were known, but doing so can be misleading when a unit's group membership is uncertain. In contrast, an alternative of unit-by-unit time series regression provides unbiased inference but is inefficient. Furthermore, unit-by-unit regression may not be applicable in models such as the time fixed effects model.\footnote{Indeed, unit-by-unit regression cannot be applied in our empirical application due to the presence of time fixed effects and the lack of variation within units.} These observations motivate inference procedures that target unit-specific coefficients while still exploiting the efficiency gains afforded by the group structure.

In this paper, we develop new methods for statistical inference on unit-specific coefficients when these coefficients exhibit a latent group structure. If a unit's group membership were known, inference would simply involve the corresponding group-level coefficient. However, a central challenge in making valid statistical inferences on the coefficient for each unit is that group memberships are usually unobserved and must be estimated from the data, which introduces additional statistical uncertainty. This is particularly relevant in situations where group-specific parameters can be estimated precisely, but where uncertainty remains in the group membership structure. \textcite{dzemskiokui2021convergence} show that such situations occur when the group memberships of a sufficient number of units—not necessarily all—are consistently estimated, while other units with larger error variances have uncertain group assignments (namely their group memberships may not be consistently estimated). In such a scenario, making statistical inferences for units with relatively large error variances by relying solely on coefficient estimates for the estimated group can lead to misleading conclusions because it ignores statistical uncertainty.

To address this issue, our approach integrates standard inference procedures for group-specific coefficients, such as $t$-tests and Wald tests, with the confidence set methodology for group membership proposed by \textcite{dzemski2024confidence}. This combination ensures that our inference procedures properly account for the uncertainty in group membership estimation and, at the same time, enjoy the efficiency of group-wise estimation.

We propose two methods for unit-level statistical inference. The first approach relies on minimum-type test statistics. It constructs a confidence set by inverting the minimum of the $t$- or Wald test statistics for group-specific parameters, with the minimum taken over the groups included in a confidence set for the group membership. The underlying idea is that, under the true null hypothesis, the test statistic should be small when the group assignment is correct; therefore, it should be small for at least one group in the confidence set. This construction suggests that the resulting confidence set can be represented as a union of confidence sets built for each group. The minimum-type test statistics method may produce confidence sets that are disconnected. This possible disconnectedness has both advantages and disadvantages: it can reduce the volume of the confidence set and increase power, but in practice, a disconnected confidence set may be harder to interpret or communicate.

The second method employs bias correction. The coefficient for a given unit is estimated using its group-specific coefficient, but this estimator can be biased if the group assignment is inaccurately determined. Our approach addresses this bias by explicitly accounting for potential group misclassification. Specifically, we adjust the test statistic by considering the difference between the coefficient under the estimated group membership and those of other plausible groups to which the unit may belong, as indicated by the confidence set of \textcite{dzemski2024confidence}. The bias correction can be applied to the $t$-statistic for a single parameter and to the Wald test statistic for multiple parameters. In contrast to the minimum-type approach, the bias-correction approach necessarily produces a connected confidence interval for a single parameter. It is important to note that bias correction for the Wald test statistic is nontrivial, as the bias can affect the statistic nonlinearly. The bias correction ensures that the confidence set achieves the desired coverage probability. However, this approach is inherently conservative.

We also propose an asymptotic variance estimator valid under fixed-$T$ asymptotics, extending \textcite[][S.2.2]{bonhomme2015grouped} to models with heterogeneous coefficients. This asymptotic variance estimator is robust to group misallocation and better suited to our setting. We find that it achieves better coverage and does not necessarily make inference conservative. This contribution may also be of independent interest.

Our proposed procedures are straightforward to implement, both conceptually and computationally. They require only the group-specific coefficient estimates, their standard errors, and unit-level confidence sets for group memberships. Our Monte Carlo simulations indicate that the most computationally intensive step is the estimation of the group membership structure and group-specific coefficients. However, once these are obtained, the calculation of confidence sets for group memberships is rapid, as documented in \textcite{dzemski2024confidence}. With this preliminary information in hand, applying our procedure is seamless.

As an empirical illustration, we examine the effect of a minimum-wage increase in each US state. The impact of the minimum wage varies across states. It is therefore important to estimate state-specific effects. This is especially relevant since states can set their own policies. We use the data from \textcite{dube2010minimum} and extend the analysis by \textcite{wang2019heterogeneous,dzemski2024confidence}. In those studies, states are grouped into four based on similarities in coefficients. 
The results highlight the importance of accounting for group membership uncertainty. Specifically, when the groups differ mainly in covariate coefficients rather than the minimum wage effect itself, ignoring uncertainty in group assignments can be misleading. Our confidence sets address this issue by explicitly incorporating group assignment uncertainty.

The Monte Carlo simulations demonstrate the situations in which our proposed methods are most effective. We compare our methods to those that ignore statistical uncertainty in group membership estimation (naive) and to the unit-by-unit estimation approach. The naive method fails to provide adequate coverage for units whose group memberships remain uncertain, leading to severe undercoverage. The unit-by-unit method achieves adequate coverage, but its interval length can be large unless the error variance is small. In contrast, our methods maintain adequate coverage and, while our confidence intervals may be wider than those from the unit-by-unit approach for units with small error variance, they are significantly shorter when the error variance is moderate to large. This advantage is especially notable when confidence intervals are allowed to be unions of disjoint intervals. We therefore conclude that our methods are particularly beneficial for units with moderate to large error variances.

Our contributions are twofold. First, while earlier studies have primarily investigated asymptotic distributions and inference for group-specific coefficients, they do not address inference for unit-level coefficients. Second, our approach does not rely on the result that the estimated group structure converges to the true membership with probability approaching one (so-called super-consistency). Note that under super-consistency, inferences for group-specific and unit-specific coefficients become equivalent. Our method incorporates the statistical uncertainty associated with estimating group memberships.

\paragraph{Literature review}

This paper contributes to the recent econometrics literature on latent group structure. A partial list includes \textcite{hahn2010panel, lin2012estimation, bonhomme2015grouped, sarafidis2015partially, ando2016panel,su_identifying_2016, wang2016homogeneity, liu2020identification, mehrabani2022estimation,chetverikov2022spectral,yu2024spectral, mugnier2025simple}.
\textcite{pionati2025latent} provides an excellent review.
This paper addresses the underdeveloped area of inference methods for panel data models with latent group structure, specifically focusing on unit-specific coefficients. To our knowledge, this problem has not been previously tackled.  \textcite{wan2025conditionalselectiveinferenceselected,akgun2025robustinferencemethodslatent} consider post-clustering inference in panel data, but their focus is again on group-specific coefficients. \textcite{beyhum2024inference} also consider post-clustering inference, but their interest is in the parameter common to all units.

Our approach combines standard statistical inference methods for the parameter of interest with a confidence set for the nuisance parameter, which, in our case, is the group membership. Our minimum-type approach builds on a long tradition of obtaining valid inference by combining over a confidence set for a nuisance parameter. \textcite{Dufour1990} developed this approach for time series models. In statistics, \textcite{BergerBoos1994} and \textcite{Silvapulle1996} independently formalized computing a valid $p$-value by taking the supremum of $p$-values over a confidence set for the nuisance parameter. The same principle has been adapted to post-moment-selection inference by \textcite{DiTraglia2016} in moment restriction models for GMM estimators. A related idea is used by \textcite{kaji2025controlling} for a financial application. A key distinction between the earlier works and the present paper is that the nuisance parameter in our model is discrete, while the earlier works considered continuous nuisance parameters; the discrete nature of the nuisance parameter may result in disconnected confidence intervals. The bias-correction approach, by contrast, relates to the literature on bias-aware confidence intervals built from bias bounds.  For instance, \textcite{ArmstrongWeidnerZeleneev2025} construct bias-aware confidence intervals for interactive fixed effects models based on a possible set of interactive fixed effects.

Our method is related to empirical Bayes in terms of purpose, but is fundamentally different from it in both approach and implementation. Empirical Bayes is a widely used approach for estimating and inferring unit-specific parameters and unit-specific predictions. Recent econometric articles on this topic include \textcite{liu2020forecasting,liu2021panel,liu2023forecasting,kwon_optimal}. The core principle of empirical Bayes methods is to utilize information from other units to improve the precision of these estimates. These methods typically begin with unit-level statistics and subsequently adjust them using information from the broader population, thereby enhancing statistical inference. In contrast, the proposed method adopts a different perspective. It assumes the presence of a latent group structure, enabling unit-specific parameters to be estimated as parameters shared among all units within a group. After estimating group-level parameters, the method incorporates information from the unit of interest to ensure robustness against statistical errors arising from the clustering process.

\paragraph{Roadmap} Section \ref{sec:setting} introduces the setting. Section \ref{sec:method} describes our two new methods for inference on unit-specific parameters. Their asymptotic properties are given in Section \ref{sec:asym}. Section \ref{sec:asym_var} develops asymptotic variance estimators, including one valid under fixed-$T$ asymptotics. Section \ref{sec:empirical} provides an empirical illustration, and Section \ref{sec:simu} presents Monte Carlo simulation results. Section \ref{sec:conclusion} concludes the paper. The appendix contains the proofs of the main results, implementation details, and additional discussions.

\section{Settings and preliminaries}\label{sec:setting}

We consider panel linear regression models with heterogeneous coefficients. Specifically, the coefficients exhibit a grouped pattern of heterogeneity. Our objective is to conduct statistical inference on the coefficients for each unit, accounting for this latent group structure.

Suppose we observe panel data $\{(y_{it}, x_{it}) : t = 1, \dots, T;\ i = 1, \dots, N\}$, where $y_{it}$ denotes the scalar outcome for unit $i$ at time $t$, and $x_{it}$ is a $p \times 1$ vector of regressors. 

We consider the following linear regression model with heterogeneous coefficients:
\begin{align}
   y_{it} = x_{it}' \beta_i + \varepsilon_{it},
\end{align}
where $\beta_i$ is the $p\times 1 $ coefficient vector for unit $i$, $\varepsilon_{it}$ is the error term, and $\beta_i^0$ denotes the true value of $\beta_i$. 

The coefficients are assumed to follow a grouped pattern of heterogeneity: units are partitioned into $G$ distinct groups, denoted by $\mathbb{G} = \{1, \dots, G\}$. Let $g_i \in \mathbb{G}$ represent the group membership for unit $i$, with $g_i^0$ the true group assignment. All units within the same group share a common coefficient vector, while units in different groups have distinct coefficient vectors. Let $\theta_g\in\Theta$ denote the coefficient vector for group $g$, and $\theta_g^0$ its true value. Formally, we assume
\begin{align}
   \beta_i^0 = \theta_{g_i^0}^0.
\end{align}

Our objective is to conduct statistical inference on $\beta_i$. With sufficiently long panel data, one could, in principle, estimate $\beta_i$ for each unit individually using time-series regressions. However, this approach can be inefficient, especially when units share a grouped pattern of heterogeneity. By pooling cross-sectional information from units within the same group, we can make more precise inferences about $\theta_{g_i}$, the common coefficient for group $g_i$. The main challenge is that group memberships are unknown and must be estimated from the data. This introduces additional statistical uncertainty, making it non-trivial to fully leverage the cross-sectional information while properly accounting for the estimation of group structure. Indeed, ignoring this uncertainty may lead to biased inference as we argue below.

\subsection{The source of bias}\label{subsec:bias}

Suppose that some estimates $\hat{\theta}_g$ and $\hat{g}_i$ for the group-level slopes and the group memberships are available. We estimate $\beta_i$ using $\hat{\theta}_{\hat g_i}$. There are two main sources of statistical uncertainty: (1) the estimation error of $\hat{\theta}_{g}$ as an estimator of $\theta_g$, and (2) uncertainty in the estimated group membership $\hat g_i$. For the first, the literature provides valid procedures to evaluate the statistical uncertainty of group-specific coefficients. However, to make inferences for $\beta_i$, the second issue, namely an additional bias caused by the misclassification of group membership, must be addressed.

To be more specific, the estimation error of $\hat{\theta}_{\hat{g}_i}$ can be decomposed as
\begin{align}
    \hat{\theta}_{\hat{g}_i} - \theta^0_{g_i^0} = \left(\hat{\theta}_{g_i^0} - \theta^0_{g_i^0}\right) + \left(\hat{\theta}_{\hat{g}_i} - \hat{\theta}_{g_i^0}\right) .
    \label{eqn:decompose}
\end{align}
In the above decomposition, the first term converges in distribution to a normal distribution, whereas the second term represents the bias induced by potential misclassification. 
The literature often assumes that the probability of $\hat g_i = g_i^0$ converges to one, effectively treating the bias term as negligible. However, in finite samples, this bias may remain substantial and should not be ignored. Inference about the unit-specific parameters requires accounting for uncertainty in the estimation of $g_i$.

In what follows, we will explain how this statistical uncertainty can be quantified. We note that we are agnostic about the procedures used in the estimation of $\theta_g$ and the construction of confidence sets for $g_i$, as long as the conditions stated later are satisfied. For concreteness, we discuss the KMeans-type procedure given in \textcite{bonhomme2015grouped} and \textcite{okui_heterogeneous_2021} for the estimation of the group-specific parameters and the procedure by \textcite{dzemski2024confidence} for the confidence set for group membership.

\subsection{Estimation of the group-level parameters}\label{subsec:est_group_slope}

The first step is to estimate the group-specific parameter $\theta_{g}$. It may be estimated by minimizing the least square objective function with respect to both the group-specific coefficients and group structure. We denote the group-specific coefficient estimator by $\hat{\theta}_g$ for $g\in \mathbb{G}$. 
Also, let $\hat{g}_{i}$ denote the estimates for $g_i^0$, for $i=1, \dots,N$. The estimator is defined as:
\begin{align}
\{ \{\hat \theta_g\}_{g \in \mathbb{G}}, \{ \hat g_i\}_{i=1,\dots,N} \}=  \argmin_{ \{ \theta_g\}_{g \in \mathbb{G}}, \{ g_i\}_{i=1,\dots,N} }
\sum_{i=1}^N \sum_{t=1}^T (y_{it} - x_{it}'\theta_{g_i})^2.
\end{align}

This minimization problem may be solved by extending the KMeans algorithm. 
\begin{enumerate}
    \item $s=0$. Let $\hat{g}_i^{(0)}$ be the initial group assignment.
    \item Compute $\theta_{g}^{(s)}$ for $g \in \mathbb{G}$ using the observations belonging to group $g$:
    \begin{align}
        \hat{\theta}_{g}^{(s)} = \argmin_{\theta_g} \sum_{\hat{g}_i^{(s)}=g} \sum_{t=1}^T (y_{it} - x_{it}'\theta_g)^2.
    \end{align}
    \item Update the group membership assignment based on $\hat{\theta}_{g}^{(s)}$:
    \begin{align}
        \hat{g}_{i}^{(s+1)} = \argmin_{g \in \mathbb{G}} \sum_{t=1}^T (y_{it} - x_{it}'\hat{\theta}_{g}^{(s)})^2.
    \end{align}
    \item $s= s+1$.
    \item Repeat 2-4 until convergence.
\end{enumerate}
The algorithm extends KMeans clustering to panel-data regression models and is considered by \textcite[][S.4.2]{bonhomme2015grouped} in the context of models with grouped fixed effects and \textcite{okui_heterogeneous_2021} for models with time-varying coefficients with group structure. Step 1 initializes the process by assigning each individual to an initial group. In Step 2, group-specific regression coefficients are estimated by minimizing the sum of squared residuals within each group, using only the observations currently assigned to that group. Step 3 reassigns each individual to the group whose estimated parameters best fit their data, specifically selecting the group that minimizes the squared prediction error aggregated across all time periods. Step 4 increments the iteration counter. Steps 2 through 4 are repeated until convergence, defined as the point at which group assignments stabilize between iterations. The algorithm alternates between estimating group-specific parameters and reassigning observations to groups, analogous to the standard KMeans algorithm but adapted for regression contexts.
This algorithm always converges to a local minimum. To ensure the global minimum, we would need to run this algorithm with many initial values. 

We also need two additional statistics that can be computed from the estimation results.
Let $\widehat{N}_{g}=\Sigma_{i=1}^{N}1\left\{\hat{g}_{i}=g\right\}$ denote the estimates for $N_g=\Sigma_{i=1}^{N}1\left\{g_{i}^0=g\right\}$ (the number of units in group $g$). We also use $\widehat{\Sigma}_{g}$ to denote the estimates for the asymptotic variance-covariance matrices $\Sigma_g$ of $\sqrt{N_{g} T}\left(\hat{\theta}_g-\theta_{g}^{0}\right)$. Section \ref{sec:asym_var} discusses how to obtain $\widehat{\Sigma}_{g}$.

We note that many alternative estimators are available in the literature \parencite[e.g.,][]{su_identifying_2016, wang2016homogeneity, mehrabani2022estimation}. Our proposed inference procedures are agnostic about the choice of the estimator. All estimators that satisfy Assumptions \ref{asm:asym_normal}--\ref{asm:vcov_consistent} given below can be used.

\subsection{Confidence set for group membership}\label{subsec:cs_membership}

Next, we construct a confidence set for the group membership parameter $g_i$ by following the procedure outlined in \textcite{dzemski2024confidence}. The central insight is that, for each unit, the sum of squared errors is minimized when the coefficient corresponds to the true group membership. To obtain the confidence set, the procedure inverts a test that examines whether this minimum is achieved under the parameter value associated with the group specified by the null hypothesis. We note that while \textcite{dzemski2024confidence} center on confidence sets for the entire group structure, our discussion here focuses on a marginal confidence set for each unit, as our interest is in inferring the coefficient for each unit.

The first step is to compute the adjusted difference between the squared losses under the two group membership assignments. Let
\begin{align*}
    \hat{d}_{it}(g, h) = \frac{1}{2} \left( \left(y_{it} -  x_{it}'\hat{\theta}_{g}\right)^2 - \left(y_{it} - x_{it}'\hat{\theta}_{h}\right)^2
    + \left(x_{it}' \left(\hat{\theta}_{g} - \hat{\theta}_{h}\right)\right)^2 \right).
\end{align*}
The first two terms on the right-hand side are squared residuals representing the fit of assigning unit $i$ to group $g$ and to group $h$, respectively. The third term is an adjustment that recenters the difference between the two squared residuals so that it is asymptotically mean zero under the null hypothesis $g_i^0 = g$. The $t$-statistic is then computed as follows:
\begin{align*}
    \widehat{D}_{i}(g, h) = \frac{\sum_{t = 1}^T \hat{d}_{it}(g, h) /\sqrt{T}}{ \sqrt{ \widehat{\Xi}_i (g,h,h)} },
\end{align*}
where $\widehat{\Xi}_i (g,h,h)$ is a long-run variance estimator of $\hat{d}_{it}(g, h)$. Details regarding the calculation of $\widehat{\Xi}_i (g,h,h)$ are provided in Appendix \ref{appendix: confsetgroup}.

The test statistic is defined as the maximum $t$-statistic, taken over all $h \neq g$, where $g$ is the group membership assignment under the null hypothesis. The test statistic is given by:
\begin{align*}
    \widehat{T}_i(g) =& \max_{h \in \mathbb{G} \setminus \{g\}} \widehat{D}_{i}(g, h).
\end{align*}
If the null hypothesis $g_i^0=g$ is true, all values of $\widehat{D}_{i}(g, h)$ are unlikely to be large and positive. The null hypothesis is rejected when $\widehat{T}_i(g)$ exceeds the critical value. This test is one-sided.

The critical value is computed from the distribution of the maximum of a random vector following a multivariate $t$-distribution and is denoted as $\hat{c}_{i, \alpha}(g)$, where $\alpha$ is a prespecified size.\footnote{The discussion here is for a general $\alpha$. Later, we use $\alpha_1$ for the confidence set for group memberships.} In the case of two groups ($G=2$), the critical value is computed from the usual $t$-distribution with degrees of freedom $T-1$:
\begin{align*}
    \hat{c}_{i, \alpha} (g)
    = \sqrt{\frac{T}{T - 1}} t_{T- 1}^{-1}\left(1 - \alpha\right),
\end{align*}
where $t_{T - 1}^{-1}(p)$ denotes the $p$-quantile of Student's $t$-distribution with $(T - 1)$ degrees of freedom. 
For $G \geq 3$, the critical value is given by
\begin{align*}
    \hat{c}_{i, \alpha} (g) =& c_{\alpha} \left( \rho(\widehat{\Omega}_i(g), \epsilon_N) \right) = \sqrt{\frac{T}{T - 1}} \left(t_{\max, \rho(\widehat{\Omega}_i(g), \epsilon_N), T - 1}\right)^{-1} \left(1 - \alpha \right),
\end{align*}
where $t_{\max, \Omega, T-1}$ denotes the distribution function of the maximal entry of a centered random vector with a multivariate $t$-distribution with scale matrix $\Omega$ and $(T-1)$ degrees of freedom. $\widehat{\Omega}_i (g)$ is an estimator of the correlation matrix of the moment inequalities, $\rho$ is a regularization function, and $\epsilon_N$ is a regularization parameter. The procedures for obtaining $\widehat{\Omega}_i (g)$, $\rho$, and $\epsilon_N$ are discussed in Appendix \ref{appendix: confsetgroup}.

The marginal confidence set for unit $i$ is obtained by inverting a test for group membership.
\begin{align*}
\mathrm{CS}_{i, \alpha} =& \left\{ g \in \mathbb{G} : \widehat{T}_i(g) \leq \hat{c}_{i, \alpha} (g)  \right\} \cup \left\{\hat{g}_i\right\}.
\end{align*}
$\hat{g}_i$ is added to the confidence set to ensure that the set is not empty.

All necessary components for the proposed procedures have now been collected. Specifically, the confidence sets for $\beta_i$ can be constructed by combining $\hat{\theta}_{g}$, $\widehat{\Sigma}_g$, and $\mathrm{CS}_{i, \alpha} $. The following section discusses the proposed procedures.

\section{Inference methods for unit-specific coefficients}\label{sec:method}

In this section, we present inference methods for unit-specific coefficients $\beta_i$. We combine the group-level slope estimates and their standard errors with the confidence set for group membership to construct confidence sets for each unit’s coefficient. Specifically, let $1-\alpha$ be our intended coverage probability. We divide $\alpha$ into two parts: $\alpha = \alpha_1 +\alpha_2$. The confidence level for the group assignment is set to $1 - \alpha_1$. For each group membership assignment, we consider a confidence set for the slope coefficients with coverage of $1- \alpha_2$. Combining these two achieves the intended coverage $1-\alpha$. Two procedures are proposed for this purpose. Both methods address the bias in \eqref{eqn:decompose}. The minimum-type approach avoids it by considering various possible groups within which the test statistic is evaluated, and the bias-correction method keeps the estimated group but expands the confidence set to include the coefficient values in the other possible groups.

\subsection{Confidence set based on minimum-type statistics}\label{subsec:method_um}

We first propose a procedure based on the minimum of the $t$- or Wald test statistics across groups in the confidence set for group membership. The idea is that, under the true value of the coefficient, $\beta_i^0=\theta_{g_i^0}^0$, the test statistic constructed from the group-level estimate, $\hat{\theta}_g$, should be small for some group $g$. The problem is that relying on $\hat{\theta}_{\hat{g}_i}$ incurs bias due to uncertainty in the point estimate $\hat{g}_i$, as shown in \eqref{eqn:decompose}. Therefore, it is natural to base inference on $\hat{\theta}_g$ for $g\in\mathrm{CS}_{i, \alpha_1}$, instead of relying solely on $\hat{\theta}_{\hat{g}_i}$, to avoid the bias identified in \eqref{eqn:decompose}. 

Specifically, our proposal is to take the minimum of the (absolute) test statistics based on $\hat{\theta}_g$ over $g\in\mathrm{CS}_{i, \alpha_1}$, and then compare it with a standard critical value. The confidence set follows upon inverting the minimum test statistic. Suppose we want to construct a confidence set for $R\beta_i^0$ for some $r\times p$ matrix $R$ ($r\leq p$). Construction is discussed separately for $r=1$ and $r \geq 2$.

When $r=1$, the minimum absolute $t$-statistic is defined as
\begin{align}
    \tilde{t}_{i,NT, \alpha_1}(\theta) \coloneqq \min_{g\in \mathrm{CS}_{i, \alpha_1}}\left|\frac{R\left(\hat{\theta}_g - \theta\right)}{\mathrm{se}(g)}\right|,
\end{align}
where $\operatorname{se}(g)\coloneqq\sqrt{R\widehat{\Sigma}_gR' / \widehat{N}_{g} T}$. The confidence interval based on $\tilde{t}_{i,NT, \alpha_1}(\theta)$ is defined as
\begin{align}
    \mathrm{CI}^{\mathrm{MT}}(i,\alpha) \coloneqq \left\{R\theta: \tilde{t}_{i,NT,\alpha_1}(\theta) \leq z_{1-\alpha_2/2}\right\},
    \label{def:umt_t}
\end{align}
with $z_{\alpha}$ the $\alpha$-quantile of $N(0,1)$. 

Next, we consider inference for multiple parameters or multiple linear combinations of parameters, that is, when $r \geq 2$. In this case, we consider the Wald statistic, $W_{i,NT}(\theta, \hat{g}_i,\hat{g}_i)$, where 
\begin{align}
W_{i, N T}(\theta, g, f)\coloneqq\widehat{N}_{g} T\left(\hat{\theta}_{f}-\theta\right)^{\prime} R'\left(R\widehat{\Sigma}_{g}R'\right)^{-1}R\left(\hat{\theta}_{f}-\theta\right).   
\end{align}
The second argument $g$ of $W_{i,NT}$ indicates the group used for scaling $\widehat{N}_{g}$ and $\widehat{\Sigma}_{g}$, and the third argument $f$ refers to the group for the coefficient vector.
As in the case with $r=1$, $W_{i,NT}(\theta, \hat{g}_i,\hat{g}_i)$ suffers from bias due to uncertainty in $\hat{g}_i$, albeit in a nonlinear fashion because of its quadratic form.
By the decomposition, $(\hat{\theta}_{g_i^0} - \theta^0_{g_i^0} ) + (\hat{\theta}_{\hat{g}_i} - \hat{\theta}_{g_i^0})$, in \eqref{eqn:decompose}, it holds that 
\begin{align}
\widehat{N}_{g_{i}^{0}}& T\left(\hat{\theta}_{\hat{g}_{i}}-\theta_{g_{i}^{0}}^{0}\right)^{\prime} R'\left(R\widehat{\Sigma}_{g_i^0}R'\right)^{-1}R\left(\hat{\theta}_{\hat{g}_{i}}-\theta_{g_{i}^{0}}^{0}\right) \\
=& 
\widehat{N}_{g_{i}^{0}} T\left(\hat{\theta}_{g_{i}^{0}}-\theta_{g_{i}^{0}}^{0}\right)^{\prime} R'\left(R\widehat{\Sigma}_{g_i^0}R'\right)^{-1}R\left(\hat{\theta}_{g_{i}^{0}}-\theta_{g_{i}^{0}}^{0}\right) \\
&+
\underbrace{\widehat{N}_{g_{i}^{0}} T\left(\hat{\theta}_{\hat{g}_{i}}-\hat{\theta}_{g_{i}^{0}}\right)^{\prime} R'\left(R\widehat{\Sigma}_{g_i^0}R'\right)^{-1}R\left(\hat{\theta}_{\hat{g}_{i}}-\hat{\theta}_{g_{i}^{0}}\right)}_{\mathrm{Bias} \ 1} \\
& + \underbrace{2 \widehat{N}_{g_{i}^{0}} T\left(\hat{\theta}_{\hat{g}_{i}}-\hat{\theta}_{g_{i}^{0}}\right)^{\prime} R'\left(R\widehat{\Sigma}_{g_i^0}R'\right)^{-1}R\left(\hat{\theta}_{g_{i}^{0}}-\theta_{g_{i}^{0}}^{0}\right)}_{\mathrm{Bias} \ 2} .
\label{eqn:decompose_2}
\end{align}
We observe two sources of bias arising even when $\widehat{N}_{g_i^0}$ and $\widehat{\Sigma}_{g_i^0}$ are known.
The first term represents the squared bias, which is necessarily positive. The second term is the cross-product between the bias and the estimation error, and its sign may be either positive or negative depending on the data. Furthermore, since $W_{i,NT}(\theta, \hat{g}_i,\hat{g}_i)$ uses $\widehat{N}_{\hat{g}_i}$ and $\widehat{\Sigma}_{\hat{g}_i}$ for scaling, the possible misclassification also affects inference through this channel.

The above bias can be avoided by using the minimum-Wald statistic:
\begin{align*}
 \widetilde{W}_{i, N T, \alpha_1}(\theta)\coloneqq \operatorname{min}_{g \in \mathrm{CS}_{i, \alpha_1}} W_{i, N T}(\theta, g, g).   
\end{align*}
The confidence set is defined as
\begin{align}
   \mathrm{CS}^{\mathrm{MW}}(i, \alpha)\coloneqq\left\{R\theta: \theta \in \Theta, \widetilde{W}_{i, N T, \alpha_1}(\theta) \leq \chi_{p, 1-\alpha_2}^{2}\right\},
\end{align}
where $\chi_{p,\alpha}^2$ is the $\alpha$-quantile of the chi-square distribution with $p$ degrees of freedom.

To practically compute these confidence sets, note that $\mathrm{CI}^{\mathrm{MT}}(i,\alpha)$ and $\mathrm{CS}^{\mathrm{MW}}(i, \alpha)$ can be calculated as the union of the confidence intervals (or confidence sets) based on $\hat{\theta}_g$, where the union is taken across $g\in\mathrm{CS}_{i, \alpha_1}$. More specifically, our confidence interval or set can be written as:
\begin{align}
    \mathrm{CI}^{\mathrm{MT}}(i,\alpha) 
    &= \bigcup_{g\in\mathrm{CS}_{i, \alpha_1}}\left[R\hat{\theta}_g-z_{1-\alpha_2 /2} \times \mathrm{se}(g), \ R\hat{\theta}_g+z_{1-\alpha_2 /2} \times \mathrm{se}(g)\right],
    \intertext{and}
   \mathrm{CS}^{\mathrm{MW}}(i, \alpha) 
    &= \bigcup_{g\in\mathrm{CS}_{i, \alpha_1}}\left\{R\theta: W_{i, N T}(\theta, g, g) \leq \chi_{p, 1-\alpha_2}^{2}\right\}.
\end{align}
This formulation follows the ideas of \textcite{Dufour1990,BergerBoos1994,Silvapulle1996}, which are early works on constructing valid confidence sets that combine over nuisance parameter values. 

The minimum-Wald approach may produce disjoint confidence sets.\footnote{We found through simulations that the minimum-type confidence interval is disjoint more frequently when $T$ is smaller and/or the error variance is larger.} This feature has both advantages and disadvantages. Compared with the bias-correction approach introduced below, which always yields connected intervals, this approach would yield smaller confidence sets and higher power. A drawback is that disjoint confidence intervals may be hard to interpret and explain. We illustrate these features in the empirical example.

\subsection{Confidence set based on bias correction}\label{subsec:method_bc}

Next, we introduce an alternative procedure that relies on bias correction. While the bias given in \eqref{eqn:decompose} and \eqref{eqn:decompose_2} cannot be directly estimated, we can bound its possible magnitude. We discuss cases with $r=1$ and $r\geq 2$ separately. The method for $r=1$ is for the scalar parameter and applies the upper and lower bounds of the bias to the confidence interval. The method for $r\geq 2$ corrects the quadratic forms of the bias and the cross-product term, corresponding to the formula in \eqref{eqn:decompose_2}.

When $r=1$, our $100(1-\alpha)\%$ bias-corrected confidence interval is defined as follows:
\begin{align}
  \mathrm{CI}^{\mathrm{BC}}(i, \alpha)
  &\coloneqq\left[R\hat{\theta}_{\hat{g}_{i}}+\widehat{\mathrm{LBias}}_i-z_{1-\alpha_2 /2} \times \mathrm{se}_{i}, \ R\hat{\theta}_{\hat{g}_{i}}+\widehat{\mathrm{UBias}}_i+z_{1-\alpha_2 /2} \times \mathrm{se}_{i}\right] \\
  &=\left[\min _{g \in \mathrm{CS}_{i, \alpha_1}}R\hat{\theta}_{g}-z_{1-\alpha_2 /2} \times \mathrm{se}_{i}, \ \max _{g \in \mathrm{CS}_{i, \alpha_1}}R\hat{\theta}_{g}+z_{1-\alpha_2 /2} \times \mathrm{se}_{i}\right],  
\end{align}
where
\begin{align}
\widehat{\mathrm{LBias}}_i\coloneqq & \min _{g \in \mathrm{CS}_{i, \alpha_1}}R\left(\hat{\theta}_{g}-\hat{\theta}_{\hat{g}_{i}}\right), \\
\widehat{\mathrm{UBias}}_i\coloneqq & \max _{g \in \mathrm{CS}_{i, \alpha_1}}R\left(\hat{\theta}_{g}-\hat{\theta}_{\hat{g}_{i}}\right), \\
\mathrm{se}_{i}\coloneqq & \max _{g \in \mathrm{CS}_{i, \alpha_1}} \operatorname{se}(g).    
\end{align}

The bias-corrected confidence interval is constructed by explicitly bounding the magnitude of the bias, which arises from statistical uncertainty in $\hat g_i$. By considering a confidence set for group membership, we bound the bias from below by $\widehat{\mathrm{LBias}}_i$ and from above by $\widehat{\mathrm{UBias}}_i$.

It is also important to correct the standard error. To illustrate this, consider the decomposition of the $t$-statistic:
\begin{align}
    \frac{\sqrt{\widehat{N}_{\hat{g}_{i}} T}\left(R\hat{\theta}_{\hat{g}_{i}}-R\theta_{g_i^0}^{0}\right)}{\sqrt{R\widehat{\Sigma}_{\hat{g}_{i}}R'}}=A_{i,NT}^* + B_{i,NT}^*,
\end{align}
where $A_{i, N T}^*=\sqrt{\widehat{N}_{\hat{g}_{i}} T}R\left(\hat{\theta}_{g_i^0}-\theta_{g_i^0}^{0}\right) / \sqrt{R\widehat{\Sigma}_{\hat{g}_{i}}R'}$ and $B_{i, N T}^*=\sqrt{\widehat{N}_{\hat{g}_{i}} T}R\left(\hat{\theta}_{\hat{g}_{i}}-\hat{\theta}_{g_i^0}\right) / \sqrt{R\widehat{\Sigma}_{\hat{g}_{i}}R'}$.
Notice that the denominator of $A_{i, N T}^*$ is $\sqrt{R\widehat{\Sigma}_{\hat{g}_{i}}R'}/ \sqrt{\widehat{N}_{\hat{g}_{i}} T}$, but this is not the standard error for the numerator $\hat{\theta}_{g_i^0}-\theta_{g_i^0}^{0}$ when $\hat g_i \neq g_i^0$. Therefore, it is necessary to account for possible group misassignment by using the maximum standard error across all groups in the confidence set for group membership.

We extend this approach to $r\geq2$. The bias formula given in \eqref{eqn:decompose_2} motivates the following bias-corrected Wald statistic:
\begin{align}
\overline{W}_{i, N T, \alpha_1}(\theta)\coloneqq& \min _{g \in \mathrm{CS}_{i, \alpha_1}} W_{i, N T}\left(\theta, g, \hat{g}_{i}\right)-\max _{g \in \mathrm{CS}_{i, \alpha_1}} \widehat{N}_{g} T\left(\hat{\theta}_{\hat{g}_{i}}-\hat{\theta}_{g}\right)^{\prime} R'\left(R\widehat{\Sigma}_{g}R'\right)^{-1}R\left(\hat{\theta}_{\hat{g}_{i}}-\hat{\theta}_{g}\right), \\
&-2 \max _{g \in \mathrm{CS}_{i, \alpha_1}}\left|\widehat{N}_{g} T\left(\hat{\theta}_{\hat{g}_{i}}-\hat{\theta}_{g}\right)^{\prime} R'\left(R\widehat{\Sigma}_{g}R'\right)^{-1}R\left(\hat{\theta}_{g}-\theta\right)\right|.
\label{def:bc_wald}
\end{align}
The bias-corrected Wald statistic, $\overline{W}_{i, N T, \alpha_1}(\theta)$, has three terms. The first, $\min _{g \in \mathrm{CS}_{i, \alpha_1}} W_{i, N T}(\theta, g, \hat{g}_{i})$, resembles the standard Wald statistic but may have bias if group membership is misclassified. We adjust by using variance-covariance matrices for each group in the confidence set (noting the minimization pertains only to $\widehat{N}_g$ and $\widehat{\Sigma}_{g}$). The second term is the maximum squared difference between the coefficients estimated for the assigned and true group. The third term is the cross-product between bias and estimation error; since its sign is indeterminate, we conservatively use its absolute value. This approach is robust, yet inherently conservative.  

A confidence set is constructed by inverting the bias-corrected Wald statistic:
\begin{align}
    \mathrm{CS}^{\mathrm{BC}}(i, \alpha)=\left\{R\theta: \theta\in \Theta, \overline{W}_{i, NT, \alpha_1}(\theta) \leq \chi_{p, 1-\alpha_2}^{2}\right\}.
\end{align}

Confidence intervals based on the bias-corrected $t$ statistic are always connected, unlike those from the minimum $t$-statistic.\footnote{Note that confidence sets based on the bias-corrected Wald statistic are not necessarily connected, whereas the bias-corrected confidence interval is always connected.} However, this connectedness leads to wider, more conservative intervals. It is also important to note that these two methods differ in the choice of the variance or standard errors in the test statistics. The minimum-type approach can utilize the standard error for each group. On the other hand, the bias correction method must use the maximum standard error, since it fixes the center at $\hat{\theta}_{\hat{g}_i}$. Thus, the bias-correction method has two sources of conservativeness. We need to consider trade-offs: the bias-correction method offers connected, interpretable intervals, while the minimum-type approach yields smaller, more powerful ones. These trade-offs are illustrated in an empirical example in Section \ref{sec:empirical} and in the simulations in Section \ref{sec:simu}.

\section{Asymptotic analysis}\label{sec:asym}

In this section, we study the theoretical properties of the confidence sets proposed in Section \ref{sec:method}. We first give a few assumptions and then provide the asymptotic justification for our methods.

\subsection{Assumptions}\label{subsec:asym_assumption}

We derive the asymptotic properties of our proposed inference methods under high-level assumptions about the estimators. These conditions are satisfied by the existing procedures under suitable conditions. We briefly discuss more primitive conditions, but the details can be found in the original papers. 

The first assumption is the asymptotic normality of the group-specific slope estimators.
\begin{asm}\label{asm:asym_normal}
    There exists a permutation $\sigma:\mathbb{G}\to\mathbb{G}$ such that $\sqrt{N_{g} T}\left(\hat{\theta}_{\sigma(g)}-\theta_{g}^{0}\right) \xrightarrow{d} N\left(0, \Sigma_g\right)$ for some $\Sigma_g>0$ for all $g \in\mathbb{G}$.
\end{asm}
The asymptotic normality of the group-specific slope estimators is established for a wide range of clustering methods \parencite{su_identifying_2016, wang2016homogeneity, mehrabani2022estimation}. \textcite{dzemskiokui2021convergence} show that Assumption \ref{asm:asym_normal} can hold even when some units are potentially misclassified---a case where our proposed methods are particularly relevant. Following the convention in the literature, we take $\sigma(g)=g$ for all $g\in\mathbb{G}$ by relabeling groups.

\begin{asm}\label{asm:est_group_size}
    $\widehat{N}_{g} / N_{g} \xrightarrow{p} 1$ for all $g$.
\end{asm}
Assumption \ref{asm:est_group_size} is typically implied by classification consistency; see, for example, Corollary 2.3 of \textcite{su_identifying_2016}. Furthermore, \textcite{dzemskiokui2021convergence} show that Assumption \ref{asm:est_group_size} holds even if some units are potentially misclassified; see the proof of Theorem 1 of \textcite{dzemskiokui2021convergence}.

The next assumption requires that the group-specific variance-covariance matrix estimators are consistent.
\begin{asm}\label{asm:vcov_consistent}
    $\widehat{\Sigma}_{g} \xrightarrow{p} \Sigma_g$ for all $g \in\mathbb{G}$.
\end{asm}
How to obtain $\widehat{\Sigma}_{g} $ satisfying Assumption \ref{asm:vcov_consistent} is discussed in Section \ref{sec:asym_var}. In our setting, group memberships may not be consistently estimated across all units. We extend the existing analyses to accommodate possible inconsistencies in group membership estimation and show that the usual group-specific asymptotic variance estimator is consistent. We also develop an asymptotic variance estimator under fixed-$T$ asymptotics. 

Lastly, we assume that the confidence sets for the group memberships have appropriate coverage.
\begin{asm}\label{asm:cs_membership}
    \begin{align}
        \liminf_{N,T\to\infty}\min_{i=1,\ldots,N}P\left(g_i^0\in \mathrm{CS}_{i, \alpha_1}\right)\geq 1-\alpha_1.
    \end{align}
\end{asm}
\textcite{dzemski2024confidence} extensively discuss sufficient conditions that ensure Assumption \ref{asm:cs_membership}; see Sections 3 and 4 of \textcite{dzemski2024confidence} for details.

\subsection{Asymptotics for the minimum-type confidence set}\label{subsec:asym_um}

We now establish the coverage properties of the confidence sets based on the minimum statistics.

\begin{thm}\label{thm:coverage_um}
    Suppose Assumptions \ref{asm:asym_normal}-\ref{asm:cs_membership} hold.
    
    (i) When $r=1$, we have
    \begin{align}
        \liminf_{N,T\to\infty}\min_{i=1,\ldots,N}P\left(R\theta_{g_i^0}^{0} \in \mathrm{CI}^{\mathrm{MT}}(i, \alpha)\right) \geq 1 - \alpha.
    \end{align}
    
    (ii) When $r\geq2$, we have
    \begin{align}
        \liminf_{N,T\to\infty}\min_{i=1,\ldots,N}P\left(R\theta_{g_{i}^{0}}^{0} \in\mathrm{CS}^{\mathrm{MW}}(i, \alpha)\right) \geq 1 - \alpha.
    \end{align}
\end{thm}

\noindent\begin{proof}[Proof of Theorem \ref{thm:coverage_um}]
    See Appendix \ref{appendix: proof}.
\end{proof}

\subsection{Asymptotics for the bias-corrected confidence set}\label{subsec:asym_bc}

The following theorem shows that the bias-corrected confidence sets attain nominal coverage.
\begin{thm}\label{thm:coverage_bc}
    Suppose Assumptions \ref{asm:asym_normal}-\ref{asm:cs_membership} hold.
    
    (i) When $r=1$, we have
    \begin{align}
        \liminf_{N,T\to\infty}\min_{i=1,\ldots,N}P\left(R\theta_{g_i^0}^{0} \in \mathrm{CI}^{\mathrm{BC}}(i, \alpha)\right) \geq 1 - \alpha.
    \end{align}
    
    (ii) When $r\geq2$, we have
    \begin{align}
        \liminf_{N,T\to\infty}\min_{i=1,\ldots,N}P\left(R\theta_{g_{i}^{0}}^{0} \in \mathrm{CS}^{\mathrm{BC}}(i, \alpha)\right) \geq 1 - \alpha.
    \end{align}
\end{thm}

\noindent\begin{proof}[Proof of Theorem \ref{thm:coverage_bc}]
    See Appendix \ref{appendix: proof}.
\end{proof}

\begin{rem}
    Theorems \ref{thm:coverage_um} and \ref{thm:coverage_bc} guarantee unit-wise (marginal) correct coverage by the bias-corrected and minimum-type confidence sets. The joint coverage across $N$ units can also be established under the assumption of joint coverage by the first step confidence set, $\mathrm{CS}_{i,\alpha}$, for group membership and a suitably adjusted nominal level for the second step confidence set construction; see Appendix \ref{appendix: joint_coverage} for details.
\end{rem}

\subsection{Length analysis}\label{subsec:asym:length}
In this subsection, we compare the length properties of our proposed method’s confidence interval with those of two existing alternatives: the unit-by-unit and naive approaches, defined below. Our minimum-type confidence interval is asymptotically shorter than the unit-by-unit approach. Although the naive procedure that ignores uncertainty in group membership estimation produces an even shorter interval, it is not guaranteed to be valid.

The first alternative, the unit-by-unit confidence interval, is constructed from the unit-wise time-series regression. Letting $\tilde{\beta}_i\coloneqq (\sum_{t=1}^Tx_{it}x_{it}')^{-1}\sum_{t=1}^Tx_{it}y_{it}$, the unit-by-unit confidence interval is defined as
\begin{align}
    \mathrm{CI}^{\mathrm{ubu}}(i,\alpha) \coloneqq \left[R\tilde{\beta}_i-z_{1-\alpha/2} \times \mathrm{se}_{i,\beta}, \ R\tilde{\beta}_i+z_{1-\alpha/2} \times \mathrm{se}_{i,\beta}\right],
\end{align}
where $\mathrm{se}_{i,\beta}$ is the standard error of $R\tilde{\beta}_i$. The second alternative is the clustering-based naive confidence interval, ignoring uncertainty in $\hat{g}_i$. This confidence interval, denoted by $\mathrm{CI}^{\mathrm{naive}}(i,\alpha)$, is constructed using $\hat{\theta}_{\hat{g}_i}$ and $\mathrm{se}(\hat{g}_i)$:
\begin{align}
    \mathrm{CI}^{\mathrm{naive}}(i,\alpha) \coloneqq \left[R\hat{\theta}_{\hat{g}_{i}}-z_{1-\alpha/2} \times \mathrm{se}(\hat{g}_i), \ R\hat{\theta}_{\hat{g}_{i}}+z_{1-\alpha/2} \times \mathrm{se}(\hat{g}_i)\right].
\end{align}

We study the length properties of four confidence intervals in total: the unit-by-unit, naive, bias-corrected, and minimum-$t$-based confidence intervals.
Under the standard assumptions, $\tilde{\beta}_i$ is $\sqrt{T}$-consistent, which implies that $|\mathrm{CI}^{\mathrm{ubu}}(i,\alpha)|^{-1} = (2z_{1-\alpha/2}\mathrm{se}_{i,\beta})^{-1} = O_p(\sqrt{T})$, where $|\mathrm{CI}(i,\alpha)|$ denotes the length of $\mathrm{CI}(i,\alpha)$ excluding the gaps between disjoint intervals (namely, it is the Lebesgue measure of $\mathrm{CI}(i,\alpha)$). In contrast, the length of the naive confidence interval is $|\mathrm{CI}^{\mathrm{naive}}(i,\alpha)|=2z_{1-\alpha/2}\mathrm{se}(\hat{g}_i)=O_p(1/\sqrt{NT})$ under Assumption \ref{asm:asym_normal}. The length of the minimum-$t$-based confidence interval, $\mathrm{CI}^{\mathrm{MT}}(i,\alpha)$, is also of this order, since it is a finite union of the group-specific confidence intervals. Hence, we have 
\begin{align}
    \frac{|\mathrm{CI}^{\mathrm{naive}}(i,\alpha)| + |\mathrm{CI}^{\mathrm{MT}}(i,\alpha)|}{|\mathrm{CI}^{\mathrm{ubu}}(i,\alpha)|} = O_p(1/\sqrt{N}) \xrightarrow{p} 0.
\end{align}
Note that $|\mathrm{CI}^{\mathrm{naive}}(i,\alpha)| \leq |\mathrm{CI}^{\mathrm{MT}}(i,\alpha)|$, because the naive procedure uses $\alpha$ for constructing the interval, while the minimum-type approach uses $\alpha_2<\alpha$ for each interval. We thus obtain $|\mathrm{CI}^{\mathrm{naive}}(i,\alpha)| \leq |\mathrm{CI}^{\mathrm{MT}}(i,\alpha)| \ll |\mathrm{CI}^{\mathrm{ubu}}(i,\alpha)|$.

The length of the bias-corrected confidence interval, $\mathrm{CI}^{\mathrm{BC}}(i,\alpha)$, is harder to study. It should be noticed that its length is $O_p(1)$ without additional assumptions. To see this, we write $|\mathrm{CI}^{\mathrm{BC}}(i,\alpha)| = 2z_{1-\alpha_2/2}\mathrm{se}_i + \max_{g,f\in\mathrm{CS}_{i,\alpha_1}}(\hat{\theta}_g - \hat{\theta}_f)$. Whereas the first term is $O_p(1/\sqrt{NT})$, the second term is $O_p(1)$ in general, as $\mathrm{CS}_{i,\alpha_1}$ may not be a singleton. That said, we recall that the advantages of the bias-corrected confidence interval are its connectedness and interpretability. We may establish conditions under which $\max_{g,f\in\mathrm{CS}_{i,\alpha_1}}(\hat{\theta}_g - \hat{\theta}_f)=O_p(1/g(N,T))$ for some $g(N,T)\to\infty$, although doing so requires analyzing the power of the confidence set $\mathrm{CS}_{i,\alpha_1}$ for the group membership. We leave this analysis as a topic for future research.

\section{Asymptotic variance estimation}\label{sec:asym_var}

In this section, we consider the estimation of the asymptotic variance of the group-specific coefficients. Namely, we construct asymptotic variance estimators, $\hat{\Sigma}_g$, that satisfy Assumption \ref{asm:vcov_consistent}. Since we consider cases where the group memberships of some units are not consistently estimated, we need to extend the existing asymptotic analyses to accommodate such situations. We first consider the usual asymptotic variance estimator and show its consistency even when some units are misclassified. We also develop an asymptotic variance estimator that is valid when $T$ is fixed and only $N$ tends to infinity. This development may be of independent interest.

\subsection{The cluster robust asymptotic variance estimator}

We first consider the usual cluster robust asymptotic variance estimator and establish its consistency. The estimand is $\Sigma_g$, which in our setting is usually the asymptotic variance-covariance matrix of the within-group OLS estimator, i.e., 
\begin{align*}
    \Sigma_g=M_g^{-1}\Omega_gM_g^{-1}
\end{align*}
where
\begin{align*}
    M_g=&\lim_{N,T\to\infty}\frac{1}{N_gT}\sum_{i=1}^N\sum_{t=1}^T1\{g_i^0=g\}E[x_{it}x_{it}'] \\
    \Omega_g=&\lim_{N,T\to\infty} \frac{1}{N_gT}\sum_{i=1}^N\sum_{j=1}^N\sum_{t=1}^T\sum_{s=1}^T1\{g_i^0=g\}1\{g_j^0=g\}E[x_{it}x_{js}'\varepsilon_{it}\varepsilon_{js}].
\end{align*}
We note that $M_g$ can be consistently estimated by $\widehat{M}_g \coloneqq (\widehat{N}_gT)^{-1}\sum_{i=1}^N\sum_{t=1}^T1\{\hat{g}_i=g\}x_{it}x_{it}'$. 
For $\Omega_g$, we often assume the independence across $i$, so that $\Omega_g$ is simplified and does not include the terms with $i\neq j$. Under such an assumption, a typical choice is 
\citeauthor{arellano1987computing}'s (\citeyear{arellano1987computing}) cluster robust variance estimator: 
\begin{align}
    \widehat{\Omega}_{g} \coloneqq \frac{1}{\widehat{N}_gT}\sum_{i=1}^N\sum_{t=1}^T\sum_{s=1}^T1\{\hat{g}_i=g\}x_{it}x_{is}'\hat{\varepsilon}_{it}\hat{\varepsilon}_{is},
\end{align}
where $\hat{\varepsilon}_{it} \coloneqq y_{it}-x_{it}'\hat{\theta}_{\hat{g}_i}$. 

The next lemma shows that, even if some units are misclassified, this estimator is still consistent under mild conditions.

\begin{lem}\label{lem:vcov_consistency}
    Suppose that Assumptions \ref{asm:asym_normal} and \ref{asm:est_group_size} hold, and that
    \begin{itemize}
        \item[(a)] $\operatorname{liminf}_{N \rightarrow \infty}\left(N_{g} / N\right)>0$ for all $g \in \mathbb{G}$,

        \item[(b)] $\sum_{i=1}^{N} 1\left\{\hat{g}_{i}=g, g_{i}^{0} \neq g\right\}=o_p\left(N / T^{1+q}\right)$ for all $g \in \mathbb{G}$ for some $q>0$,

        \item[(c)] $\sup_{i,t} E[\|x_{it} \varepsilon_{it}\|^{2(1+q) / q+\varepsilon}] + \sup_{i,t}E[\|x_{it}\|^{2(1+q) / q}]<\infty$ for some $\varepsilon>0$,

        \item[(d)] $\{ (x_{it}, \varepsilon_{it}) \}_{t=1}^T$ is independent across $i$,

        \item[(e)] $(x_{it}, \varepsilon_{it})$ for each $i$ is a strong mixing sequence with mixing coefficients bounded uniformly across $i$ by size $(1-c) r /(r-c)$ for some $c \in 2 \mathbb{N}$, $c \geq 2(1+q) / q$, and $r>c$,

        \item[(f)] $\widetilde{\Omega}_g \xrightarrow{p} \Omega_g$ for all $g \in \mathbb{G}$, where
        \begin{align}
            \widetilde{\Omega}_{g}\coloneqq\frac{1}{N_{g} T} \sum_{i=1}^{N} \sum_{t=1}^{T} \sum_{s=1}^{T} 1\left\{g_{i}^{0}=g\right\} x_{it} x_{is}^{\prime} \tilde{\varepsilon}_{it} \tilde{\varepsilon}_{is},
        \end{align}
        and $\tilde{\varepsilon}_{it}\coloneqq y_{it}-x_{it}^{\prime} \hat{\theta}_{g_i^0}$.
    \end{itemize}
    
 Then we have $\widehat{M}_g\xrightarrow{p} M_g$, and $\widehat{\Omega}_{g} \xrightarrow{p} \Omega_{g}$ for all $g \in \mathbb{G}$.
\end{lem}

\noindent\begin{proof}[Proof of Lemma \ref{lem:vcov_consistency}]
    See Appendix \ref{appendix: proof}.
\end{proof}

Lemma \ref{lem:vcov_consistency} shows that a consistent estimation of the variance-covariance matrices is possible even if some units are misclassified. 

\begin{rem}
    \textcite{dzemskiokui2021convergence} establish the asymptotic normality of the KMeans estimator in the presence of potential misclassification, under a condition that implies $\sum_{i=1}^{N} 1\left\{\hat{g}_{i}=g, g_{i}^{0} \neq g\right\}=o_p\left(N / T\right)$. Proving the consistency of the variance-covariance matrix estimators requires a bit stronger condition stated in (b). Assuming condition (d), \textcite{hansen2007asymptotic} shows that condition (f) holds if $(x_{it},\varepsilon_{it})$ is weakly serially dependent; see Theorem 3 of \textcite{hansen2007asymptotic}.
\end{rem}

We may use other estimators for $\Omega_g$ proposed in the literature. 
An alternative approach is to use \citeauthor{driscoll1998consistent}'s (\citeyear{driscoll1998consistent}) kernel-based estimator, which is robust to both cross-sectional and temporal dependence. We note that, if classification based on $\{\hat{g}_i\}_{i=1}^N$ is consistent, the consistency of these estimators can be established following the same lines as \textcite{hansen2007asymptotic} and \textcite{driscoll1998consistent}. Formally establishing their theoretical properties under possible misclassification is beyond the scope of the current paper.

\subsection{Fixed-$T$ adjustment for variance-covariance matrices}\label{sec:fixed_T}

We next consider the asymptotic variance estimator under fixed-$T$ framework. The previous double asymptotic framework used there might be restrictive in that $T$ is assumed to diverge, so that misclassification is nullified in variance-covariance estimation. In contrast, if $T$ is short, then the variance-covariance estimation needs to incorporate possible misclassification, as pointed out by \textcite{bonhomme2015grouped}. In this section, we follow \textcite{pollard1981strong}, \textcite{pollard1982central} and \textcite[][S.2.2]{bonhomme2015grouped} to construct short-$T$ variance-covariance matrix estimators, assuming $T$ is fixed as $N\to\infty$.

For the purpose of exposition, define $\boldsymbol{\theta}\coloneqq\left(\theta_{1}^{\prime}, \ldots, \theta_{G}^{\prime}\right)^{\prime}$, $\hat{\boldsymbol{\theta}}\coloneqq\left(\hat{\theta}_{1}^{\prime}, \ldots, \hat{\theta}_{G}^{\prime}\right)^{\prime}, y_{i}\coloneqq\left(y_{i 1}, \ldots, y_{i T}\right)^{\prime}$, and $x_{i}\coloneqq\left(x_{i 1}, \ldots, x_{i T}\right)^{\prime}$. 
Let $\hat{g}_{i}(\boldsymbol{\theta})=\argmin_{g \in \mathbb{G}}\left\|y_{i}-x_{i} \theta_{g}\right\|^{2}$ denote the optimal group assignment for unit $i$ conditional on $\boldsymbol{\theta}$. We allow the specified number of groups, $G$, to not coincide with the true number of groups, $G^0$, as long as $G$ is no larger than $G^0$. Following \textcite{pollard1981strong} and \textcite[][S.2.2]{bonhomme2015grouped}, we assume the following condition.
\begin{asm}
    \begin{itemize}
    \item[(a)] $\theta_{g} \in \Theta$ for all $g \in \mathbb{G}$, where $\Theta$ is compact.

    \item[(b)] $\left(y_{i}, x_{i}\right)$ are i.i.d. across $i$.

    \item[(c)] $E\left[\left\|y_{i}\right\|^{2}\right]+E\left[\left\|x_{i}\right\|^{2}\right]<\infty$.

    \item[(d)] $\max _{i} \lambda_{\max }\left(x_{i}^{\prime} x_{i}\right) \leq M<\infty$ for some $M>0$, where $\lambda_{\max}(\cdot)$ is the maximum eigenvalue.

    \item[(e)] The solution
    \begin{align}
        \bar{\boldsymbol{\theta}}=\underset{\boldsymbol{\theta}}{\argmin} E\left[\left\|y_{i}-x_{i} \theta_{\hat{g}_{i}(\boldsymbol{\theta})}\right\|^{2}\right]
    \end{align}
    is unique up to relabeling of groups for $G \leq G^{0}$.

    \item[(f)] $\left(y_{i}, x_{i}\right)$ has a continuous density, and $y_{i}$ has a continuous density given $x_{i}$.
\end{itemize}
\label{asm:fixed_T}
\end{asm}

Adapting the argument by \textcite{pollard1981strong}, it can be shown that under Assumptions \ref{asm:fixed_T} (a)--(e), $d_{H}(\hat{\boldsymbol{\theta}}, \bar{\boldsymbol{\theta}}) \xrightarrow{p} 0$, where $d_{H}$ is the Hausdorff distance. Furthermore, $\bar{\boldsymbol{\theta}}$ solves the following moment condition:
\begin{align}
    E\left[1\left\{\hat{g}_{i}(\bar{\boldsymbol{\theta}})=g\right\} x_{i}^{\prime}\left(y_{i}-x_{i} \bar{\theta}_g\right)\right]=0 \text{, for all } g \in \mathbb{G}.
    \label{eqn:moment_condition}
\end{align}
In this sense, $\bar{\boldsymbol{\theta}}$ may be viewed as the pseudo-true slope parameter that is identified based on the fixed-$T$ moment restrictions, \eqref{eqn:moment_condition}.

Following the argument by \textcite{bonhomme2015grouped}, we then have, under Assumption \ref{asm:fixed_T}, as $N\to \infty$ while $T$ is fixed, 
\begin{align}
    \sqrt{N}(\hat{\boldsymbol{\theta}}-\bar{\boldsymbol{\theta}}) \xrightarrow{d} N\left(0, \Gamma^{-1} V \Gamma^{-1}\right),
\end{align}
where
\begin{align}
    V=E\left[W_i(\bar{\boldsymbol{\theta}})\left(y_{i}-x_i \theta_{\hat{g}_{i}(\bar{\boldsymbol{\theta}})}\right)\left(y_{i}-x_i \theta_{\hat{g}_{i}(\bar{\boldsymbol{\theta}})}\right)^{\prime} W_i(\bar{\boldsymbol{\theta}})^{\prime}\right],
\end{align}
$W_i(\boldsymbol{\theta})=e_{\hat{g}_{i}(\boldsymbol{\theta})} \otimes x_{i}^{\prime}$, $e_g$ is the $G\times 1$ vector whose $g$-th element is 1 and 0's elsewhere, and $\Gamma=\left(\Gamma_{g\widetilde{g}} \right) \in \mathbb{R}^{p G \times p G}$ with
\begin{align}
    \Gamma_{g \widetilde{g}}=-\left.\frac{\partial}{\partial \theta_{\widetilde{g}}^{\prime}}\right|_{\boldsymbol{\theta}=\overline{\boldsymbol{\theta}}} E\left[1\left\{\hat{g}_{i}(\boldsymbol{\theta})=g\right\} x_{i}^{\prime}\left(y_{i}-x_{i} \theta_g\right)\right], \  g, \tilde{g} \in \mathbb{G}.
\end{align}

The following representation of $\Gamma$ is useful in the estimation.
\begin{prop}
Suppose Assumption \ref{asm:fixed_T} holds. Then we have
    \begin{align}
    &\Gamma_{gg} = E\left[x_{i}^{\prime} x_{i} 1\left\{\hat{g}_{i}(\bar{\boldsymbol{\theta}})=g\right\}\right] -E\left[\sum_{h \neq g} \int_{\overline{S}_{g h}} \frac{x_{i}^{\prime}\left(y-x_{i} \bar{\theta}_{g}\right)\left(y-x_{i} \bar{\theta}_{g}\right)^{\prime} x_{i}}{\left\|x_{i}\left(\bar{\theta}_{h}-\bar{\theta}_{g}\right)\right\|} f\left(y \mid x_{i}\right) d y\right]
    \intertext{and}
    &\Gamma_{g\widetilde{g}} 
    = E\left[\int_{\overline{S}_{g \widetilde{g}}} \frac{x_{i}^{\prime}\left(y-x_{i} \bar{\theta}_{g}\right)\left(y-x_{i} \bar{\theta}_{\widetilde{g}}\right)^{\prime} x_{i}}{\left\|x_{i}\left(\bar{\theta}_{\widetilde{g}}-\bar{\theta}_{g}\right) \right\|} f\left(y\mid x_{i}\right) d y\right] \text{ for } g \neq \tilde{g},
\end{align}
where
\begin{align}
     S_{g h}=\left\{y \in \mathbb{R}^{T} ;\left\|y-x \theta_{g}\right\|^{2}=\left\|y-x \theta_{h}\right\|^{2}, \text{ and } \left\|y-x \theta_{g}\right\|^{2} \leq\left\|y-x \theta_{\widetilde{g}}\right\|^{2} \text{ for all } \tilde{g} \neq g, h\right\},
\end{align}
and $\overline{S}_{gh}$ is $S_{gh}$ evaluated at $\boldsymbol{\theta}=\overline{\boldsymbol{\theta}}$.
     \label{prop:fixed_T}
\end{prop}

\noindent\begin{proof}[Proof of Proposition \ref{prop:fixed_T}]
    See Appendix \ref{appendix: proof}.
\end{proof}

To estimate $V$ and $\Gamma$, we follow \textcite{bonhomme2015grouped}. Specifically, we use the following estimators:
\begin{align}
    &\hat{V} \coloneqq \frac{1}{N}\sum_{i=1}^NW_i\left(\hat{\boldsymbol{\theta}}\right)\hat{\varepsilon}_{i,\widehat{g}_i(\widehat{\boldsymbol{\theta}})}\hat{\varepsilon}_{i,\widehat{g}_i(\widehat{\boldsymbol{\theta}})}'W_i\left(\hat{\boldsymbol{\theta}}\right)', \\
    &\hat{\Gamma}_{gg} \coloneqq \frac{1}{N}\sum_{i=1}^Nx_i'x_i1\{\hat{g}_i(\hat{\boldsymbol{\theta}})=g\} - \frac{1}{N}\sum_{i=1}^N\sum_{h\neq g} \hat{\Delta}_{igh}(\eta_N) \frac{x_i'\hat{\varepsilon}_{i,g}\hat{\varepsilon}_{i,g}'x_i}{\|x_i(\hat{\theta}_h - \hat{\theta}_g)\|}, \\
    \intertext{and}
    &\hat{\Gamma}_{g\widetilde{g}} \coloneqq \frac{1}{N}\sum_{i=1}^N\hat{\Delta}_{ig\widetilde{g}}(\eta_N) \frac{x_i'\hat{\varepsilon}_{i,g}\hat{\varepsilon}_{i, \widetilde{g}}'x_i}{\|x_i(\hat{\theta}_{\widetilde{g}} - \hat{\theta}_g)\|},
\end{align}
where $\hat{\varepsilon}_{i, g} \coloneqq y_i - x_i\hat{\theta}_g$, and
\begin{align}
    &\hat{\Delta}_{igh}(\eta_N) \coloneqq \frac{1}{\eta_N}\lambda \left(\eta_N^{-1}\frac{(\hat{\theta}_h-\hat{\theta}_g)'x'}{\|x(\hat{\theta}_h - \hat{\theta}_g)\|} \ \frac{\hat{\varepsilon}_{i,h}+\hat{\varepsilon}_{i,g}}{2}\right)  \\
    &\hspace{2.5cm} \times 1\left\{\max\left(\left\|\hat{\varepsilon}_{i,g}\right\|^2,\left\|\hat{\varepsilon}_{i,h}\right\|^2\right) \leq \min_{\tilde{h}\neq g,h}\left\|\hat{\varepsilon}_{i,\widetilde{h}}\right\|^2\right\}.
\end{align}
$\lambda(\cdot)$ is a kernel function, and $\eta_{N}$ is the bandwidth parameter such that $\eta_{N} \rightarrow 0$ and $\sqrt{N}{\eta_{N}} \rightarrow \infty$. We use the Gaussian kernel, $\lambda=\phi$, and
\begin{align}
    \eta_N = 1.06 \min_{g,h\neq g}\left(\sqrt{\widehat{\mathrm{Var}}\left(\frac{(\hat{\theta}_h-\hat{\theta}_g)^{'}x'}{\|x(\hat{\theta}_h - \hat{\theta}_g)\|} \ \frac{\hat{\varepsilon}_{i,h}+\hat{\varepsilon}_{i,g}}{2}\right)}\right)N^{-1/5}.
\end{align}

We then obtain $\hat \Sigma_g$ by taking the $g$-th block of $(\hat N_g / N) T \hat{\Gamma}^{-1} \hat{V}\hat{\Gamma}^{-1} $. This asymptotic variance estimator satisfies Assumption \ref{asm:vcov_consistent} when we take both $N, T\to \infty$. Note that the fixed-$T$ asymptotic distribution is centered around $\overline{\boldsymbol{\theta}}$, which is not the same as the true value of $\boldsymbol{\theta}$. As $T\to \infty$, they tend to be identical. Because of this property, this estimator satisfies Assumption \ref{asm:vcov_consistent} only under $T \to \infty$. Nonetheless, this fixed-$T$ version is expected to perform better than the usual cluster-robust variance estimator because it can accommodate misclassification of group memberships.

\section{Empirical applications}
\label{sec:empirical}

As an empirical illustration, we investigate the impact of minimum wage increases across U.S. states. The impact of the minimum wage has been documented as heterogeneous, so it is interesting to estimate the effect for each state, given that each state can, to some extent, set its own minimum wage policy. Our analysis builds on the foundational work of \textcite{dube2010minimum} and uses similar data sources. We adopt the group structure framework explored by \textcite{wang2019heterogeneous}, and compute confidence sets for group memberships following \textcite{dzemski2024confidence}. Importantly, our approach extends \textcite{dzemski2024confidence} by enabling inference on the effect of minimum wage changes at the state level.

Our empirical setting closely follows \textcite{dzemski2024confidence}, using the same approaches for coefficient estimation and confidence set construction for group memberships. However, our interest is in the unit-specific coefficients. We also use fixed-$T$ standard errors. The dataset, described in \textcite{dube2010minimum}, comprises quarterly observations for 1,380 U.S. counties from 1990Q1 to 2006Q2. Consistent with \textcite{wang2019heterogeneous}, we set the number of groups to four, as determined by the information criterion of \textcite{su_identifying_2016}. Coefficient estimates are obtained using the CLasso method \parencite{su_identifying_2016}. The estimated model is as follows:
\begin{align*}
    \log (\texttt{emp}_{ict}) = \theta_{g^0_i, 1} \log (\texttt{mw}_{ict}) + \theta_{g^0_i, 2} \log (\texttt{pop}_{ict}) + \theta_{g^0_i, 3} \log (\texttt{emp}_{ict}^{\texttt{TOT}}) + \phi_c + \tau_t + \sigma_i v_{ict}
\end{align*}
where $i = 1, \dotsc, 51$ indexes states, $c = 1, \dotsc, n_i$ indexes counties within state $i$, and $t = 1, \dotsc, T = 66$ indexes time periods, with $n_i$ denoting the number of counties in state $i$. Here, $\texttt{emp}_{ict}$ denotes employment in the restaurant sector, $\texttt{mw}_{ict}$ is the minimum wage, and $\texttt{emp}_{ict}^{\text{TOT}}$ represents total employment across all sectors. $\phi_c$ and $\tau_t$ capture county and time fixed effects, respectively, and $\sigma_i v_{ict}$ is an idiosyncratic error term.
Table \ref{tab:app:slope_coef} presents the estimated coefficients and associated fixed-$T$ standard errors. Our primary interest lies in $\theta_{g,1}$, which measures the effect of minimum wage changes for each group $g$. We subsequently compute the confidence sets for group memberships. Details of the computational procedure are provided in Appendix \ref{app:detail_empirical}. Table \ref{tab:empirics_cs_categories} summarizes the point estimate and confidence set for group membership for each state. Notably, the confidence sets and point estimates for group memberships allow states to be classified into 11 categories as in Table \ref{tab:empirics_cs_categories}, with states within a given category sharing the same confidence interval.

\begin{table}
\centering

\begin{tabular}{lcccc}
\toprule
 & \multicolumn{4}{c}{Groups} \\
\cmidrule(lr){2-5}
Parameter & $g=1$ & $g=2$ & $g=3$ & $g=4$ \\
\midrule
$\theta_{1}$ & 0.55 & -0.03 & 0.06 & -0.25 \\
 & (0.14) & (0.06) & (0.04) & (0.04) \\[0.5em]
$\theta_{2}$ & 0.63 & 0.60 & 0.34 & 0.47 \\
 & (0.47) & (0.34) & (0.08) & (0.11) \\[0.5em]
$\theta_{3}$ & 0.51 & 0.61 & 0.41 & 0.53 \\
 & (0.39) & (0.34) & (0.04) & (0.09) \\
\bottomrule
\end{tabular}

\caption{\label{tab:app:slope_coef}Estimates for the group-specific slope coefficients and fixed-$T$ standard errors. The coefficient estimates are derived from the replication of \textcite[][Table 1]{dzemski2024confidence}. The standard errors are the fixed-$T$ version.}
\end{table}

\begin{table}
	\centering
	\begin{threeparttable}
		\renewcommand{\arraystretch}{1.4}	\begin{tabular}{ccc}
			\hline $\hat{g}_i$ & $\mathrm{CS}_{i,0.05}$ &states \\
            \hline 1 & $\{1\}$ & Alabama, Georgia, Louisiana, Mississippi, Ohio, South Carolina, Texas \\
            \hline\multirow{3}{*}{2} & $\{2\}$ & \makecell{Arkansas, Maine, Maryland, Michigan, Minnesota, Missouri, \\ New York, North Carolina, Rhode Island, Tennessee} \\
            & $\{2,3\}$ & Delaware, Massachusetts,  New Jersey, Virginia \\
            & $\{1,2,3,4\}$ & District of Columbia, Hawaii, Nevada \\
            \hline\multirow{4}{*}{3} & $\{3\}$ & California, Illinois, Indiana, Pennsylvania, West Virginia \\
            & $\{2,3\}$ & Connecticut,Oklahoma, Wisconsin \\
            & $\{3,4\}$ & Idaho, Kentucky, New Hampshire \\
            & $\{2,3,4\}$ & Alaska, Arizona, Montana, New Mexico, Wyoming \\
            \hline\multirow{3}{*}{4} & $\{4\}$ & Florida,  Iowa, Kansas, North Dakota, Oregon, Vermont, Washington \\
            & $\{3,4\}$ & Colorado, Utah \\
            & $\{2,3,4\}$ & Nebraska, South Dakota \\
             \hline 
		\end{tabular}
        \caption{\label{tab:empirics_cs_categories}Categories based on point estimates and confidence sets for group membership. This table is based on the replication of \textcite[][B.1]{dzemski2024confidence} and formatted by the authors.}
	\end{threeparttable} 
\end{table} 

The estimated effect of the minimum wage increase exhibits heterogeneity; see $\theta_{g,1} $ in Table \ref{tab:app:slope_coef} and their 90\% naive confidence intervals given in Figure \ref{fig:empirics_naive}. There are two groups with estimated positive effects of the minimum wage on employment: one with a large coefficient and the other with a small effect. Similarly, two groups have negative effects, with one having a large effect and the other having a small effect. However, the interpretation of this result is more nuanced once we take into account the statistical uncertainty of the group membership estimation, as we discuss in detail below. 

\begin{figure}
	\centering
     \includegraphics[width=\columnwidth]{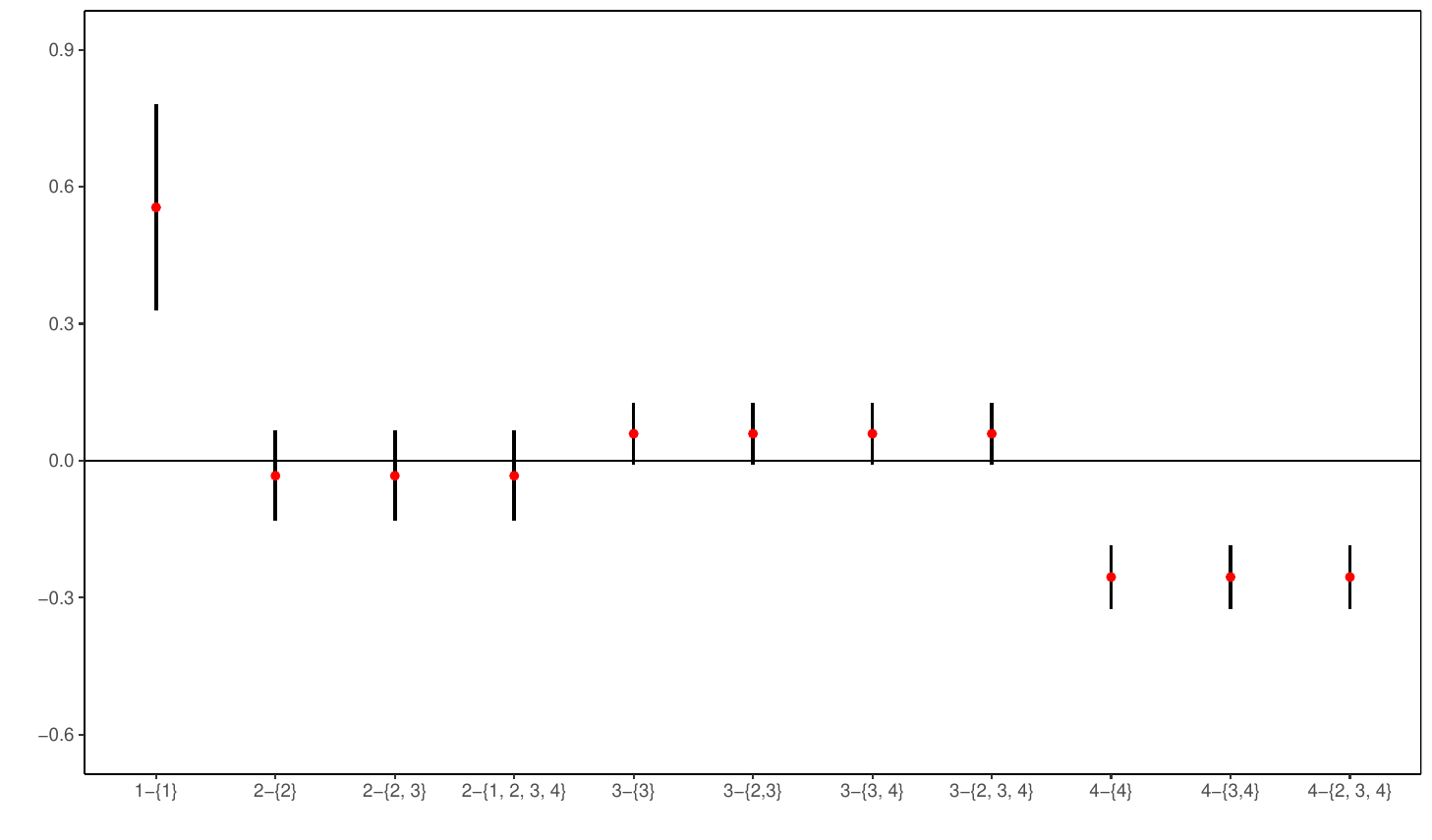}
    \caption{90\% naive confidence intervals}
    \label{fig:empirics_naive}
\end{figure}

\begin{figure}
	\centering
     \includegraphics[width=\columnwidth]{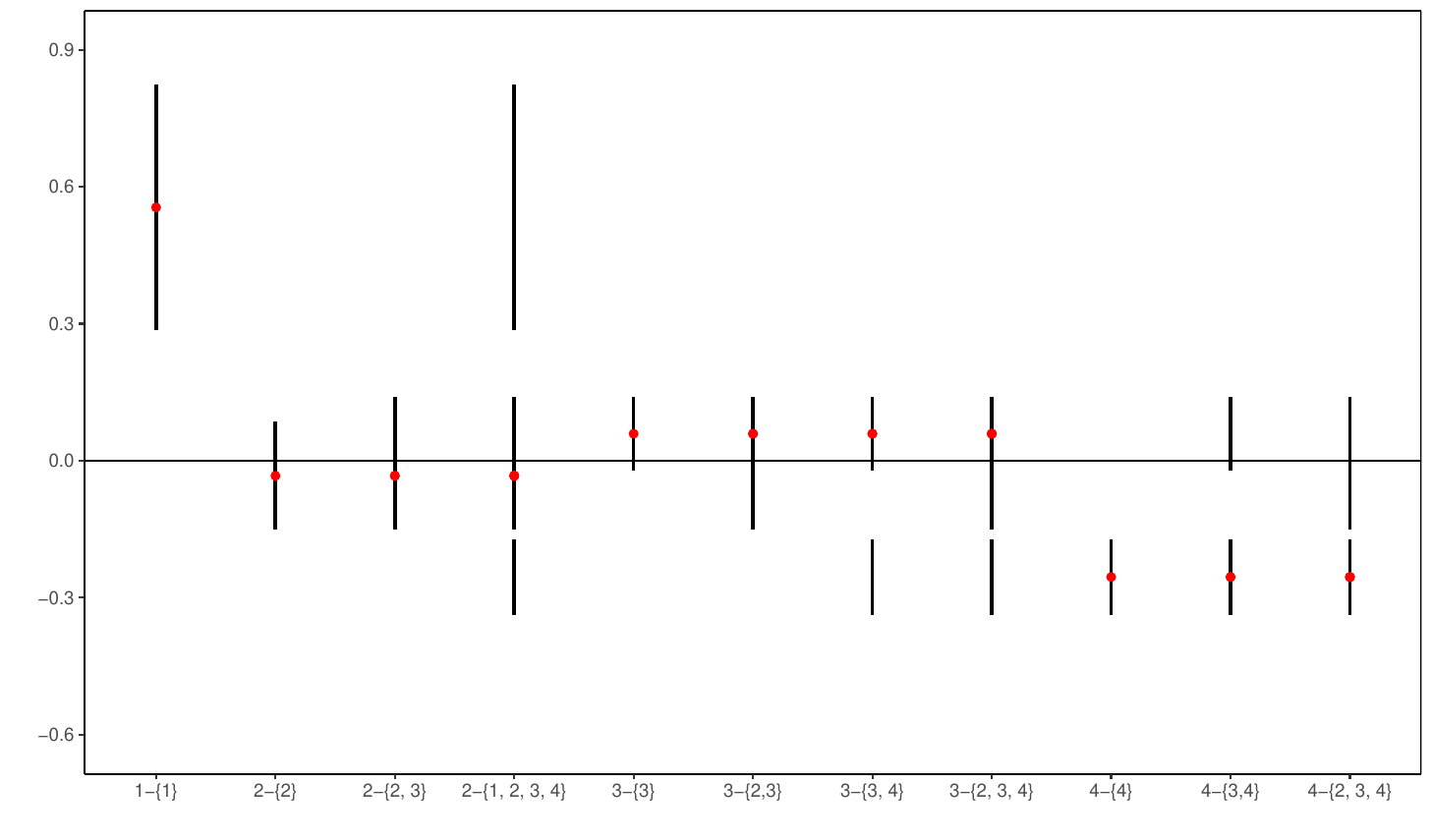}
    \caption{90\% minimum-t-based confidence intervals}
    \label{fig:empirics_umt}
\end{figure}

\begin{figure}
	\centering
        \includegraphics[width=\columnwidth]{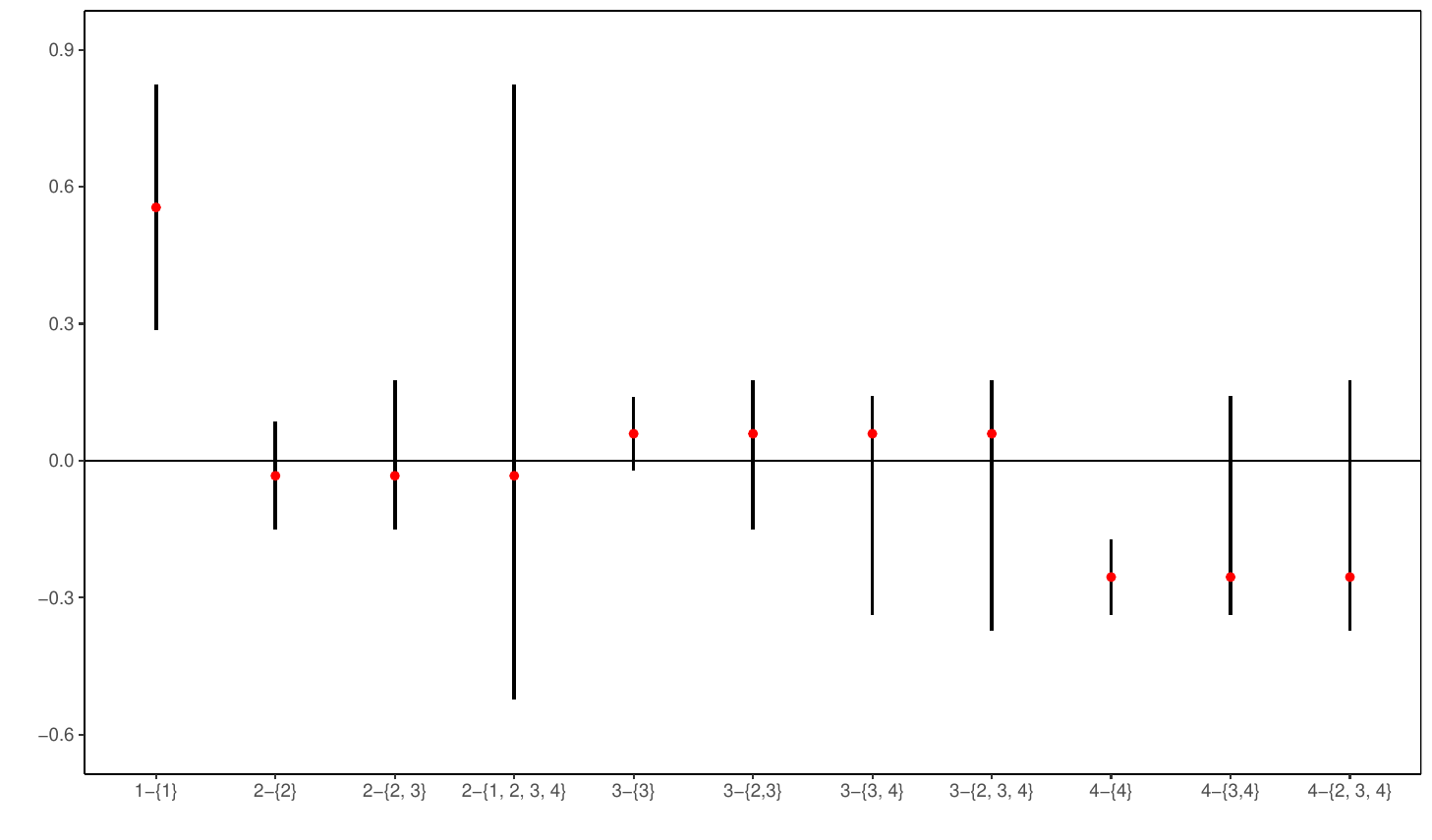}
    \caption{90\% bias-corrected confidence intervals}
    \label{fig:empirics_bc}
\end{figure}

We calculate the 90\% confidence intervals for each state with $\alpha_1=\alpha_2=0.05$. Figures \ref{fig:empirics_umt} and \ref{fig:empirics_bc} display the minimum-type and bias-corrected confidence intervals, respectively. The results of statistical inference, accounting for statistical uncertainty, reveal interesting patterns. Seven states in Category 1-\{1\} have a statistically significant positive effect of a minimum wage increase. Seven states in Category 4-\{4\} have a statistically significant negative effect. All other states have statistically insignificant effects. Note that all states with statistically significant effects have group memberships estimated precisely. While four states in Categories 4-\{3,4\} and 4-\{2,3,4\} (namely, Colorado, Nebraska, South Dakota, and Utah) are estimated to experience statistically significant negative effects from a minimum wage increase if possible misassignment is ignored, these effects are not statistically significant once statistical uncertainty is taken into account. 

We also observe that the minimum-type procedure yields disconnected confidence intervals in certain cases. In particular, Categories 3-\{3,4\} and 4-\{3,4\} have wide gaps between two sub-intervals that constitute the confidence intervals. In contrast, the bias-corrected confidence intervals are always connected by construction. However, they can be much longer, particularly when the minimum-type confidence intervals are disconnected and have wide gaps between sub-intervals.

\begin{rem}
    We note that state-by-state regression cannot be applied in our dataset. This is because, within each state, the cross-county variation of the minimum wage is almost absent in each quarter and thus absorbed in the time fixed effects $\tau_t$. In contrast, our proposed methods can be applied because they utilize cross-state variation of the minimum wage in each quarter.
\end{rem}

\section{Monte Carlo simulations}\label{sec:simu}

This section provides simulation results. Monte Carlo simulations are conducted to examine the finite-sample performance of the proposed methods in comparison with existing alternatives. In particular, we aim at quantifying the sensitivity of these confidence sets to unit-wise error variances. On the one hand, units having a large error variance are liable to misclassification, as pointed out by \textcite{dzemskiokui2021convergence,dzemski2024confidence}. This can lead to incorrect centering for the clustering-based naive confidence set. On the other hand, a large error variance implies that the parameter estimates have large standard errors, resulting in longer confidence intervals. In this section, we evaluate the coverage and length properties of several confidence sets.

Our simulation design is based on our empirical application and given by

\begin{align}
    y_{it}=\theta_{g_i^0, 1}x_{it,1} +\theta_{g_i^0, 2}x_{it,2} + \theta_{g_i^0, 3}x_{it,3}+\sigma_{u,i}u_{it}, \quad i=1, \ldots, N, \quad t=1, \ldots, T,
\end{align}
where $g_i^0 \in\{1,2,3,4\}$, $x_{it,j} \sim N(0,1)$, $u_{it}\sim N(0,1)$, and $\sigma_{u,i} = 0.5 \sigma_{i}$ with $\sigma_i \sim \chi^2(4)/4$ generated from the chi-square distribution with 4 degrees of freedom.\footnote{We use one realized vector $(\sigma_1,\ldots,\sigma_N)$ for all simulations. This design allows us to easily analyze how the performance of the procedures is related to the magnitudes of the error variances. Although we unavoidably choose one particular set of error variances, the realized values are sufficiently diverse and informative.} $\theta_{g,1}, \theta_{g,2}, \theta_{g,3}$ are set equal to the estimated coefficients in Table \ref{tab:app:slope_coef}, and we set $g_i^0 = g$ for $i\in\{(g-1)\times N/4+1,\ldots,g\times N/4\}$.

We generate panels of size $N \in \{40,80,160\}$ and $T\in\{40,80\}$. In each simulation, we calculate 90\% confidence intervals for $\theta_{g_i^0,1}$ and confidence sets for $(\theta_{g_i^0,1}, \theta_{g_i^0,2})$ for each $i$. For unit-wise confidence intervals, we calculate the unit-wise mean lengths and coverage ratios, averaging over 1000 simulations. For unit-wise confidence sets, we report the unit-wise coverage ratios. To save space, we only report the cases with $T=80$. The results for $T=40$ are presented in Appendix \ref{appendix:simulation}. Thoughout the simulations, standard errors (or asymptotic variance estimators) are fixed-$T$ versions.\footnote{We also tried the usual cluster-robust standard errors or asymptotic variance estimators. We found that the fixed-$T$ versions achieve better coverage and do not necessarily make inference more conservative. This improvement is more visible when $T=40$.}

\subsection{Performance of confidence intervals}\label{subsec:simu_scalar}

We calculate four 90\% confidence intervals for $\theta_{g_i^0,1}$: the unit-by-unit confidence interval, $\mathrm{CI}^{\mathrm{ubu}}(i,0.1)$, the naive confidence interval (ignoring uncertainty in $\hat{g}_i$), $\mathrm{CI}^{\mathrm{naive}}(i,0.1)$, the bias-corrected confidence interval, $\mathrm{CI}^{\mathrm{BC}}(i,0.1)$, and $\mathrm{CI}^{\mathrm{MT}}(i,0.1)$. We set $\alpha_1=\alpha_2=0.05$.

We first investigate coverage properties. Figure \ref{fig:coverage_80} plots the empirical coverage fractions of the confidence intervals across different magnitudes of $\sigma_i$ (the double of the error standard deviation $\sigma_{u,i}$). When $N=40$, the unit-by-unit confidence interval generally attains the nominal 90\% coverage rate for all $\sigma_i$, while the naive confidence interval suffers from severe undercoverage for moderate to large $\sigma_i$, because they do not account for potential misclassification, or uncertainty in $\hat{g}_i$. The confidence interval based on the minimum $t$-statistic can undercover for small $\sigma_i$ when $N=40$, yet as $N$ increases, it attains nominal coverage rate or becomes conservative. The bias-corrected confidence interval is conservative for moderate to large $\sigma_i$. These patterns are also observed for $N=80$ and $N=160$, with an exception that the minimum-statistic-based confidence interval is slightly conservative for all $\sigma_i$. Generally speaking, the proposed methods successfully overcome the severe undercoverage problem that the naive method suffers from by taking uncertainty in $\hat{g}_i$ into account.

\begin{figure}
	\centering
	    \begin{subfigure}{0.64\columnwidth}
        \includegraphics[width=\columnwidth]{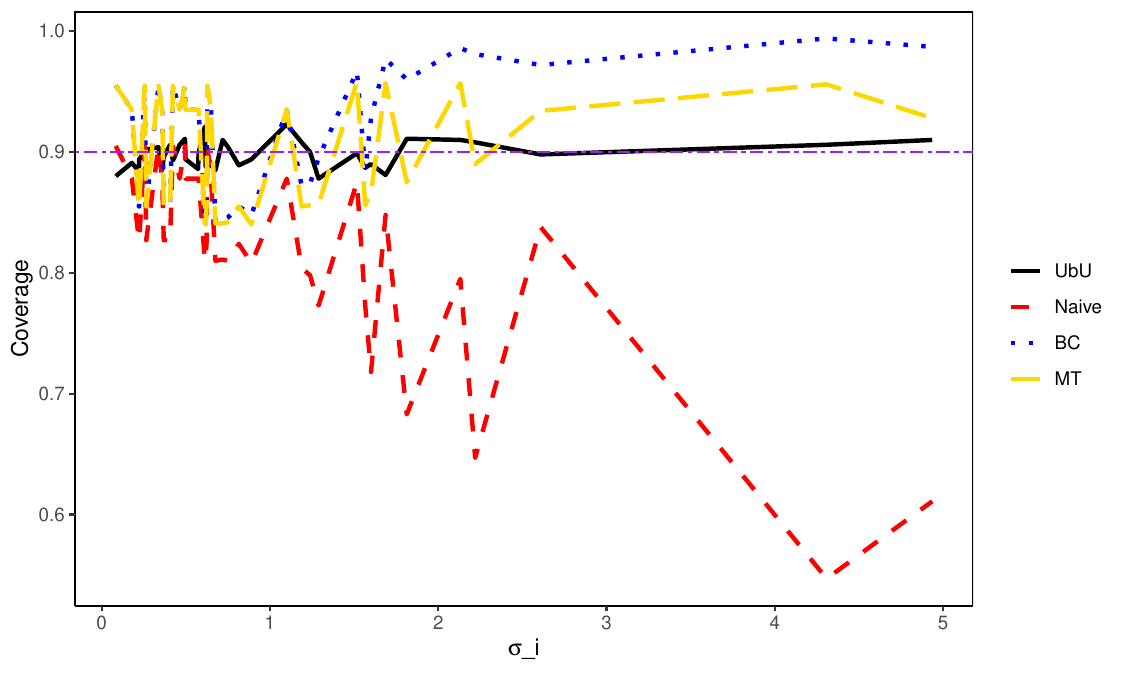}
    \caption{$(N,T)=(40,80)$}
	\label{fig:coverage_40_80}
    \end{subfigure} \quad
    \begin{subfigure}{0.64\columnwidth}
        \includegraphics[width=\columnwidth]{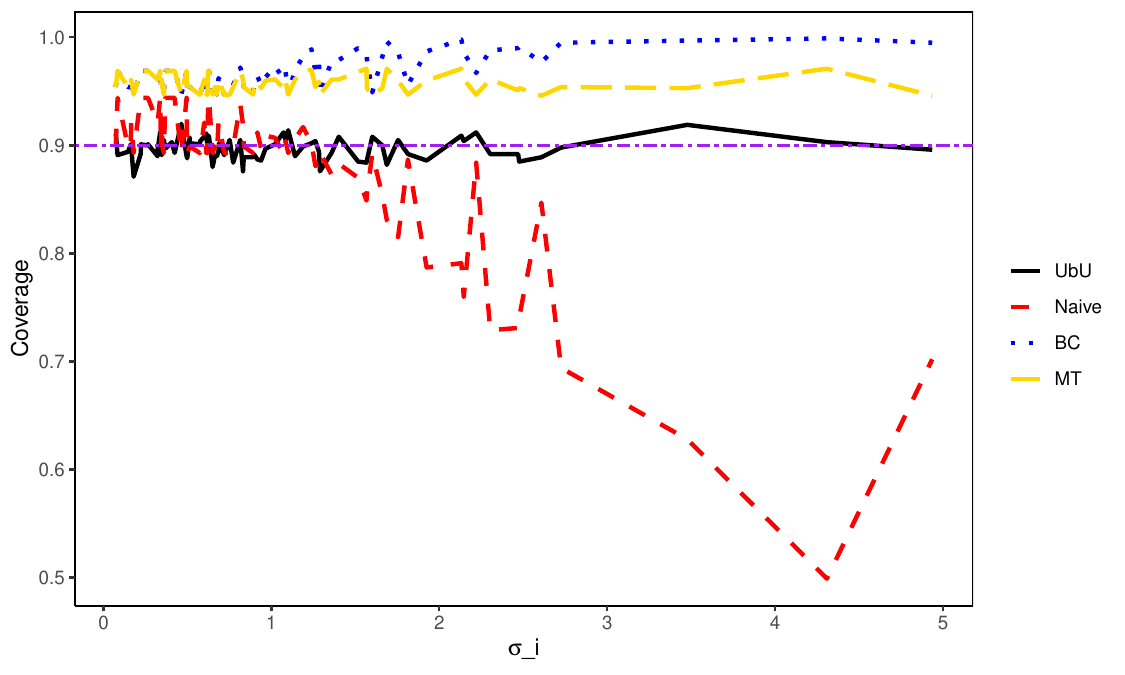}
    \caption{$(N,T)=(80,80)$}
	\label{fig:coverage_80_80}
    \end{subfigure} \quad
    \hfill
    \begin{subfigure}{0.64\columnwidth}
        \includegraphics[width=\columnwidth]{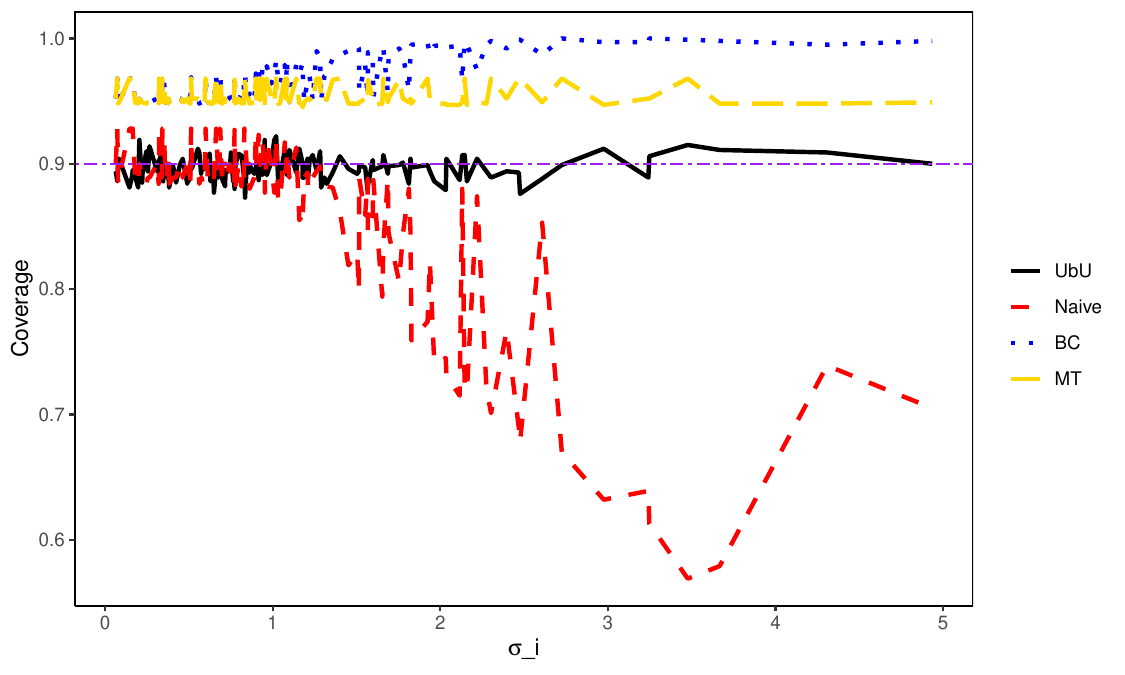}
    \caption{$(N,T)=(160,80)$}
	\label{fig:coverage_160_80}
    \end{subfigure} \quad
    \caption{Unit-wise coverage rates of confidence intervals, $T=80$}
    \label{fig:coverage_80}
    \caption*{\full: Unit-by-unit (UbU), \quad \color{red}\dashed\color{black}: Naive, \quad \color{black}\dotted\color{black}: Bias-corrected (BC), \quad \color{gold}\longdash\color{black}: Minimum t (MT)}
\end{figure}

Next, let us evaluate the length properties. We plot the unit-wise mean lengths against the corresponding values of $\sigma_i$. See Figure \ref{fig:length_80}. First, we look at the case with $N=40$. The naive confidence interval performs best in terms of mean length for all $\sigma_i$. This is because the naive confidence interval is based on a single, small group-specific standard error. The unit-by-unit confidence interval is short only for small $\sigma_i$, and its length linearly increases as $\sigma_i$ increases. The bias-corrected confidence interval is as long as the unit-by-unit confidence interval for almost all $\sigma_i$. This result is likely from the fact that the bias-correction method is connected and need to cover the range between multiple parameter values when the confidence set for group membership is not singleton. The minimum-type confidence interval is the second shortest on average. As $N$ increases, the naive confidence interval remains the shortest, the length of the bias-corrected confidence interval does not change much, and the minimum-type confidence interval becomes shorter than in the $N=40$ case. This is an advantage of the minimum-type confidence interval over the unit-by-unit confidence interval, whose power does not improve as $N$ increases.

\begin{figure}
	\centering
	    \begin{subfigure}{0.64\columnwidth}
        \includegraphics[width=\columnwidth]{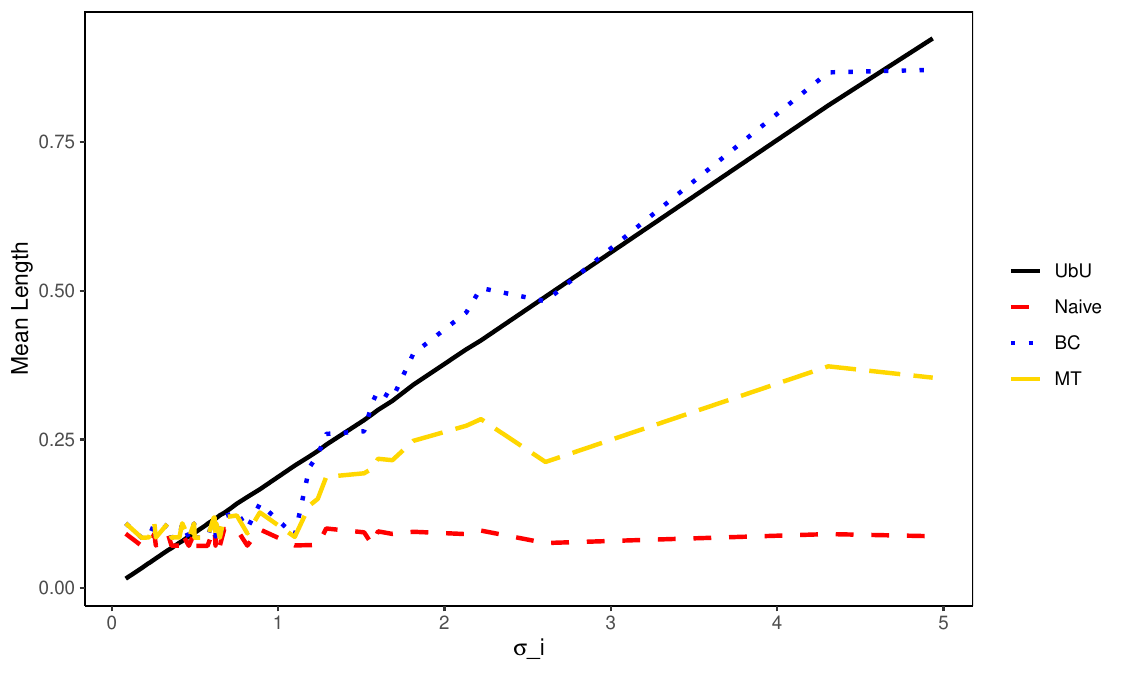}
    \caption{$(N,T)=(40,80)$}
	\label{fig:length_40_80}
    \end{subfigure} \quad
    \begin{subfigure}{0.64\columnwidth}
        \includegraphics[width=\columnwidth]{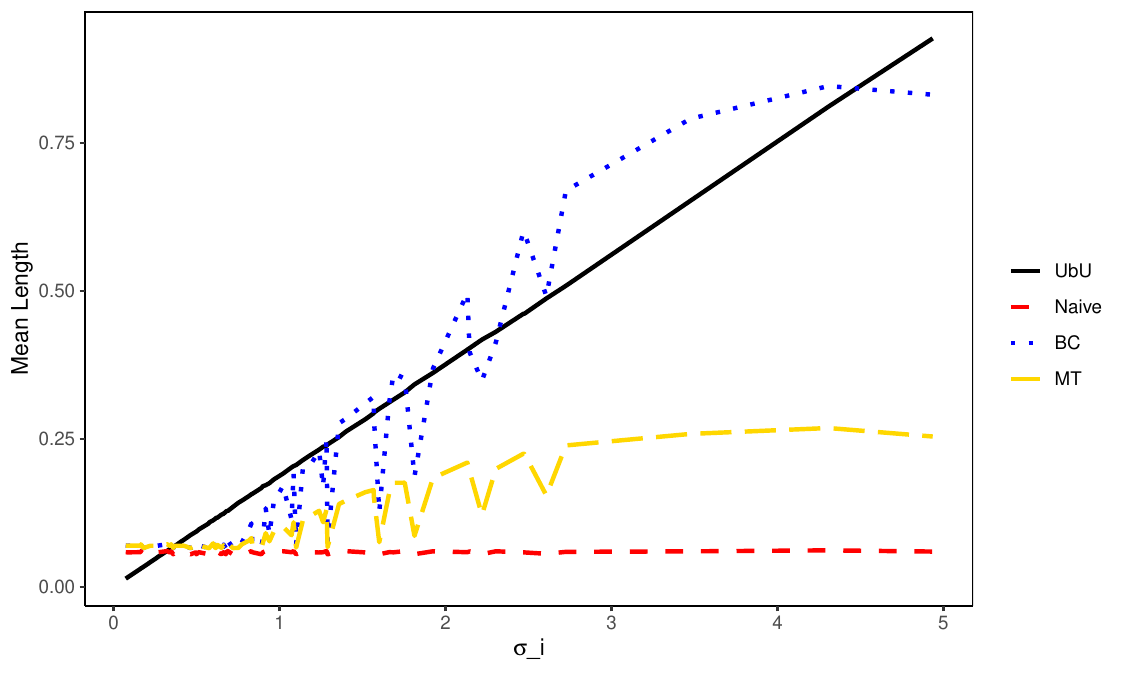}
    \caption{$(N,T)=(80,80)$}
	\label{fig:length_80_80}
    \end{subfigure} \quad
    \hfill
    \begin{subfigure}{0.64\columnwidth}
        \includegraphics[width=\columnwidth]{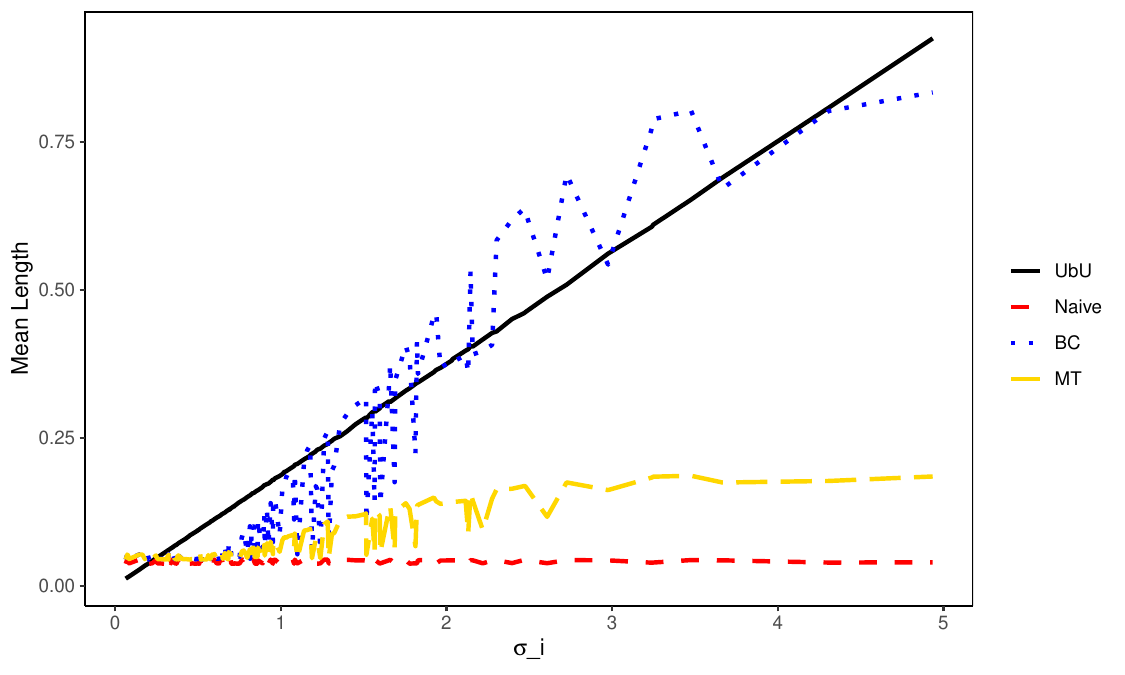}
    \caption{$(N,T)=(160,80)$}
	\label{fig:length_160_80}
    \end{subfigure} \quad
    \caption{Unit-wise mean lengths of confidence intervals, $T=80$}
    \label{fig:length_80}
    \caption*{\full: Unit-by-unit (UbU), \quad \color{red}\dashed\color{black}: Naive, \quad \color{black}\dotted\color{black}: Bias-corrected (BC), \quad \color{gold}\longdash\color{black}: Minimum t (MT)}
\end{figure}

\subsection{Performance of confidence sets}\label{subsec:simu_vec}

In this subsection, we evaluate the coverage properties of four 90\% confidence sets for bivariate parameters $(\theta_{g_i^0,1}, \theta_{g_i^0,2})$, namely, the unit-by-unit confidence set, the naive confidence set, $\mathrm{CS}^{\mathrm{BC}}(i,0.1)$, and $\mathrm{CS}^{\mathrm{MW}}(i,0.1)$ with $\alpha_1=\alpha_2=0.05$. Figure \ref{fig:coverage_80_vec} shows the result for $T=80$. The behaviors of the four confidence sets are qualitatively the same as those of the confidence intervals considered in Section \ref{subsec:simu_scalar}: the unit-by-unit confidence set has correct coverage, the naive confidence set severely undercovers when $N$ is small or $\sigma_i$ is moderate to large, and the bias-corrected confidence set tends to be conservative. The only exception is that the minimum-type confidence set tends to undercover for small to moderate $\sigma_i$ when $N=40$, but it becomes slightly conservative for larger $N$, as observed for the minimum-type confidence interval.

\begin{figure}
	\centering
	    \begin{subfigure}{0.64\columnwidth}
        \includegraphics[width=\columnwidth]{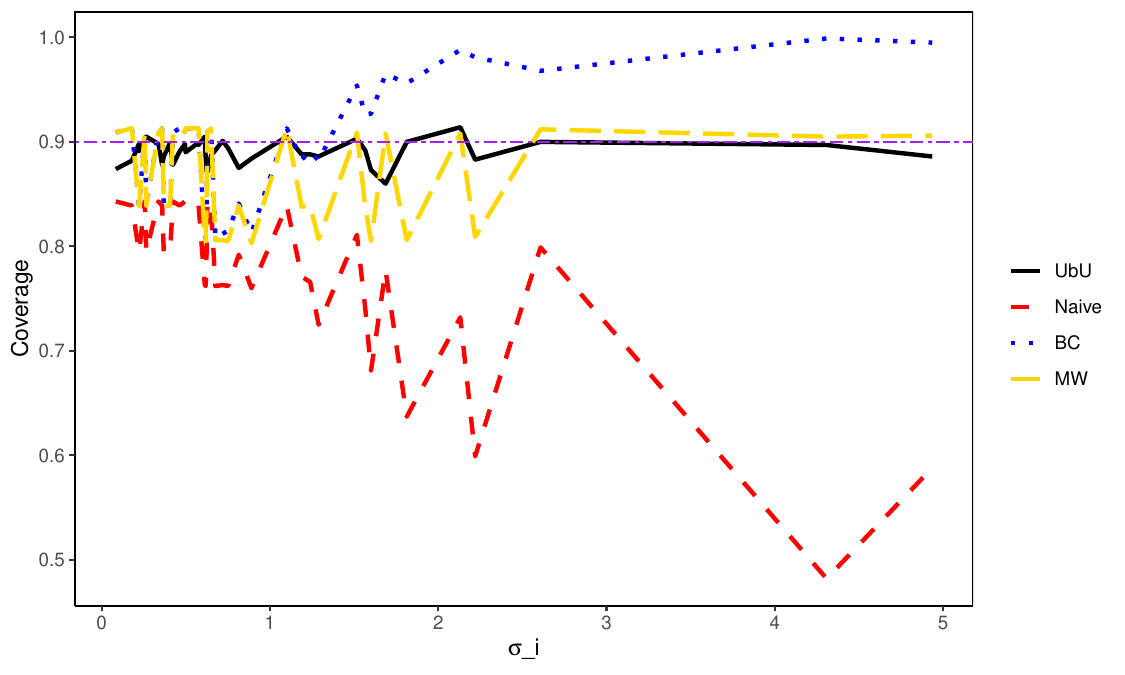}
    \caption{$(N,T)=(40,80)$}
	\label{fig:coverage_40_80_vec}
    \end{subfigure} \quad
    \begin{subfigure}{0.64\columnwidth}
        \includegraphics[width=\columnwidth]{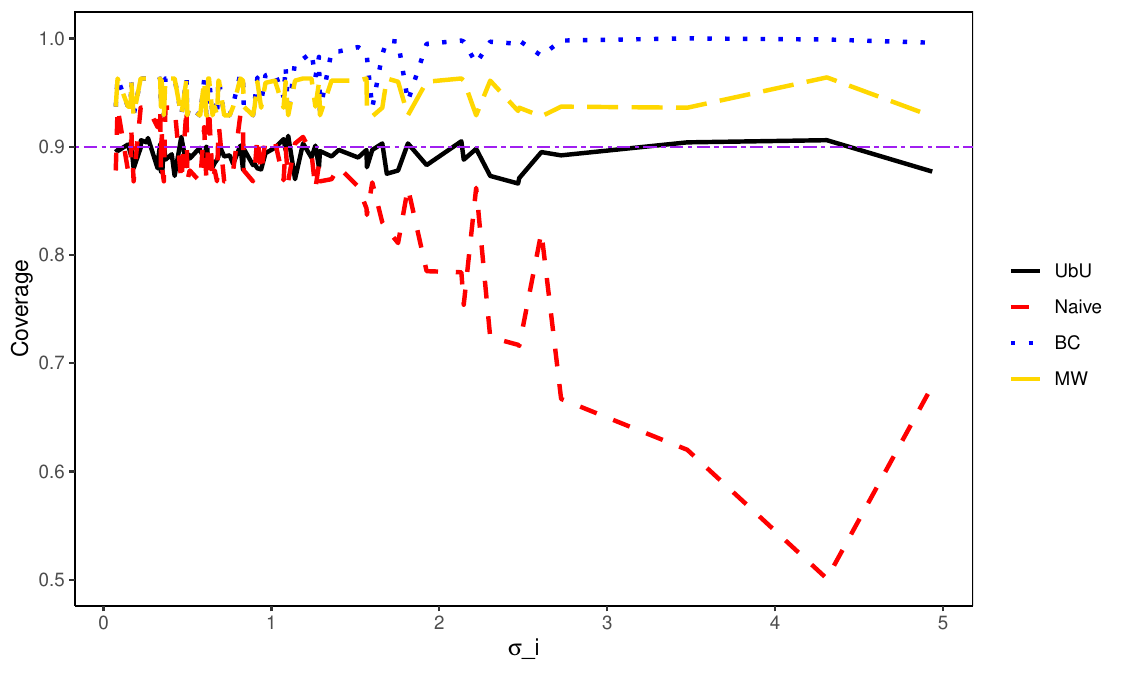}
    \caption{$(N,T)=(80,80)$}
	\label{fig:coverage_80_80_vec}
    \end{subfigure} \quad
    \hfill
    \begin{subfigure}{0.64\columnwidth}
        \includegraphics[width=\columnwidth]{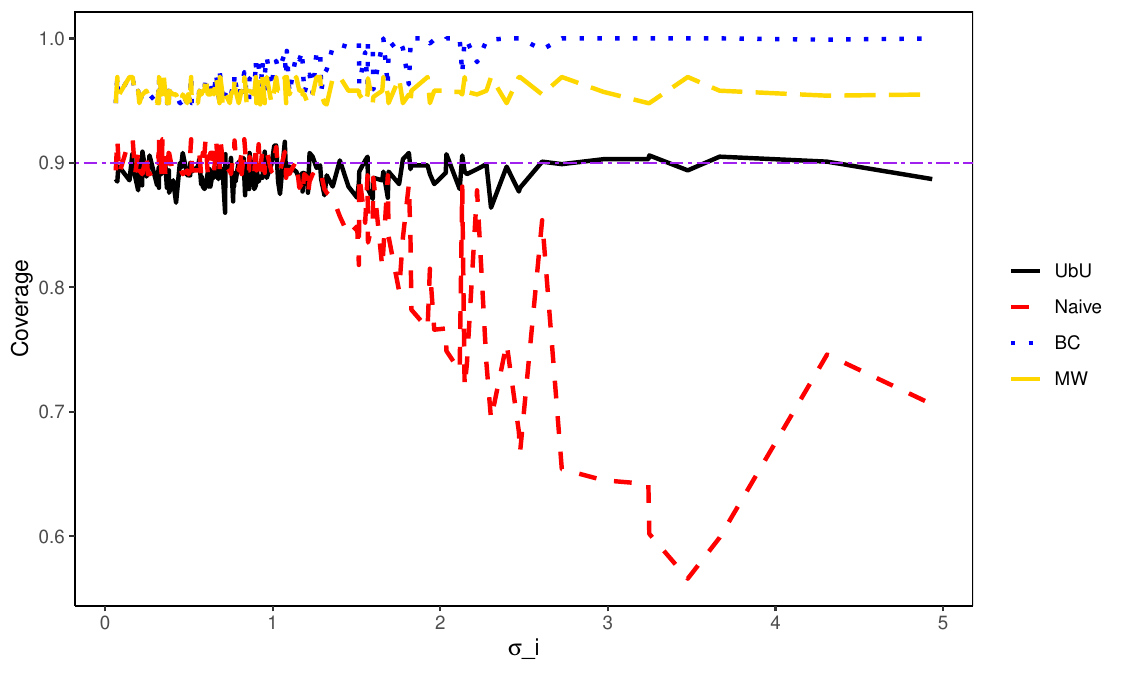}
    \caption{$(N,T)=(160,80)$}
	\label{fig:coverage_160_80_vec}
    \end{subfigure} \quad
    \caption{Unit-wise coverage rates of confidence sets, $T=80$}
    \label{fig:coverage_80_vec}
    \caption*{\full: Unit-by-unit (UbU), \quad \color{red}\dashed\color{black}: Naive, \quad \color{black}\dotted\color{black}: Bias-corrected (BC), \quad \color{gold}\longdash\color{black}: Minimum Wald (MW)}
\end{figure}

\section{Conclusion}
\label{sec:conclusion}

This paper introduces inference methods for unit-specific coefficients that exhibit a grouped pattern of heterogeneity. Our approach combines the efficiency gains from clustering units with accounting for statistical uncertainty in group membership. We propose two methods: one uses the minimum-type statistic within the confidence set for group membership, and the other corrects for bias due to possible group misassignment. The two methods present a deliberate trade-off. The minimum-type method yields shorter and more powerful confidence sets, but these may be disconnected and harder to communicate. The bias-correction method always returns a connected, readily interpretable interval for a scalar parameter, at the cost of being conservative and, in the worst case, no shorter than the unit-by-unit interval. Because both methods are valid, the choice can be made on the ground of a researcher's preference toward interpretability and power. These methods are particularly useful when group assignment is not perfectly certain. Ignoring statistical uncertainty in group membership estimation can severely distort test size and confidence set coverage. While unit-by-unit estimation is robust, it tends to be inefficient. Our proposed methods achieve adequate size and coverage while harnessing the efficiency of clustering. These tools enable empirical researchers to analyze unit-specific coefficients accurately, without distorting statistical properties.

Our inference procedures offer several opportunities for further development. This paper focuses on linear models, but extending these methods to nonlinear contexts would be beneficial. Empirical research often examines unit-specific coefficients in nonlinear models with latent group structures, a scenario common in experimental behavioral economics \parencite{bruhin2010risk, conte2011mixture}. In these studies, researchers classify participants into behavioral types and examine the magnitude of behavioral biases at the individual level. Most analyses use structural and nonlinear models.

The current paper focuses on situations in which the covariates are exogenous. We chose this focus because, in such situations, both consistent group-specific coefficient estimators and a valid confidence set for group membership have been developed. When the covariates are not exogenous, we may employ instrumental-variable estimators or generalized method-of-moments estimators \parencite{su_identifying_2016,mehrabani2022estimation,choi2024latentgroupstructurelinear}. In those cases, methods to construct confidence sets for group membership have not been developed. In the future, once such methods are developed, our proposed procedure can be readily applied.

\appendix

\section*{Appendix}

\section{Details of constructing confidence sets for group memberships}
\label{appendix: confsetgroup}

This section describes the details of constructing the marginal confidence set for the group memberships. We complement the brief discussion in Section \ref{subsec:cs_membership}. Note that all the materials presented here follow \textcite{dzemski2024confidence} and we present them merely for the current paper to be self-contained. We refer the reader to \textcite{dzemski2024confidence} for motivation and theoretical properties.

\subsection{Long-run variance estimation}
The long-run variance-covariance estimator is given by
\begin{align*}
\widehat{\Xi}_{i}(g, h, h') = \sum_{j = -T+1}^{T-1} K\left(\frac{j}{\kappa_N}\right) \widehat{H}_{ij}(g, h, h'),
\end{align*}
where
\begin{align}
\widehat{H}_{ij}(g, h, h')
= \frac{1}{T} \sum_{t = |j|+1}^{T}
\left(\hat{d}_{i, t+\min(0,j)}(g, h) - \bar{\hat{d}}_{i}(g, h)\right)
\left(\hat{d}_{i, t-\max(0,j)}(g, h') - \bar{\hat{d}}_{i}(g, h')\right) \label{eq:H_hat}
\end{align}
with $\bar{\hat{d}}_{i}(g, h) = \frac{1}{T}\sum_{t=1}^T \hat{d}_{it}(g, h)$, and $\kappa_N$ is a bandwidth parameter.
We use the quadratic spectral kernel \parencite{Andrews91}:
\begin{align*}
K_{QS}(x) = \frac{25}{12\pi^2 x^2}\left(\frac{\sin(6\pi x / 5)}{6\pi x / 5} - \cos(6\pi x / 5)\right).
\end{align*}
The bandwidth $\kappa_N$ is selected by the following data-driven algorithm adapted from \textcite{chang2022central}:
\begin{description}
\item[Step A] For each $g \in \mathbb{G}$ and $h \in \mathbb{G} \setminus \{g\}$, take each unit $i$ such that $\hat{g}_i = g$ and fit an $\mathrm{AR}(1)$ model on $(\hat{d}_{it}(g, h))_{t=1}^T$. Let $\hat{\rho}_{igh}$ denote the estimated autoregressive coefficient and $\hat{\sigma}^2_{igh}$ the estimated innovation variance.
\item[Step B] Select the bandwidth
\begin{align*}
\kappa_N = 1.3221 \left(
T \times \frac{
\sum_{i=1}^N \sum_{g \in \mathbb{G}} \sum_{h \in \mathbb{G} \setminus \{g\}} \hat{\rho}{igh}^2 \hat{\sigma}_{igh}^4 / (1 - \hat{\rho}_{igh}^2)^8
}{
\sum_{i=1}^N \sum_{g \in \mathbb{G}} \sum_{h \in \mathbb{G} \setminus \{g\}} \hat{\sigma}_{igh}^4 / (1 - \hat{\rho}_{igh}^2)^4
}
\right)^{1/5}.
\end{align*}
\end{description}

\subsection{Critical value}

Recall that for $G \geq 3$, the critical value is given by
\begin{align*}
\hat{c}_{i, \alpha}(g)
= \sqrt{\frac{T}{T-1}} \left(t_{\max, \rho(\widehat{\Omega}_i(g), \epsilon_N), T-1}\right)^{-1} \left(1 - \alpha \right),
\end{align*}
where $t_{\max, \Omega, T-1}$ denotes the distribution function of the maximum entry of a centered random vector with multivariate $t$-distribution with scale matrix $\Omega$ and $(T-1)$ degrees of freedom. 

We now explain how to obtain $\widehat{\Omega}_i(g)$, $\rho $ and $\epsilon_N$.
The estimated correlation matrix $\widehat{\Omega}_i(g)$ is a $(G-1) \times (G-1)$ matrix with entry $(j, j')$ equal to
\begin{align*}
\left(\widehat{\Omega}_i(g)\right)_{j,j'}
= \frac{\widehat{\Xi}_i(g, h, h')}{\sqrt{\widehat{\Xi}_i(g, h, h), \widehat{\Xi}_i(g, h', h')}},
\end{align*}
where $h$ and $h'$ denote the $j$th and $j'$th elements of $\mathbb{G} \setminus \{g\}$, respectively.
The regularization function $\rho$ is defined as
\begin{align*}
\rho(\Omega, \epsilon) = \mathrm{diag}^{-1/2}\left(\Omega + \epsilon^*(\Omega, \epsilon), I_{G-1}\right)
\left(\Omega + \epsilon^*(\Omega, \epsilon), I_{G-1}\right)
\mathrm{diag}^{-1/2}\left(\Omega + \epsilon^*(\Omega, \epsilon), I_{G-1}\right),
\end{align*}
where
\begin{align*}
\epsilon^*(\Omega, \epsilon) = \max \left\{0, \epsilon - \left(1 - \max_{i < j} \Omega_{ij}\right)\right\}.
\end{align*}
We set $\epsilon_N = 0.01$.

\section{Details of the empirical application}
\label{app:detail_empirical}

This appendix provides the details of constructing the confidence set for group memberships in the empirical application in Section \ref{sec:empirical}. All materials presented here follow \textcite{dzemski2024confidence}. We present them to ensure the current paper is self-contained. For motivation and more detailed explanations, refer to \textcite{dzemski2024confidence}.

We estimate group memberships by running a single update step of the \emph{k-means} algorithm using the estimated slope coefficients. This update step ensures that the estimated group memberships satisfy the procedure's condition. The resulting estimates are identical to the CLasso estimates in all but six states.

We identify states (subscript $i$) as the cross-sectional dimension and counties and quarters (subscripts $c$ and $t$) jointly as the second dimension. The confidence sets are computed based on the fixed-effect transformation:
\begin{align*}
\widetilde{\texttt{lemp}}_{ict} = \theta_{g^0_s, 1}, \widetilde{\texttt{lmw}}_{ict} + \theta_{g^0_s, 2}, \widetilde{\texttt{lpop}}_{ict} + \theta_{g^0_s, 3}, \widetilde{\texttt{lemp}}_{ict}^{\texttt{TOT}} + \tilde{\tau}_t + \sigma_i \tilde{v}_{ict},
\end{align*}
where $\widetilde{\texttt{lemp}}_{ict} = \log(\texttt{lemp}_{ict}) - \sum_{t'=1}^T \log(\texttt{lemp}_{ict'}) / T$, and $\widetilde{\texttt{lmw}}_{ict}$, $\widetilde{\texttt{lpop}}_{ict}$, $\widetilde{\texttt{lemp}}_{ict}^{\texttt{TOT}}$ are defined analogously. Moreover, $\tilde{\tau}_t = \tau_t - \sum_{t'=1}^T \tau_{t'} / T$ and $\tilde{v}_{ict} = v_{ict} - \sum_{t'=1}^T v_{ict'} / T$.
Since the second dimension comprises both cross-sectional variation between counties and time-series variation between quarters, the estimated autocovariance $\widehat{H}_{ij}$ is redefined as
\begin{align*}
\widehat{H}_{ij}(g, h, h')
= \frac{1}{T n_i} \sum_{c=1}^{n_i} \sum_{t=|j|+1}^{T}
\left(\hat{d}_{ic, t+\min(0,j)}(g, h) - \bar{\hat{d}}_{ic}(g, h)\right)
\left(\hat{d}_{ic, t-\max(0,j)}(g, h') - \bar{\hat{d}}_{ic}(g, h')\right),
\end{align*}
where $\bar{\hat{d}}_{ic}(g, h) = \sum_{t=1}^T \hat{d}_{ict}(g, h) / T$ and $\hat{d}_{ict}$ replaces $\hat{d}_{it}$ with the double index $ct$ in place of $t$. The bandwidth selection algorithm is modified accordingly:
\begin{description}
\item[Step A] For each $g \in \mathbb{G}$ and $h \in \mathbb{G} \setminus \{g\}$, take each state $i$ with $\hat{g}_i = g$ and fit an $\mathrm{AR}(1)$ model on $(\hat{d}_{ict}(g, h))_{t=1}^T$ for each county $c = 1, \dotsc, n_i$. Let $\hat{\rho}_{icgh}$ and $\hat{\sigma}^2_{icgh}$ denote the estimated autoregressive coefficient and innovation variance, respectively.
\item[Step B] Select the bandwidth
\begin{align*}
\hat{\kappa}_N = 1.3221 \left(
T \times \frac{
\sum_{i=1}^N \sum_{c=1}^{n_i} \sum_{g \in \mathbb{G}} \sum_{h \in \mathbb{G} \setminus \{g\}} \hat{\rho}_{icgh}^2, \hat{\sigma}_{icgh}^4 / (1 - \hat{\rho}_{icgh}^2)^8
}{
\sum_{i=1}^N \sum_{c=1}^{n_i} \sum_{g \in \mathbb{G}} \sum_{h \in \mathbb{G} \setminus \{g\}} \hat{\sigma}_{icgh}^4 / (1 - \hat{\rho}_{icgh}^2)^4
}
\right)^{1/5}.
\end{align*}
\end{description}
This yields $\hat{\kappa}_N = 21.37$.

\section{Proofs of the theoretical results}
\label{appendix: proof}

\noindent\begin{proof}[Proof of Theorem \ref{thm:coverage_um}]

(i) By the definition of $\mathrm{CI}^{\mathrm{MT}}(i, \alpha)$, we have
\begin{align}
P\left(R\theta_{g_i^0}^{0} \notin \mathrm{CI}^{\mathrm{MT}}(i, \alpha)\right) 
& \leq P\left(R\theta_{g_i^0}^{0} \notin \mathrm{CI}^{\mathrm{MT}}(i, \alpha), \ g_i^0 \in \mathrm{CS}_{i, \alpha_1}\right)+\alpha_1 + o(1) \\
&=P\left(\min_{g\in \mathrm{CS}_{i,\alpha_1}} \left|\frac{R\left(\hat{\theta}_g-\theta_{g_i^0}^0\right)}{\mathrm{se}(g)}\right| > z_{1-\alpha_2/2}, \ g_i^0 \in \mathrm{CS}_{i, \alpha_1}\right)+\alpha_1 + o(1) \\
&\leq P\left(\left|\frac{R\left(\hat{\theta}_{g_i^0}-\theta_{g_i^0}^0\right)}{\mathrm{se}(g_i^0)}\right| > z_{1-\alpha_2/2}\right)+\alpha_1 + o(1) \\
&\leq \max_{g=1,\ldots,G} P\left(\left|\frac{R\left(\hat{\theta}_{g}-\theta_{g}^0\right)}{\mathrm{se}(g)}\right| > z_{1-\alpha_2/2}\right)+\alpha_1 + o(1) \\
&=\alpha+o(1),
\end{align}
uniformly in $i$, where the first inequality follows from Assumption \ref{asm:cs_membership}, and the last equality follows from Assumptions \ref{asm:asym_normal} and \ref{asm:est_group_size}. This yields
\begin{align}
    \limsup_{N,T\to\infty}\max_{i=1,\ldots,N}P\left(R\theta_{g_i^0}^{0} \notin \mathrm{CI}^{\mathrm{MT}}(i, \alpha)\right) \leq \alpha.
\end{align}

(ii) Recalling the definition of $\mathrm{CS}^{\mathrm{MW}}(i, \alpha)$, we have

\begin{align}
P\left(R \theta_{g_{i}^{0}}^{0} \notin\mathrm{CS}^{\mathrm{MW}}(i, \alpha)\right) 
&\leq P\left(R\theta_{g_{i}^{0}}^{0} \notin\mathrm{CS}^{\mathrm{MW}}(i, \alpha), \ g_{i}^{0} \in \mathrm{CS}_{i, \alpha_1}\right)+\alpha_1 + o(1) \\
&= P\left(\widetilde{W}_{i, N T, \alpha_1}\left(\theta_{g_{i}^{0}}^{0}\right)>\chi_{p, 1-\alpha_2}^{2}, \ g_{i}^{0} \in \mathrm{CS}_{i, \alpha_1}\right) +\alpha_1 + o(1) \\
&\leq  P\left(\widehat{N}_{g_{i}^{0}} T\left(\hat{\theta}_{g_{i}^{0}}-\theta_{g_{i}^{0}}^{0}\right)^{\prime} R'\left(R\widehat{\Sigma}_{g_i^0}R'\right)^{-1}R\left(\hat{\theta}_{g_{i}^{0}}-\theta_{g_{i}^{0}}^{0}\right)>\chi_{p, 1-\alpha_2}^{2}\right)+\alpha_1 + o(1) \\
&\leq \max _{g=1, \ldots, G} P\left(\widehat{N}_{g} T\left(\hat{\theta}_{g}-\theta_{g}^{0}\right)^{\prime}R' \left(R\widehat{\Sigma}_{g}R'\right)^{-1}R\left(\hat{\theta}_{g}-\theta_{g}^{0}\right)>\chi_{p, 1-\alpha_2}^{2}\right) +\alpha_1 + o(1)\\
&= \alpha+o(1),
\end{align}
uniformly in $i$. This leads to
\begin{align}
    \limsup_{N,T\to\infty}\max_{i=1,\ldots,N}P\left(R\theta_{g_i^0}^{0} \notin\mathrm{CS}^{\mathrm{MW}}(i, \alpha)\right) \leq \alpha,
\end{align}
which completes the proof.
\end{proof}

\noindent\begin{proof}[Proof of Theorem \ref{thm:coverage_bc}]
    (i) It suffices to show 
    \begin{align}
        \limsup_{N,T\to\infty}\max_{i=1,\ldots,N}P\left(R\theta_{g_i^0}^{0} \notin \mathrm{CI}^{\mathrm{BC}}(i, \alpha)\right) \leq \alpha.
        \label{bound:noncover}
    \end{align}
    We bound the left-hand side as
\begin{align}
P\left(R\theta_{g_i^0}^{0} \notin \mathrm{CI}^{\mathrm{BC}}(i, \alpha)\right) 
& =P\left(R\theta_{g_i^0}^{0} \notin \mathrm{CI}^{\mathrm{BC}}(i, \alpha), \ g_i^0 \in \mathrm{CS}_{i, \alpha_1}\right)+P\left(R\theta_{g_i^0}^{0} \notin \mathrm{CI}^{\mathrm{BC}}(i, \alpha), \ g_i^0 \notin \mathrm{CS}_{i, \alpha_1}\right) \\
& \leq P\left(R\theta_{g_i^0}^{0} \notin \mathrm{CI}^{\mathrm{BC}}(i, \alpha), \ g_i^0 \in \mathrm{CS}_{i, \alpha_1}\right)+\alpha_1 + o(1),
\label{bound:noncover_intermediate}
\end{align}
where the $o(1)$ term is uniform in $i$ by Assumption \ref{asm:cs_membership}. Let us bound the first term of the above display:
\begin{align}
& P\left(R\theta_{g_i^0}^{0} \notin \mathrm{CI}^{\mathrm{BC}}(i, \alpha), \ g_i^0 \in \mathrm{CS}_{i, \alpha_1}\right) \\
& \leq P\left(\min _{g \in \mathrm{CS}_{i, \alpha_1}}R\hat{\theta}_{g}-z_{1-\alpha_2 /2} \times \mathrm{se}_{i}>R\theta_{g_i^0}^{0}, \ g_i^0 \in \mathrm{CS}_{i, \alpha_1}\right) \\
& \quad +P\left(\max _{g \in \mathrm{CS}_{i, \alpha_1}}R\hat{\theta}_{g}+z_{1-\alpha_2 /2} \times \mathrm{se}_{i}<R\theta_{g_i^0}^{0}, \ g_i^0 \in \mathrm{CS}_{i, \alpha_1}\right).
\label{bound:noncover_true_in_cs}
\end{align}

The lower-end noncoverage probability is bounded as
\begin{align}
& P\left(\min _{g \in \mathrm{CS}_{i, \alpha_1}}R\hat{\theta}_{g}-z_{1-\alpha_2 /2} \times \mathrm{se}_{i}>R\theta_{g_i^0}^{0}, \ g_i^0 \in\mathrm{CS}_{i, \alpha_1}\right) \\
& \leq P\left(\frac{R\left(\hat{\theta}_{g_i^0}-\theta_{g_i^0}^{0}\right)}{\mathrm{se}_{i}}>z_{1-\alpha_2 /2}, \ g_i^0 \in \mathrm{CS}_{i, \alpha_1}\right) \\
& \leq P\left(\frac{R\left(\hat{\theta}_{g_i^0}-\theta_{g_i^0}^{0}\right)}{\mathrm{se}\left(g_i^0\right)}>z_{1-\alpha_2 /2}\right) \\
&=P\left(\sqrt{\widehat{N}_{g_i^0} T}R\left(\hat{\theta}_{g_i^0}-\theta_{g_i^0}^{0}\right) / \sqrt{R\widehat{\Sigma}_{g_i^0}R'}>z_{1-\alpha_2 /2}\right) \\
&\leq \max_{g=1,\ldots,G} P\left(\sqrt{\widehat{N}_{g} T}R\left(\hat{\theta}_{g}-\theta_{g}^{0}\right) / \sqrt{R\widehat{\Sigma}_{g}R'}>z_{1-\alpha_2 /2}\right),
\end{align}
where the last upper bound is independent of $i$ and hence holds uniformly in $i=1,\ldots,N$. It follows from the above display and Assumptions \ref{asm:asym_normal} and \ref{asm:est_group_size} that
\begin{align}
    \limsup _{N, T \rightarrow \infty} \max_{i=1,\ldots,N} P\left(\min _{g \in \mathrm{CS}_{i, \alpha_1}}R\hat{\theta}_{g}-z_{1-\alpha_2 /2} \times \mathrm{se}_{i}>R\theta_{g_i^0}^{0}, \ g_i^0 \in \mathrm{CS}_{i, \alpha_1}\right) \leq \alpha_2 /2.
\end{align}

For the upper-end noncoverage probability, a symmetric argument shows
\begin{align}
    \limsup _{N, T \rightarrow \infty} \max_{i=1,\ldots,N} P\left(\max _{g \in \mathrm{CS}_{i, \alpha_1}}R\hat{\theta}_{g}+z_{1-\alpha_2 /2} \times \mathrm{se}_{i}<R\theta_{g_i^0}^{0}, \ g_i^0 \in \mathrm{CS}_{i, \alpha_1}\right) \leq \alpha_2 /2.
\end{align}
Substituting these results into \eqref{bound:noncover_true_in_cs} yields
\begin{align}
    \limsup _{N, T \rightarrow \infty} \max_{i=1,\ldots,N}P\left(R\theta_{g_i^0}^{0} \notin \mathrm{CI}^{\mathrm{BC}}(i, \alpha), \ g_i^0 \in \mathrm{CS}_{i, \alpha_1}\right) \leq \alpha_2.
\end{align}
This, in conjunction with \eqref{bound:noncover_intermediate}, establishes \eqref{bound:noncover}.

(ii) Next, we consider the case with $r\geq2$. As in the proof of part (i), we have
\begin{align}
P\left(R\theta_{g_{i}^{0}}^{0} \notin \mathrm{CS}^{\mathrm{BC}}(i, \alpha)\right)  \leq P\left(R\theta_{g_{i}^{0}}^{0} \notin \mathrm{CS}^{\mathrm{BC}}(i, \alpha), g_{i}^{0} \in \mathrm{CS}_{i, \alpha_1}\right)+\alpha_1 + o(1)
\end{align}

We bound the first term. Recalling \eqref{eqn:decompose_2}, we have
\begin{align}
& P\left(R\theta_{g_{i}^{0}}^{0} \notin \mathrm{CS}^{\mathrm{BC}}(i, \alpha), \ g_{i}^{0} \in \mathrm{CS}_{i, \alpha_1}\right) \\
& =P\left(\overline{W}_{i, N T, \alpha_1}\left(\theta_{g_{i}^{0}}^{0}\right)>\chi_{p, 1-\alpha_2}^{2}, \ g_{i}^{0} \in \mathrm{CS}_{i, \alpha_1}\right) \\
& =P\biggl(\min _{g \in \mathrm{CS}_{i, \alpha_1}} \widehat{N}_{g} T\left(\hat{\theta}_{\hat{g}_{i}}-\theta_{g_{i}^{0}}^{0}\right)^{\prime} R'\left(R\widehat{\Sigma}_{g}R'\right)^{-1}R\left(\hat{\theta}_{\hat{g}_{i}}-\theta_{g_{i}^{0}}^{0}\right) \\
&\hspace{1.5cm} -\max _{g \in \mathrm{CS}_{i, \alpha_1}} \widehat{N}_{g} T\left(\hat{\theta}_{\hat{g}_{i}}-\hat{\theta}_{g}\right)^{\prime} R'\left(R\widehat{\Sigma}_{g}R'\right)^{-1}R\left(\hat{\theta}_{\hat{g}_{i}}-\hat{\theta}_{g}\right) \\
&\hspace{1.5cm} -2 \max _{g \in \mathrm{CS}_{i, \alpha_1}}\left|\widehat{N}_{g} T\left(\hat{\theta}_{\hat{g}_{i}}-\hat{\theta}_{g}\right)^{\prime} R'\left(R\widehat{\Sigma}_{g}R'\right)^{-1}R\left(\hat{\theta}_{g}-\theta_{g_{i}^{0}}^{0}\right)\right|>\chi_{p, 1-\alpha_2}^{2}, \ g_{i}^{0} \in \mathrm{CS}_{i, \alpha_1}\biggr) \\  
&\leq P\left(\widehat{N}_{g_{i}^0}T\left(\hat{\theta}_{g_{i}^0}-\theta_{g_{i}^0}^{0}\right)^{\prime} R'\left(R\widehat{\Sigma}_{g_i^0}R'\right)^{-1}R\left(\hat{\theta}_{g_{i}^0}-\theta_{g_{i}^0}^{0}\right)>\chi_{p,1-\alpha_2}^{2}\right) \\
& \leq \max _{g=1, \ldots, G} P\left(\widehat{N}_{g} T\left(\hat{\theta}_{g}-\theta_{g}^{0}\right)^{\prime} R'\left(R\widehat{\Sigma}_{g}R'\right)^{-1}R\left(\hat{\theta}_{g}-\theta_{g}^{0}\right)>\chi_{p,1-\alpha_2}^{2}\right) \\
&=\alpha_2+o(1)
\end{align}
uniformly in $i$ by Assumptions \ref{asm:asym_normal} and \ref{asm:est_group_size}. This shows
\begin{align}
    \limsup _{N, T \rightarrow \infty} \max_{i=1,\ldots,N}P\left( R \theta_{g_{i}^{0}}^{0} \notin \mathrm{CS}^{\mathrm{BC}}(i, \alpha)\right) \leq \alpha,
\end{align}
which completes the proof.
\end{proof}

\noindent\begin{proof}[Proof of Lemma \ref{lem:vcov_consistency}]
    The consistency of $\widehat{M}_g$ immediately follows from Assumption \ref{asm:est_group_size} and conditions (a)--(c).

For $\widehat{\Omega}_{g}$, we decompose it as
\begin{align}
\widehat{\Omega}_{g}&=  \frac{N_{g}}{\widehat{N}_g} \frac{1}{N_{g} T}\left(\sum_{i=1}^{N} \sum_{t=1}^{T} \sum_{s=1}^{T} 1\left\{\hat{g}_{i}=g, g_{i}^{0}=g\right\} x_{it} x_{is}^{\prime} \hat{\varepsilon}_{it} \hat{\varepsilon}_{is} +\sum_{i=1}^{N} \sum_{t=1}^{T} \sum_{s=1}^{T} 1\left\{\hat{g}_{i} \neq g, g_{i}^{0}=g\right\} x_{it} x_{is}^{\prime} \tilde{\varepsilon}_{it} \tilde{\varepsilon}_{is}\right) \\
& \quad -\frac{N_{g}}{\widehat{N}_g} \frac{1}{N_{g} T} \sum_{i=1}^{N} \sum_{t=1}^{T} \sum_{s=1}^{T} 1\left\{\hat{g}_{i} \neq g, g_{i}^{0}=g\right\} x_{it} x_{is}^{\prime} \tilde{\varepsilon}_{it} \tilde{\varepsilon}_{is} \\
& +\frac{N_{g}}{\widehat{N}_g} \frac{1}{N_{g} T} \sum_{i=1}^{N} \sum_{t=1}^{T} \sum_{s=1}^{T} 1\left\{\hat{g}_{i}=g, g_{i}^{0} \neq g\right\} x_{it} x_{is}^{\prime} \hat{\varepsilon}_{it} \hat{\varepsilon}_{is,} \\
&= (1+o_p(1)) \times \widetilde{\Omega}_{g}-A_{1}+A_{2},
\end{align}
where
\begin{align}
    A_{1}\coloneqq\frac{N_{g}}{\widehat{N}_g} \frac{1}{N_{g} T} \sum_{i=1}^{N} \sum_{t=1}^{T} \sum_{s=1}^{T} 1\left\{\hat{g}_{i} \neq g, g_{i}^{0}=g\right\} x_{it} x_{is}^{\prime} \tilde{\varepsilon}_{it} \tilde{\varepsilon}_{is}
\end{align}
and 
\begin{align}
    A_{2}\coloneqq\frac{N_{g}}{\widehat{N}_{g}} \frac{1}{N_{g} T} \sum_{i=1}^{N} \sum_{t=1}^{T} \sum_{s=1}^{T} 1\left\{\hat{g}_{i}=g, g_{i}^{0} \neq g\right\} x_{it} x_{is}^{\prime} \hat{\varepsilon}_{it} \hat{\varepsilon}_{is}.
\end{align}

For $A_{1}$, we have
\begin{align}
A_{1}=(1+o_p(1)) & \frac{1}{N_{g} T} \sum_{i=1}^{N} \sum_{t=1}^{T} \sum_{s=1}^{T} 1\left\{\hat{g}_{i} \neq g, g_{i}^{0}=g\right\} x_{it} x_{is}^{\prime}\left(\varepsilon_{it}-x_{it}^{\prime}\left(\hat{\beta}_{g_i^0}-\beta_{g_i^0}^0\right)\right)\left(\varepsilon_{is}-x_{is}^{\prime}\left(\hat{\beta}_{g_i^0}-\beta_{g_{i}^{0}}^{0}\right)\right) \\
=(1+o_p(1)) &\left\{ \frac{1}{N_{g} T} \sum_{i=1}^{N} \sum_{t=1}^{T} \sum_{s=1}^{T} 1\left\{\hat{g}_{i} \neq g, g_{i}^{0}=g\right\} x_{it} x_{is}^{\prime} \varepsilon_{it} \varepsilon_{is} \right. \\
& \left. \qquad -\frac{2}{N_{g} T} \sum_{i=1}^{N} \sum_{t=1}^{T} \sum_{s=1}^{T} 1\left\{\hat{g}_{i} \neq g, g_{i}^{0}=g\right\} x_{it} x_{is}^{\prime} \varepsilon_{it} x_{is}^{\prime}\left(\hat{\beta}_{g}-\beta_{g}^{0}\right) \right. \\
& \left. \qquad +\frac{1}{N_{g} T} \sum_{i=1}^{N} \sum_{t=1}^{T} \sum_{s=1}^{T} 1\left\{\hat{g}_{i} \neq g, g_{i}^{0}=g\right\} x_{it} x_{is}^{\prime} x_{it}^{\prime}\left(\hat{\beta}_g-\beta_{g}^{0}\right)\left(\hat{\beta}_g-\beta_{g}^{0}\right)^{\prime} x_{is}\right\}.
\end{align}

Using the H\"{o}lder inequality, the first term in the curly brackets is bounded as
\begin{align}
    &\left\|\frac{1}{N_g T} \sum_{i=1}^{N} \sum_{t=1}^{T} \sum_{s=1}^{T} 1\left\{\hat{g}_{i} \neq g, g_{i}^{\circ}=g\right\} x_{it} x_{is}^{\prime} \varepsilon_{it} \varepsilon_{is}\right\| \\
    & \leq \frac{1}{N_g} \sum_{i=1}^{N} 1\left\{\hat{g}_{i} \neq g, g_{i}^{0}=g\right\}\left\|\frac{1}{T} \sum_{t=1}^{T} \sum_{s=1}^{T} x_{it} x_{is}^{\prime} \varepsilon_{it} \varepsilon_{is}\right\| \\
& \leq\left(\frac{1}{N_g} \sum_{i=1}^{N} 1\left\{\hat{g}_{i} \neq g, g_{i}^{0}=g\right\}\right)^{\frac{1}{1+q}}\left(\frac{1}{N_g} \sum_{i=1}^{N}\left\|\frac{1}{\sqrt{T}} \sum_{t=1}^{T} x_{it} \varepsilon_{it}\right\|^{\frac{2(1+q)}{q}}\right)^{\frac{q}{1+q}} \\
& =o_p(1/T),
\end{align}
where the last probability order follows from condition (b) and Lemma 3 of Hansen (2007), which is applicable under conditions (c)--(e). Using a similar argument coupled with Assumption 1, the second and third terms in the curly brackets are $o_p(1 / \sqrt{N T})$ and $o_p(1 /(N T))$, respectively. This shows $A_{1}=o_p(1)$. Next, we decompose $A_{2}$ as
\begin{align}
A_{2}&=(1+o_p(1)) \frac{1}{N_{g} T} \sum_{h \neq g} \sum_{i=1}^{N} \sum_{t=1}^{T} \sum_{s=1}^{T} 1\left\{\hat{g}_{i}=g, g_{i}^{0}=h\right\} x_{it} x_{is}^{\prime}\left(\varepsilon_{it}-x_{it}^{\prime}\left(\hat{\beta}_{g}-\beta_{h}^{0}\right)\right)\left(\varepsilon_{is}-x_{is}^{\prime}\left(\hat{\beta}_{g}-\beta_{h}^{0}\right)\right) \\
&=(1+o_p(1))\left\{\frac{1}{N_{g} T} \sum_{h \neq g} \sum_{i=1}^{N} \sum_{t=1}^{T} \sum_{s=1}^{T} 1\left\{\hat{g}_{i}=g, g_{i}^{0}=h\right\} x_{it} x_{is}^{\prime} \varepsilon_{it} \varepsilon_{is} \right. \\
& \left. \hspace{3cm} -\frac{2}{N_{g} T} \sum_{h \neq g} \sum_{i=1}^{N} \sum_{t=1}^{T} \sum_{s=1}^{T} 1\left\{\hat{g}_{i}=g, g_{i}^{0}=h\right\} x_{it} x_{is}^{\prime} \varepsilon_{it} x_{is}^{\prime}\left(\hat{\beta}_{g}-\beta_{h}^{0}\right) \right. \\
& \left. \hspace{3cm} +\frac{1}{N_{g} T} \sum_{h \neq g} \sum_{i=1}^{N} \sum_{t=1}^{T} \sum_{s=1}^{T} 1\left\{\hat{g}_{i}=g, g_{i}^{0}=h\right\} x_{it}^{\prime}\left(\hat{\beta}_{g}-\beta_{h}^{0}\right)\left(\hat{\beta}_{g}-\beta_{h}^{0}\right)^{\prime} x_{is}\right\}.
\end{align}
The same argument as used above shows that the first, second, and third terms in the curly brackets are $o_p(1/T)$, $o_p(1)$, and $o_p(1)$, respectively. This yields $A_{2}=o_p(1)$.

Combining the above results with condition (f), we deduce $\widehat{\Omega}_g \xrightarrow{p} \Omega_g$.
\end{proof}

Lemmas \ref{lem:first_deriv_indicator} and \ref{lem:partial_exp} below are useful to establish Proposition \ref{prop:fixed_T}.

\begin{lem}
    \begin{align}
        \frac{\partial}{\partial \theta_{\widetilde{g}}'} E\left[1\left\{\hat{g}_{i}(\boldsymbol{\theta})=g\right\} \mid x_{i}=x\right]=
        \begin{cases}
        \sum_{h \neq g} \int_{S_{g h}} \left\|x\left(\theta_{h}-\theta_{g}\right)\right\|^{-1}\left(y-x \theta_{g}\right)^{\prime} x f(y \mid x) d y \quad \text{ for } \ \tilde{g}=g \\
        -\int_{S_{g \widetilde{g}}}\left\|x\left(\theta_{\widetilde{g}}-\theta_{g}\right)\right\|^{-1} \left(y-x \theta_{\widetilde{g}}\right)^{\prime} x f(y \mid x) d y \quad \text{ for } \ \tilde{g} \neq g
        \end{cases}.
    \end{align}
    \label{lem:first_deriv_indicator}
\end{lem}

\noindent\begin{proof}
Let
\begin{align}
    V_{g}=\left\{y \in \mathbb{R}^{T} ;\left\|y-x \theta_{g}\right\|^{2} \leq\|y-x \theta_{\widetilde{g}}\|^{2} \text{ for all } \tilde{g} \neq g\right\}.
\end{align}
Note that $V_{g}$ is the intersection of ($G-1$) half-spaces in $\mathbb{R}^{T}$ as
\begin{align}
\left\|y-x \theta_{\widetilde{g}}\right\|^{2}-\left\|y-x \theta_{g}\right\|^{2}=  2\left(\theta_{g}-\theta_{\widetilde{g}}\right)^{\prime} x^{\prime}\left\{y-x\left(\theta_{g}-\theta_{\widetilde{g}}\right)\right\}.
\end{align}
Using differential calculus as in \textcite{pollard1982central} and \textcite{bonhomme2015grouped}, we have
\begin{align}
    \frac{\partial}{\partial \theta_{\widetilde{g}}^{\prime}} E\left[1\left\{\hat{g}_{i}(\boldsymbol{\theta})=g\right\} \mid x_{i}=x\right] 
    & =\frac{\partial}{\partial \theta_{\widetilde{g}}^{\prime}} \int_{V_{g}} f(y \mid x) d y \\
& =\int_{\partial V_{g}} f(y \mid x) v_{g}\left(y ; \theta_{\widetilde{g}}\right)^{\prime} d y,
\end{align}
where $\partial V_{g}$ is the boundary of $V_{g}$, and $v_{g}\left(y ; \theta_{\widetilde{g}}\right)=\left(v_{g}\left(y ; \theta_{\widetilde{g} k}\right)\right)$ is the $p$-dimensional velocity vector. Noting that $\partial V_{g}=\bigcup_{h \neq g} S_{g h}$ and thus
\begin{align}
    \int_{\partial V_{g}} f(y \mid x) v_{g}\left(y ; \theta_{\widetilde{g}}\right)^{\prime} d y=\sum_{h \neq g} \int_{S_{g h}} f(y \mid x) v_{g}\left(y ; \theta_{\widetilde{g}}\right)^{\prime} d y,
    \label{eqn:int_decompose}
\end{align}
we compute $v_{g}\left(y ; \theta_{\widetilde{g}}\right)$ on $S_{g h}$.

First consider the case $\tilde{g}=g$. For a given (small) $\xi \in \mathbb{R}$, define
\begin{align}
    \theta_{g}^{*}(\xi)=\theta_{g}+\xi e_{k},
\end{align}
and
\begin{align}
    S_{g h}^{*}(\xi)=\left\{y \in \mathbb{R}^{T}; \left\|y-x \theta_{g}^{*}(\xi)\right\|^{2}=\left\|y-x \theta_{h}\right\|^{2},\text{ and } \left\|y-x \theta_{g}^{*}(\xi)\right\|^{2} \leq\left\|y-x \theta_{\widetilde{g}}\right\|^{2} \text{ for all } \tilde{g} \neq g, h\right\}.
\end{align}
To any given $y \in S_{g h}$, associate the point $y^{*}(\xi) \in S_{g h}^*(\xi)$ such that $y^{*}(\xi)-y$ is orthogonal to the hypersurface $S_{gh}$. The velocity is defined by
\begin{align}
    v_{g}\left(y; \theta_{g k}\right)=\lim _{\xi \rightarrow 0} \frac{\left(y^{*}(\xi)-y\right)^{\prime} \vec{n}}{\xi}, \quad \text{ for } k=1, \ldots, p,
\end{align}
where $\vec{n}$ is the unit normal vector to $S_{gh}$ that points outside of $V_g$. Because $S_{g h}$ is characterized by the equation
\begin{align}
    \frac{1}{2}\left(\left\|y-x \theta_{g}\right\|^{2}-\left\|y-x \theta_{h}\right\|^{2}\right)=\left(\theta_{h}-\theta_{g}\right)^{\prime} x^{\prime}\left(y-x \frac{\theta_{h}+\theta_{g}}{2}\right)=0,
\end{align}
we have
\begin{align}
    \vec{n}=\frac{x\left(\theta_{h}-\theta_{g}\right)}{\left\|x\left(\theta_{h}-\theta_{g}\right)\right\|}.
\end{align}
Moreover, $y^{*}(\xi)$ satisfies
\begin{align}
    y^{*}(\xi)=y+\lambda(\xi) x\left(\theta_{h}-\theta_{g}\right),
\end{align}
where $\lambda(\xi)$ is such that
\begin{align}
    \left(\theta_{h}-\theta_{g}^{*}(\xi)\right)^{\prime} x^{\prime}\left(y^{*}(\xi)-x \frac{\theta_{h}+\theta_{g}^{*}(\xi)}{2}\right)=0.
\end{align}
That is,
\begin{align}
    \left(\theta_{h}-\theta_{g}-\xi e_{k}\right)^{\prime} x^{\prime}\left\{y+\lambda(\xi) x\left(\theta_{h}-\theta_{g}\right)-x \frac{\theta_{h}+\theta_{g}}{2}-\frac{\xi}{2} x e_{k}\right\}=0.
\end{align}
Rearranging terms gives
\begin{align}
    \lambda(\xi)=\xi \frac{\left(y-x \theta_{g}\right)^{\prime} x e_{k}}{\left\|x\left(\theta_h-\theta_{g}\right)\right\|^{2}}+o(\xi). \label{eqn:lambda_xi}
\end{align}
It thus follows that
\begin{align}
    v_{g}\left(y ; \theta_{g k}\right) 
    & =\lim _{\xi \rightarrow 0} \frac{\lambda(\xi)\left(\theta_{h}-\theta_{g}\right)^{\prime} x^{\prime}\left\{x\left(\theta_{h}-\theta_{g}\right)/\left\|x\left(\theta_{h}-\theta_{g}\right)\right\|\right\}}{\xi} \\
& =\lim _{\xi \rightarrow 0} \frac{\xi \{\left(y-x \theta_{g}\right)^{\prime} x e_{k}/\left\|x\left(\theta_{h}-\theta_{g}\right)\right\|^{2}\}\left\|x\left(\theta_{h}-\theta_{g}\right)\right\|+o(\xi)}{\xi} \\
& =\frac{\left(y-x \theta_{g}\right)^{\prime} x e_{k}}{\left\|x\left(\theta_{h}-\theta_{g}\right)\right\|}.
\end{align}
This gives
\begin{align}
    v_{g}\left(y ; \theta_{g}\right)^{\prime}=\frac{\left(y-x \theta_{g}\right)^{\prime} x}{\left\|x\left(\theta_{h}-\theta_{g}\right)\right\|}.
\end{align}

Combining the above results, we get
\begin{align}
    \int_{\partial V_{g}} f(y \mid x) v_{g}\left(y ; \theta_{\widetilde{g}}\right)^{\prime} d y=\sum_{h \neq g} \int_{S_{g h}} \frac{\left(y-x \theta_{g}\right)^{\prime} x}{\left\|x\left(\theta_{h}-\theta_{g}\right)\right\|} f(y \mid x) d y.
\end{align}
This yields the first result of Lemma \ref{lem:first_deriv_indicator}.

We next compute $v_{g}\left(y ; \theta_{\widetilde{g}}\right)$ for $\tilde{g} \neq g$. Again invoking \eqref{eqn:int_decompose}, we do this on $S_{g h}$ for each $h \neq g$. If $\tilde{g} \neq h$, then $\lambda(\xi)=0$ so $v_{g}\left(y, \theta_{\widetilde{g}}\right)=0$. If $\tilde{g}=h$, then $y^{*}(\xi) \in S_{gh}$ solves
\begin{align}
    \left(\theta_{h}^{*}(\xi)-\theta_{g}\right)^{\prime} x^{\prime}\left\{y^{*}(\xi)-x \frac{\theta_{h}^{*}(\xi)+\theta_{g}}{2}\right\}=0,
\end{align}
where $\theta_h^{*}(\xi)=\theta_h+\xi e_{k}$. That is
\begin{align}
    &\left(\theta_{h}-\theta_{g}+\xi e_{k}\right)^{\prime} x^{\prime}\left(y-x \frac{\theta_{h}+\theta_{g}}{2}+\lambda(\xi) x\left(\theta_{h}-\theta_{g}\right)-\frac{\xi}{2} x e_{k}\right)=0, \\
    \intertext{so}
    &\lambda(\xi)=-\xi \frac{\left(y-x \theta_{h}\right)^{\prime} x e_{k}}{\left\|x\left(\theta_{h}-\theta_{g}\right)\right\|^{2}}+o(\xi).
\end{align}

Following the same steps as used after \eqref{eqn:lambda_xi} yields the second result.
\end{proof}

The following result is an immediate consequence of this lemma.

\begin{lem}
\begin{align}
    &\left.\frac{\partial}{\partial \theta_{\widetilde{g}}^{\prime}}\right|_{\boldsymbol{\theta}=\overline{\boldsymbol{\theta}}} E\left[1 \left\{ \hat{g}_{i}(\boldsymbol{\theta})=g\right\}x_{i}^{\prime}\left(y_{i}-x_{i} \overline{\theta}_{g}\right) \mid x_{i}=x\right] \\
   & = \begin{cases}\sum_{h \neq g} \int_{\overline{S}_{g h}} \left\{x^{\prime}\left(y-x\overline{\theta}_{g}\right)\left(y-x \overline{\theta}_{g}\right)^{\prime} x/\left\|x\left(\overline{\theta}_{h}-\overline{\theta}_{g}\right)\right\|\right\} f(y \mid x) d y & \text{ for } \ \tilde{g}=g \\
   -\int_{\overline{S}_{g \widetilde{g}}} \left\{x^{\prime}\left(y-x \overline{\theta}_{g}\right)\left(y-x \overline{\theta}_{\widetilde{g}}\right)^{\prime} x/\left\|x\left(\overline{\theta}_{\widetilde{g}}-\overline{\theta}_{g}\right)\right\|\right\} f(y \mid x) d y & \text{ for } \ \tilde{g} \neq g\end{cases}.
\end{align}
\label{lem:partial_exp}
\end{lem}

Now we are ready to prove Proposition \ref{prop:fixed_T}.

\noindent\begin{proof}[Proof of Proposition \ref{prop:fixed_T}]
    If $g=\tilde{g}$, then
    \begin{align}
        \Gamma_{g g} 
        & =-\left.\frac{\partial}{\partial \theta_{g}^{\prime}}\right|_{\boldsymbol{\theta}=\overline{\boldsymbol{\theta}}} E\left[1 \left\{\hat{g}_{i}(\boldsymbol{\theta})=g\right\}x_{i}^{\prime}\left(y_{i}-x_{i} \theta_{g}\right)\right] \\
        & =E\left[x_{i}^{\prime} x_{i} 1\left\{\hat{g}_{i}(\overline{\boldsymbol{\theta}})=g\right\}\right]  -E\left[\left.\frac{\partial}{\partial \theta_{g}^{\prime}} \right|_{\boldsymbol{\theta}=\overline{\boldsymbol{\theta}}} E\left[1\left\{\hat{g}_{i}(\boldsymbol{\theta})=g\right\} x_{i}^{\prime}\left(y_{i}-x_{i} \overline{\theta}_{g}\right) \mid x_{i}\right]\right] \\
        & =E\left[x_{i}^{\prime} x_{i} 1\left\{\hat{g}_{i}(\overline{\boldsymbol{\theta}})=g\right\}\right]  -E\left[\sum_{h \neq g} \int_{\overline{S}_{g h}} \frac{x_{i}^{\prime}\left(y-x_{i} \overline{\theta}_{g}\right)\left(y-x_{i} \overline{\theta}_{g}\right)^{\prime} x_{i}}{\left\|x_{i}\left(\overline{\theta}_{h}-\overline{\theta}_{g}\right)\right\|} f\left(y\mid x_{i}\right) d y\right],
    \end{align}   
    where the last equality follows from Lemma \ref{lem:partial_exp}.

Similarly, for $\tilde{g} \neq g$, we have
\begin{align}
    \Gamma_{g \widetilde{g}} 
    & =-\left.\frac{\partial}{\partial \theta_{\widetilde{g}}^{\prime}}\right|_{\boldsymbol{\theta}=\overline{\boldsymbol{\theta}}} E\left[1\left\{\hat{g}_{i}(\boldsymbol{\theta})=g\right\} x_{i}^{\prime}\left(y_{i}-x_{i} \theta_g\right)\right] \\
    & =-E\left[\left.\frac{\partial}{\partial \theta_{\widetilde{g}}^{\prime}}\right|_{\boldsymbol{\theta}=\overline{\boldsymbol{\theta}}} E\left[1\left\{\hat{g}_{i}(\boldsymbol{\theta})=g\right\} x_{i}^{\prime}\left(y_{i}-x_{i} \overline{\theta}_{g}\right) \mid x_{i}\right]\right] \\
    & =E\left[\int_{\overline{S}_{g \widetilde{g}}} \frac{x_{i}^{\prime}\left(y-x_{i} \overline{\theta}_{g}\right)\left(y-x_i \overline{\theta}_{\widetilde{g}}\right)^{\prime} x_{i}}{\left\|x_{i}\left(\overline{\theta}_{\widetilde{g}} - \overline{\theta}_{g}\right) \right\|} f\left(y\mid x_{i}\right) d y\right].
\end{align}
\end{proof}

\section{Joint coverage by confidence intervals and confidence sets}
\label{appendix: joint_coverage}

Theorems \ref{thm:coverage_um} and \ref{thm:coverage_bc} guarantee unit-wise (marginal) correct coverage by the bias-corrected and minimum-statistic-based confidence sets. The joint coverage across $N$ units can also be established under the assumption of joint coverage by the first step confidence set, $\mathrm{CS}_{i,\alpha}$, for group membership and a suitably adjusted nominal level for the second step confidence set construction. Specifically, we assume the following condition.

\begin{asm}\label{asm:cs_membership_joint}
    \begin{align}
        \liminf_{N,T\to\infty}P\left(\bigcap_{i=1}^N \{g_i^0\in \mathrm{CS}_{i, \alpha}\}\right)\geq 1-\alpha.
    \end{align}
\end{asm}
\textcite{dzemski2024confidence} discuss sufficient conditions under which their confidence set for group membership satisfies Assumption \ref{asm:cs_membership_joint}. In particular, the critical value, $c_{i,\alpha}(g)$, used to construct $\mathrm{CS}_{i,\alpha}$ needs to be the $(1-\alpha/N)$-quantile of a suitable distribution.

\begin{thm}\label{thm:joint_coverage}
    Suppose Assumptions \ref{asm:asym_normal}--\ref{asm:vcov_consistent} and \ref{asm:cs_membership_joint} hold. Suppose $\mathrm{CI}^{\mathrm{BC}}(i,\alpha)$, $\mathrm{CS}^{\mathrm{BC}}(i,\alpha)$, $\mathrm{CI}^{\mathrm{MT}}(i,\alpha)$, and $\mathrm{CS}^{\mathrm{MW}}(i,\alpha)$ are defined as in the paper but calculated using $\alpha_2/G$ in place of $\alpha_2$.

    When $r=1$, we have
    \begin{itemize}
        \item[(i)] $\liminf_{N,T\to\infty}P(\bigcap_{i=1}^N\{R\theta_{g_i^0}^{0} \in \mathrm{CI}^{\mathrm{BC}}(i, \alpha)\}) \geq 1 - \alpha$,

        \item[(ii)] $\liminf_{N,T\to\infty}P(\bigcap_{i=1}^N\{R\theta_{g_i^0}^{0} \in \mathrm{CI}^{\mathrm{MT}}(i, \alpha)\}) \geq 1 - \alpha$.
    \end{itemize}
    
    When $r\geq2$, we have
    \begin{itemize}
        \item[(iii)] $\liminf_{N,T\to\infty}P(\bigcap_{i=1}^N\{R\theta_{g_{i}^{0}}^{0} \in \mathrm{CS}^{\mathrm{BC}}(i, \alpha)\}) \geq 1 - \alpha$,

        \item[(iv)] $\liminf_{N,T\to\infty}P(\bigcap_{i=1}^N\{R\theta_{g_{i}^{0}}^{0} \in\mathrm{CS}^{\mathrm{MW}}(i, \alpha)\}) \geq 1 - \alpha$.
    \end{itemize}
\end{thm}

\noindent\begin{proof}[Proof of Theorem \ref{thm:joint_coverage}]
    We prove part (i) only to save space. The proofs for the other results are similar. As in the proof of Theorem \ref{thm:coverage_bc}(i), it suffices to show
    \begin{align}
        \limsup_{N,T\to\infty}P\left(\bigcup_{i=1}^N\left\{R\theta_{g_i^0}^0\notin \mathrm{CI}^{\mathrm{BC}}(i,\alpha)\right\}, \ \bigcap_{i=1}^N\left\{g_i^0\in \mathrm{CS}_{i,\alpha_1}\right\}\right) \leq \alpha_2.
    \end{align}

    We have
    \begin{align}
        &P\left(\bigcup_{i=1}^N\left\{R\theta_{g_i^0}^0\notin \mathrm{CI}^{\mathrm{BC}}(i,\alpha)\right\}, \ \bigcap_{i=1}^N\left\{g_i^0\in \mathrm{CS}_{i,\alpha_1}\right\}\right) \\
        &\leq P\left(\bigcup_{i=1}^N\left\{R\theta_{g_i^0}^0> \max_{g\in\mathrm{CS}_{i,\alpha_1}} R\hat{\theta}_g + z_{1-\alpha_2/(2G)}\times \mathrm{se}_i, \ g_i^0\in \mathrm{CS}_{i,\alpha_1} \right\}\right) \\
        &\qquad +P\left(\bigcup_{i=1}^N\left\{R\theta_{g_i^0}^0< \min_{g\in\mathrm{CS}_{i,\alpha_1}} R\hat{\theta}_g - z_{1-\alpha_2/(2G)}\times \mathrm{se}_i, \ g_i^0\in \mathrm{CS}_{i,\alpha_1} \right\}\right) \\
        &\leq P\left(\bigcup_{i=1}^N\left\{\frac{R\left(\hat{\theta}_{g_i^0}-\theta_{g_i^0}^0\right)}{\mathrm{se}(g_i^0)}< -z_{1-\alpha_2/(2G)}\right\}\right) + P\left(\bigcup_{i=1}^N\left\{\frac{R\left(\hat{\theta}_{g_i^0}-\theta_{g_i^0}^0\right)}{\mathrm{se}(g_i^0)}> z_{1-\alpha_2/(2G)}\right\}\right) \\
        &=P\left(\bigcup_{g=1}^G\left\{\frac{R\left(\hat{\theta}_{g}-\theta_{g}^0\right)}{\mathrm{se}(g)}< -z_{1-\alpha_2/(2G)}\right\}\right) + P\left(\bigcup_{g=1}^G\left\{\frac{R\left(\hat{\theta}_{g}-\theta_{g}^0\right)}{\mathrm{se}(g)}> z_{1-\alpha_2/(2G)}\right\}\right) \\
        &\leq \sum_{g=1}^G P\left(\left\{\frac{R\left(\hat{\theta}_{g}-\theta_{g}^0\right)}{\mathrm{se}(g)}< -z_{1-\alpha_2/(2G)} \text{ or } \frac{R\left(\hat{\theta}_{g}-\theta_{g}^0\right)}{\mathrm{se}(g)}> z_{1-\alpha_2/(2G)} \right\}\right) \\
        &\leq \alpha_2 + o(1).
    \end{align}
    This proves part (i).
\end{proof}

\section{Additional simulation results}\label{appendix:simulation}

In this appendix, we present the simulation results for $T=40$.

Figure \ref{fig:coverage_40} gives the coverage properties for the confidence intervals. When $N=40$, the unit-by-unit confidence interval attains the nominal 90\% coverage rate, while the naive confidence interval suffers from severe undercoverage, particularly when $\sigma_i$ is moderate to large. The bias-corrected confidence interval can undercover for small $\sigma_i$ but tends to be conservative for moderate to large $\sigma_i$, while the minimum-type confidence interval can suffer from severe undercoverage for relatively small $\sigma_i$. For larger $N$, the pattern is similar to those observed for $T=80$: the unit-by-unit confidence set has correct coverage, the naive confidence set severely undercovers except for small $\sigma_i$, the bias-corrected confidence set tends to be conservative, and the minimum-type confidence interval attains nominal coverage or can be conservative. Similar comments apply to the coverage properties of the confidence sets (see Figure \ref{fig:coverage_40_vec}).

\begin{figure}[ht]
	\centering
	    \begin{subfigure}{0.64\columnwidth}
        \includegraphics[width=\columnwidth]{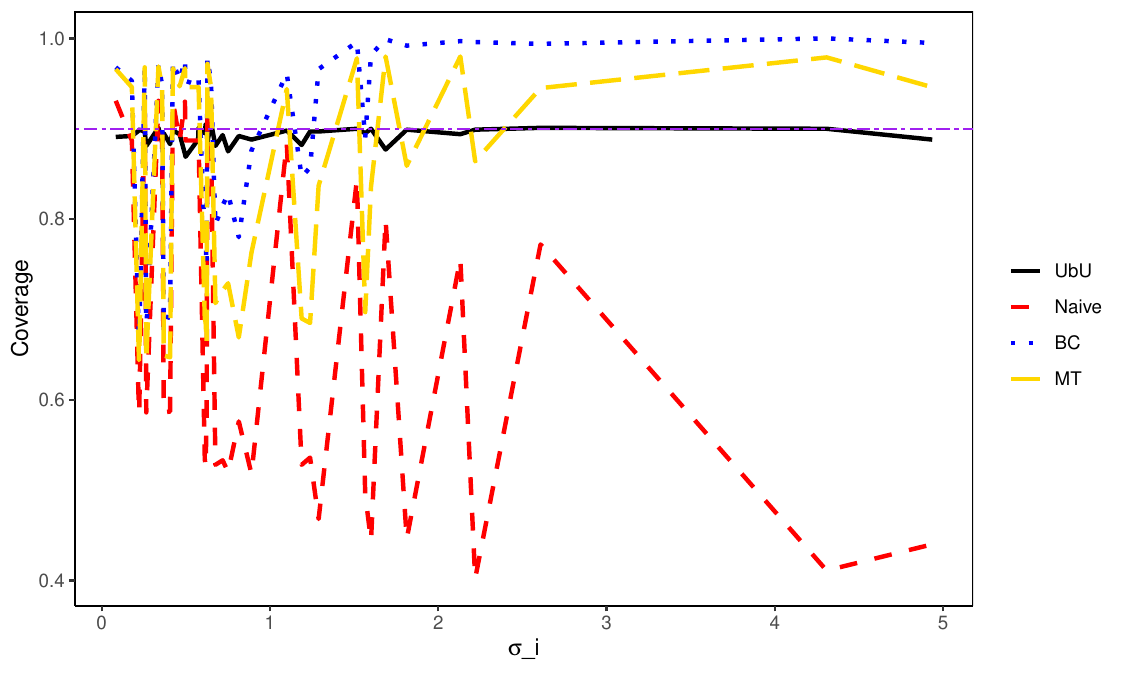}
    \caption{$(N,T)=(40,40)$}
	\label{fig:coverage_40_40}
    \end{subfigure} \quad
    \begin{subfigure}{0.64\columnwidth}
        \includegraphics[width=\columnwidth]{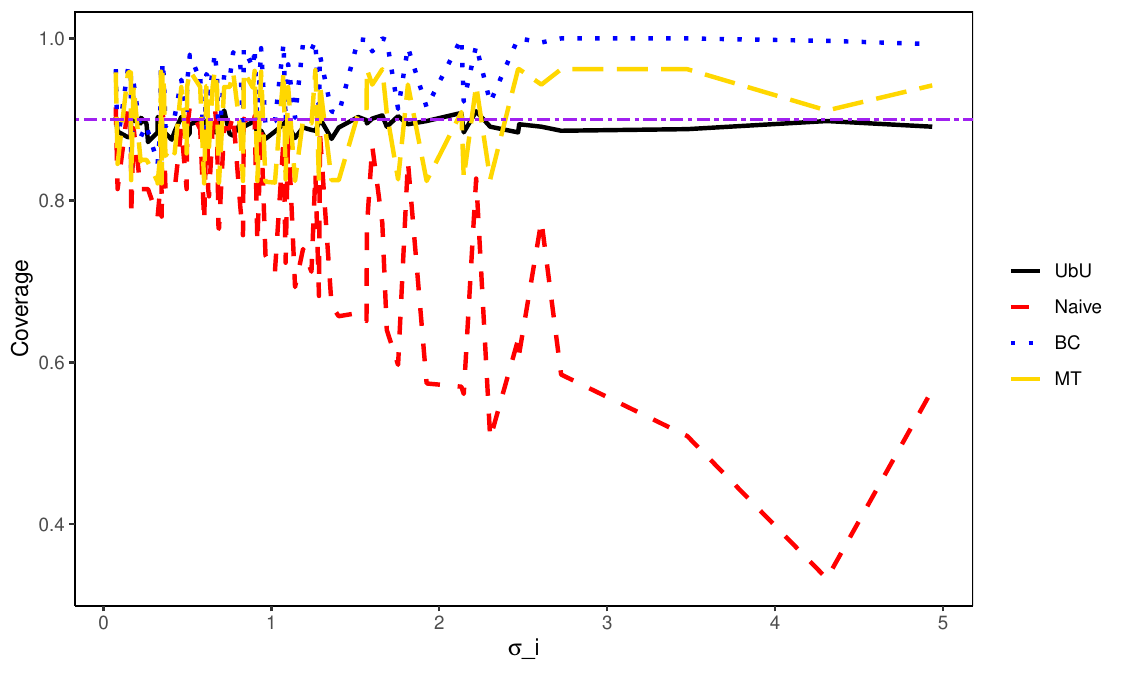}
    \caption{$(N,T)=(80,40)$}
	\label{fig:coverage_80_40}
    \end{subfigure} \quad
    \hfill
    \begin{subfigure}{0.64\columnwidth}
        \includegraphics[width=\columnwidth]{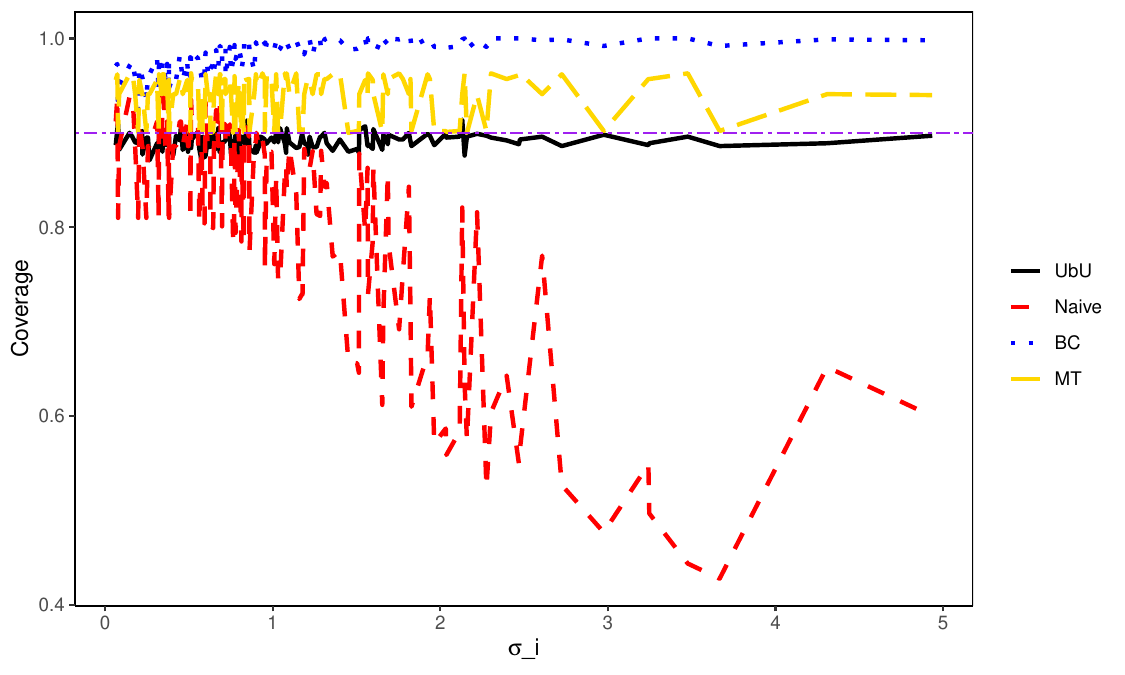}
    \caption{$(N,T)=(160,40)$}
	\label{fig:coverage_160_40}
    \end{subfigure} \quad
    \caption{Unit-wise coverage rates of confidence intervals, $T=40$}
    \label{fig:coverage_40}
    \caption*{\full: Unit-by-unit (UbU), \quad \color{red}\dashed\color{black}: Naive, \quad \color{black}\dotted\color{black}: Bias-corrected (BC), \quad \color{gold}\longdash\color{black}: Minimum t (MT)}
\end{figure}

\begin{figure}[ht]
	\centering
	    \begin{subfigure}{0.64\columnwidth}
        \includegraphics[width=\columnwidth]{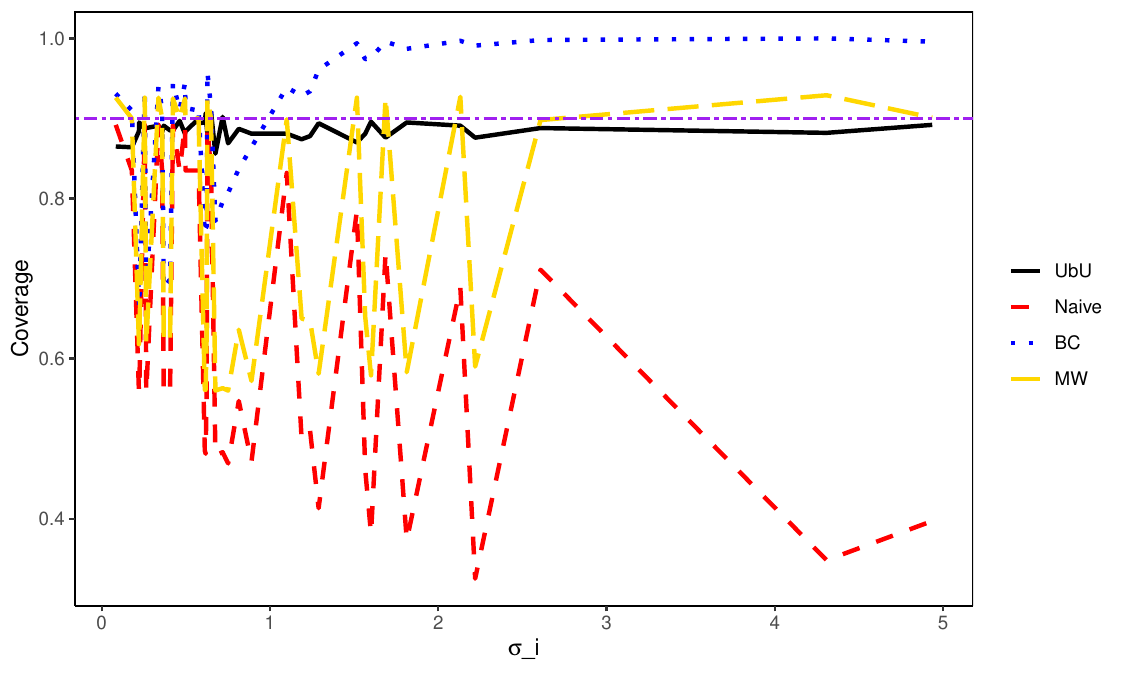}
    \caption{$(N,T)=(40,40)$}
	\label{fig:coverage_40_40_vec}
    \end{subfigure} \quad
    \begin{subfigure}{0.64\columnwidth}
        \includegraphics[width=\columnwidth]{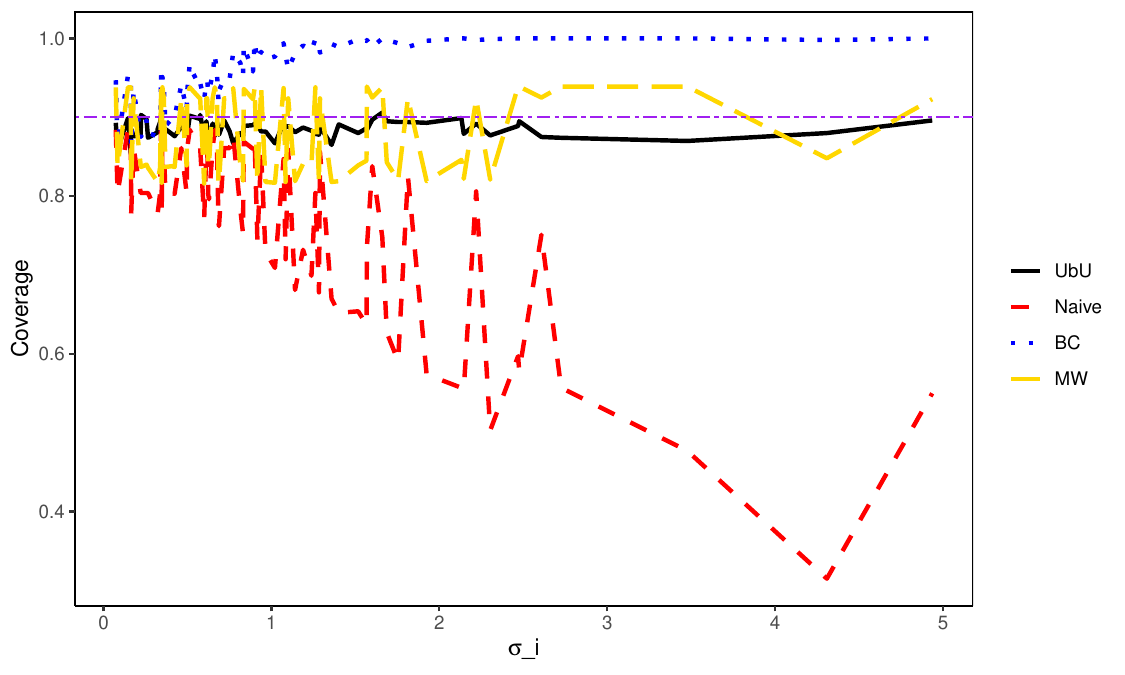}
    \caption{$(N,T)=(80,40)$}
	\label{fig:coverage_80_40_vec}
    \end{subfigure} \quad
    \hfill
    \begin{subfigure}{0.64\columnwidth}
        \includegraphics[width=\columnwidth]{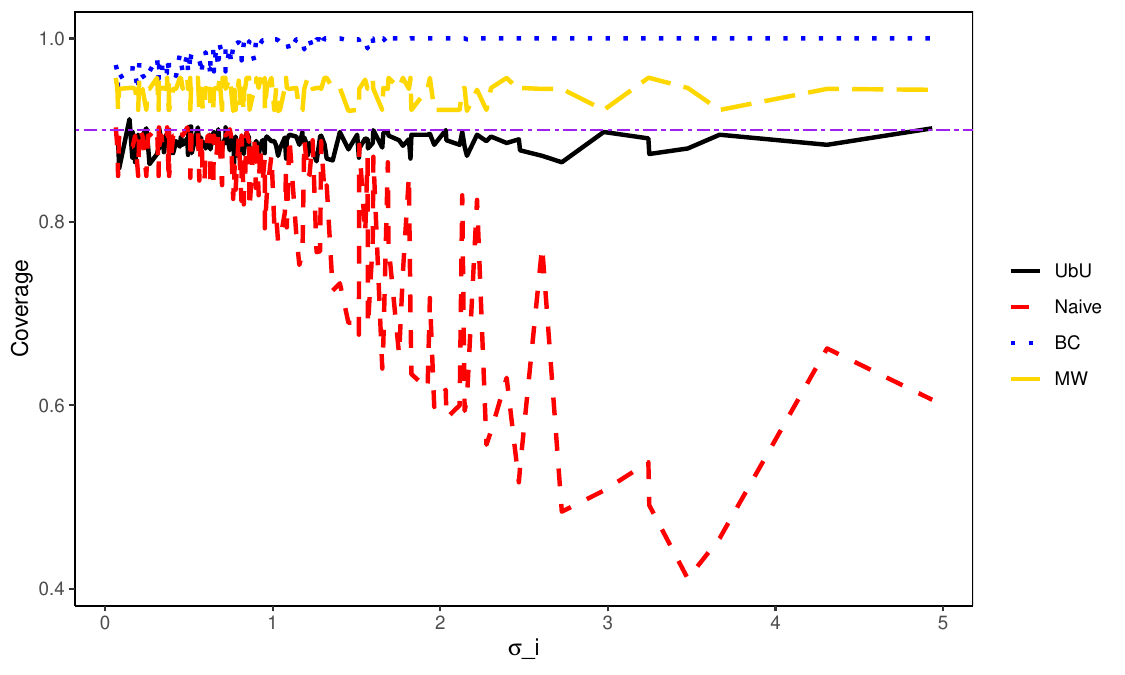}
    \caption{$(N,T)=(160,40)$}
	\label{fig:coverage_160_40_vec}
    \end{subfigure} \quad
    \caption{Unit-wise coverage rates of confidence sets, $T=40$}
    \label{fig:coverage_40_vec}
    \caption*{\full: Unit-by-unit (UbU), \quad \color{red}\dashed\color{black}: Naive, \quad \color{black}\dotted\color{black}: Bias-corrected (BC), \quad \color{gold}\longdash\color{black}: Minimum Wald (MW)}
\end{figure}

Next, let us evaluate the power properties, which are shown in Figure \ref{fig:length_40}. When $N=40$, the naive confidence interval performs best in terms of mean length except for small $\sigma_i$, while the unit-by-unit confidence interval is short only for small $\sigma_i$. The bias-corrected confidence interval is the longest among the considered confidence intervals, and the minimum-type confidence interval is long on average. As $N$ increases, the naive confidence interval remains the shortest, and the proposed methods become shorter than in the $N=40$ case. In particular, the minimum-type confidence interval is the second shortest except for very small $\sigma_i$ and outperforms the unit-by-unit confidence interval.

\begin{figure}[ht]
	\centering
	    \begin{subfigure}{0.64\columnwidth}
        \includegraphics[width=\columnwidth]{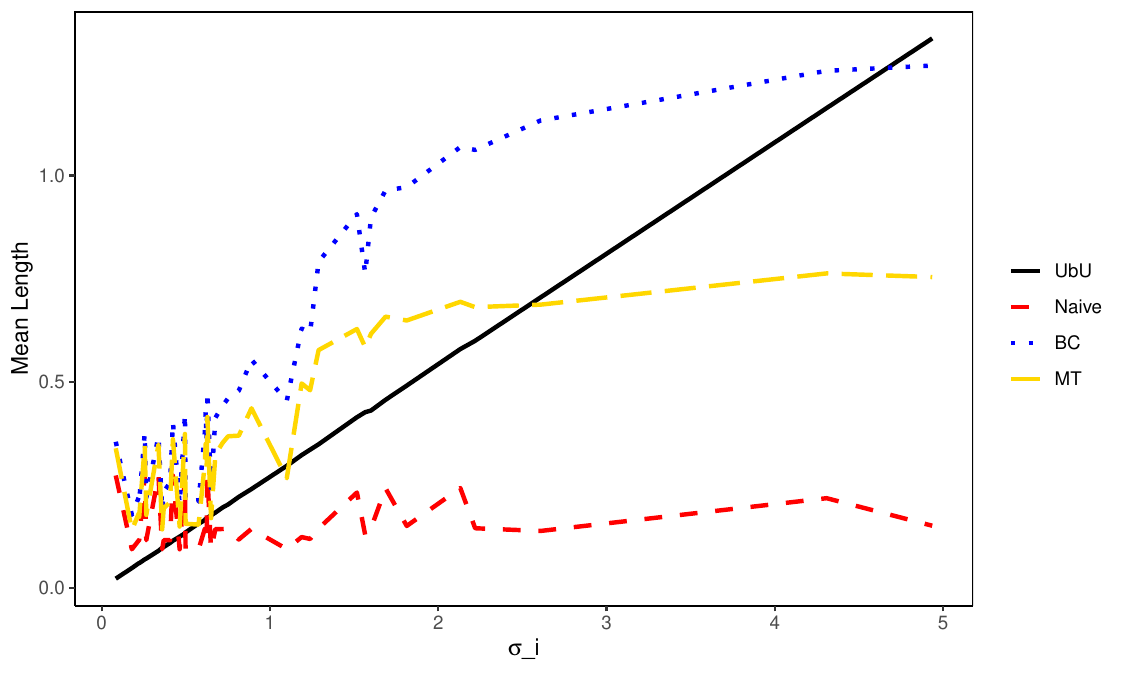}
    \caption{$(N,T)=(40,40)$}
	\label{fig:length_40_40}
    \end{subfigure} \quad
    \begin{subfigure}{0.64\columnwidth}
        \includegraphics[width=\columnwidth]{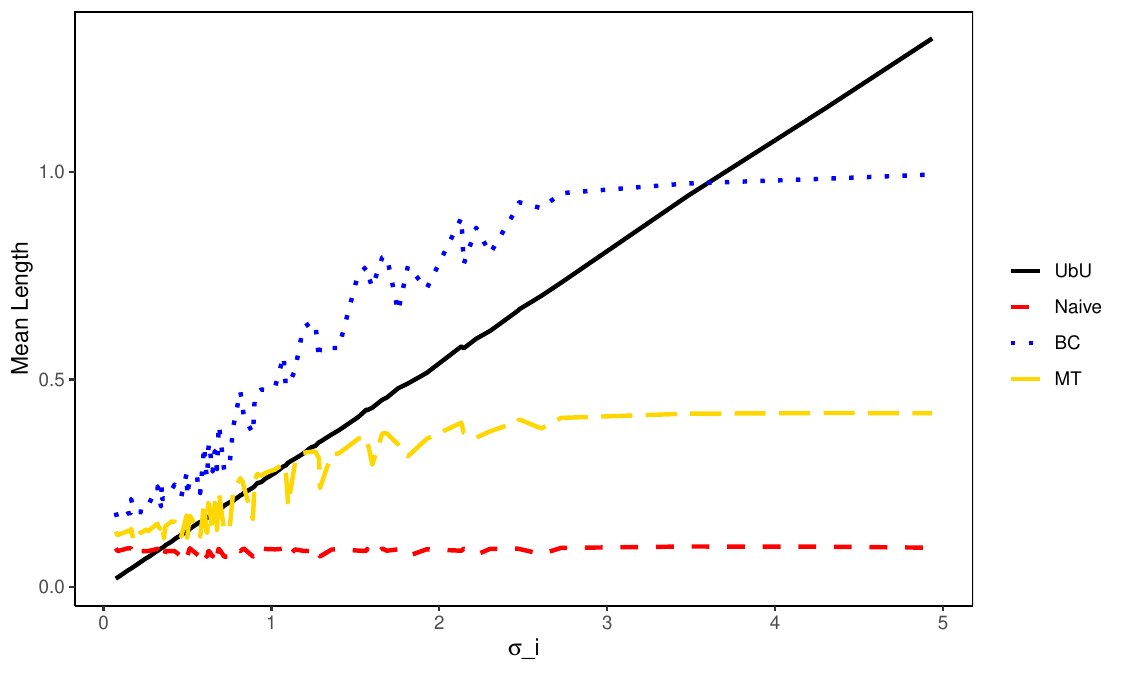}
    \caption{$(N,T)=(80,40)$}
	\label{fig:length_80_40}
    \end{subfigure} \quad
    \hfill
    \begin{subfigure}{0.64\columnwidth}
        \includegraphics[width=\columnwidth]{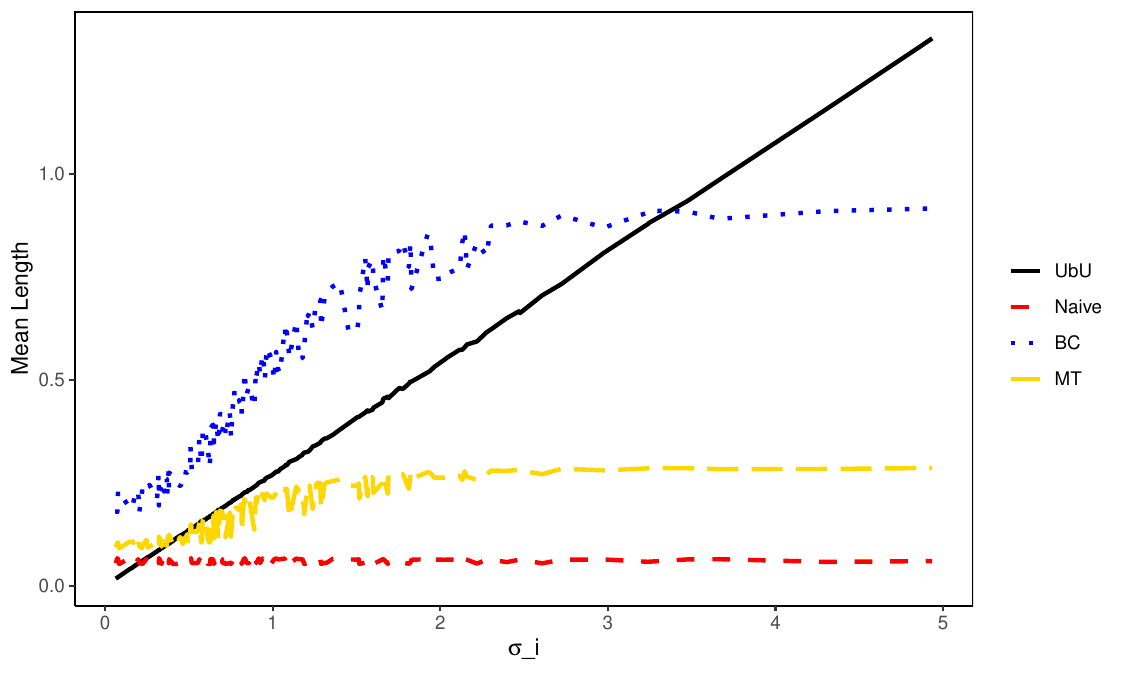}
    \caption{$(N,T)=(160,40)$}
	\label{fig:length_160_40}
    \end{subfigure} \quad
    \caption{Unit-wise mean lengths of confidence intervals, $T=40$}
    \label{fig:length_40}
    \caption*{\full: Unit-by-unit (UbU), \quad \color{red}\dashed\color{black}: Naive, \quad \color{black}\dotted\color{black}: Bias-corrected (BC), \quad \color{gold}\longdash\color{black}: Minimum t (MT)}
\end{figure}

\printbibliography

@unpublished{kaji2025controlling,
  author    = {Tetsuya Kaji and Hyungjune Kang},
  title     = {Controlling Tail Risk Measures with Estimation Error},
  year      = {2025},
  note      = {Working paper},
}

@unpublished{beyhum2024inference,
  title   = {Inference after Discretizing Time-Varying Unobserved Heterogeneity},
  author  = {Beyhum, Jad and Mugnier, Martin},
  year    = {2024},
  note    = {arXiv preprint arXiv:2412.07352}
}

@article{BergerBoos1994,
  author  = {Berger, Roger L. and Boos, Dennis D.},
  title   = {P Values Maximized Over a Confidence Set for the Nuisance Parameter},
  journal = {Journal of the American Statistical Association},
  year    = {1994},
  volume  = {89},
  number  = {427},
  pages   = {1012--1016},
  doi     = {10.1080/01621459.1994.10476836}
}

@article{Dufour1990,
  author  = {Dufour, Jean-Marie},
  title   = {Exact Tests and Confidence Sets in Linear Regressions with Autocorrelated Errors},
  journal = {Econometrica},
  year    = {1990},
  volume  = {58},
  number  = {2},
  pages   = {475--494}
}

@unpublished{ArmstrongWeidnerZeleneev2025,
  author = {Armstrong, Timothy B. and Weidner, Martin and Zeleneev, Andrei},
  title  = {Robust Estimation and Inference in Panels with Interactive Fixed Effects},
  year   = {2025},
  note   = {arXiv preprint arXiv:2210.06639v4}
}

@article{Silvapulle1996,
  author  = {Silvapulle, Mervyn J.},
  title   = {A Test in the Presence of Nuisance Parameters},
  journal = {Journal of the American Statistical Association},
  year    = {1996},
  volume  = {91},
  number  = {436},
  pages   = {1690--1693},
  doi     = {10.1080/01621459.1996.10476739}
}

@article{DiTraglia2016,
  author  = {DiTraglia, Francis J.},
  title   = {Using Invalid Instruments on Purpose: Focused Moment Selection and Averaging for {GMM}},
  journal = {Journal of Econometrics},
  year    = {2016},
  volume  = {195},
  number  = {2},
  pages   = {187--208},
  doi     = {10.1016/j.jeconom.2016.07.006}
}

@article{bruhin2010risk,
  title = {Risk and rationality: Uncovering heterogeneity in probability distortion},
  author = {Bruhin, Adrian and Fehr-Duda, Helga and Epper, Thomas},
  journal = {Econometrica},
  volume = {78},
  number = {4},
  pages = {1375--1412},
  year = {2010},
  publisher = {Wiley Online Library},
  doi = {10.3982/ECTA7139}
}

@article{conte2011mixture,
  title = {Mixture models of choice under risk},
  author = {Conte, Anna and Hey, John D. and Moffatt, Peter G.},
  journal = {Journal of Econometrics},
  volume = {162},
  number = {1},
  pages = {79--88},
  year = {2011},
  publisher = {Elsevier},
  doi = {10.1016/j.jeconom.2009.10.011}
}

@online{choi2024latentgroupstructurelinear,
      title={Latent group structure in linear panel data models with endogenous regressors}, 
      author={Junho Choi and Ryo Okui},
      year={2024},
      eprint={2405.08687},
      archivePrefix={arXiv},
      primaryClass={econ.EM},
      url={https://arxiv.org/abs/2405.08687}, 
}

@online{akgun2025robustinferencemethodslatent,
      title={Robust Inference Methods for Latent Group Panel Models under Possible Group Non-Separation}, 
      author={Oguzhan Akgun and Ryo Okui},
      year={2025},
      eprint={2511.18550},
      archivePrefix={arXiv},
      primaryClass={econ.EM},
      url={https://arxiv.org/abs/2511.18550}, 
}

@online{wan2025conditionalselectiveinferenceselected,
      title={Conditional Selective Inference for the Selected Groups in Panel Data}, 
      author={Chuang Wan and Jiajun Sun and Xingbai Xu},
      year={2025},
      eprint={2511.04466},
      archivePrefix={arXiv},
      primaryClass={stat.ME},
      url={https://arxiv.org/abs/2511.04466}, 
}

@techreport{pionati2025latent,
  title={Latent grouped structures in panel data: A review},
  author={Pionati, Alessandro},
  institution={Munich Personal RePEc Archive},
  type={MPRA Paper},
  number={123954},
  year={2025},
  note={Available at \url{https://mpra.ub.uni-muenchen.de/123954/}}
}

@article{yu2024spectral,
  title={Spectral clustering with variance information for group structure estimation in panel data},
  author={Yu, Lu and Gu, Jiaying and Volgushev, Stanislav},
  journal={Journal of Econometrics},
  volume={241},
  number={1},
  pages={105709},
  year={2024},
  publisher={Elsevier}
}

@article{mugnier2025simple,
  title={A simple and computationally trivial estimator for grouped fixed effects models},
  author={Mugnier, Martin},
  journal={Journal of Econometrics},
  volume={250},
  pages={106011},
  year={2025},
  publisher={Elsevier}
}

@article{sarafidis2015partially,
	author = {Sarafidis, Vasilis and Weber, Neville},
	journal = {Oxford Bulletin of Economics and Statistics},
	number = {2},
	pages = {274--296},
	publisher = {Wiley Online Library},
	title = {A partially heterogeneous framework for analyzing panel data},
	volume = {77},
	year = {2015}}

@article{ando2016panel,
	author = {Ando, Tomohiro and Bai, Jushan},
	journal = {Journal of Applied Econometrics},
	number = {1},
	pages = {163--191},
	publisher = {Wiley Online Library},
	title = {Panel data models with grouped factor structure under unknown group membership},
	volume = {31},
	year = {2016}}

@article{kwon_optimal,
  title={Optimal shrinkage estimation of fixed effects in linear panel data models},
  author={Kwon, Soonwoo},
  journal={Econometrica},
  year={2026},
  note={Forthcoming}
}

@article{liu2021panel,
  title={Panel forecasts of country-level {C}ovid-19 infections},
  author={Liu, Laura and Moon, Hyungsik Roger and Schorfheide, Frank},
  journal={Journal of Econometrics},
  volume={220},
  pages={2--22},
  year={2021},
  publisher={Elsevier}
}

@article{liu2023forecasting,
  title={Forecasting with a panel {T}obit model},
  author={Liu, Laura and Moon, Hyungsik Roger and Schorfheide, Frank},
  journal={Quantitative Economics},
  volume={14},
  pages={117--159},
  year={2023},
  publisher={Wiley Online Library}
}

@article{liu2020forecasting,
  title={Forecasting with dynamic panel data models},
  author={Liu, Laura and Moon, Hyungsik Roger and Schorfheide, Frank},
  journal={Econometrica},
  volume={88},
  number={1},
  pages={171--201},
  year={2020},
  publisher={Wiley Online Library}
}

@article{pollard1982central,
	author = {Pollard, David},
	journal = {The Annals of Probability},
	number = {4},
	pages = {919--926},
	publisher = {Institute of Mathematical Statistics},
	title = {A central limit theorem for $ k $-means clustering},
	volume = {10},
	year = {1982}}

@article{chang2022central,
	author = {Chang, Jinyuan and Chen, Xiaohui and Wu, Mingcong},
	journal = {Bernoulli},
	number = {1},
	pages = {712-742},
	title = {Central limit theorems for high dimensional dependent data},
	volume = {30},
	year = {2024}}

@article{dzemskiokui2021convergence,
	author = {Andreas Dzemski and Ryo Okui},
	journal = {Economics Letters},
	pages = {109844},
	title = {Convergence rate of estimators of clustered panel models with misclassification},
	volume = {203},
	year = {2021}}

@article{mehrabani2022estimation,
	author = {Mehrabani, Ali},
	journal = {Journal of Econometrics},
	number = {2},
	pages = {1464-1482},
	title = {Estimation and identification of latent group structures in panel data},
	volume = {235},
	year = {2023}}

@article{bonhomme2015grouped,
	author = {Bonhomme, St{\'e}phane and Manresa, Elena},
	journal = {Econometrica},
	number = {3},
	pages = {1147--1184},
	publisher = {Wiley Online Library},
	title = {Grouped patterns of heterogeneity in panel data},
	volume = {83},
	year = {2015}}

@article{Andrews91,
	author = {Donald W. K. Andrews},
	journal = {Econometrica},
	number = {3},
	pages = {817--858},
	title = {Heteroskedasticity and Autocorrelation Consistent Covariance Matrix Estimation},
	volume = {59},
	year = {1991}}

@article{wang2016homogeneity,
	author = {Wuyi Wang and Peter C. B. Phillips and Liangjun Su},
	journal = {Journal of Applied Econometrics},
	pages = {797--815},
	title = {Homogeneity pursuit in panel data models: theory and applications},
	volume = {33},
	year = {2018}}

@article{liu2020identification,
	author = {Liu, Ruiqi and Shang, Zuofeng and Zhang, Yonghui and Zhou, Qiankun},
	journal = {Journal of Econometrics},
	number = {2},
	pages = {574--590},
	publisher = {Elsevier},
	title = {Identification and estimation in panel models with overspecified number of groups},
	volume = {215},
	year = {2020}}

@online{chetverikov2022spectral,
  author       = {Chetverikov, Denis and Manresa, Elena},
  title        = {Spectral and post-spectral estimators for grouped panel data models},
  year         = {2022},
  eprint       = {2212.13324},
  archivePrefix={arXiv},
  primaryClass={stat.ME},
  url          = {https://arxiv.org/abs/2212.13324},
}

@article{wang2019heterogeneous,
	author = {Wuyi Wang and Peter C. B. Phillips and Liangjun Su},
	journal = {Economics Letters},
	pages = {179--185},
	title = {The heterogeneous effects of the minimum wage on employment across states},
	volume = {174},
	year = {2019}}

@article{su_identifying_2016,
	author = {Su, Liangjun and Shi, Zhentao and Phillips, Peter C. B.},
	issn = {0012-9682},
	journal = {Econometrica},
	number = {6},
	pages = {2215--2264},
	title = {Identifying Latent Structures in Panel Data},
	url = {https://www.jstor.org/stable/44155362},
	urldate = {2024-03-25},
	volume = {84},
	year = {2016}}

@article{okui_heterogeneous_2021,
	author = {Okui, Ryo and Wang, Wendun},
	doi = {10.1016/j.jeconom.2020.04.009},
	issn = {0304-4076},
	journal = {Journal of Econometrics},
	month = feb,
	number = {2},
	pages = {447--473},
	series = {Annals {Issue}: {Celebrating} 40 {Years} of {Panel} {Data} {Analysis}: {Past}, {Present} and {Future}},
	title = {Heterogeneous structural breaks in panel data models},
	url = {https://www.sciencedirect.com/science/article/pii/S0304407620301287},
	urldate = {2024-03-25},
	volume = {220},
	year = {2021}}

@article{dzemski2024confidence,
	author = {Dzemski, Andreas and Okui, Ryo},
	journal = {Quantitative Economics},
	number = {2},
	pages = {245--277},
	publisher = {Wiley Online Library},
	title = {Confidence set for group membership},
	volume = {15},
	year = {2024}}

@article{lin2012estimation,
	author = {Lin, Chang-Ching and Ng, Serena},
	journal = {Journal of Econometric Methods},
	number = {1},
	pages = {42--55},
	publisher = {De Gruyter},
	title = {Estimation of panel data models with parameter heterogeneity when group membership is unknown},
	volume = {1},
	year = {2012}}

@article{hahn2010panel,
	author = {Hahn, Jinyong and Moon, Hyungsik Roger},
	journal = {Econometric Theory},
	number = {03},
	pages = {863--881},
	publisher = {Cambridge Univ Press},
	title = {Panel data models with finite number of multiple equilibria},
	volume = {26},
	year = {2010}}

@article{dube2010minimum,
	author = {Dube, Arindrajit and Lester, T. William and Reich, Michael},
	journal = {The Review of Economics and Statistics},
	number = {4},
	pages = {945--964},
	publisher = {MIT Press},
	title = {Minimum wage effects across state borders: Estimates using contiguous counties},
	volume = {92},
	year = {2010}}

@article{pollard1981strong,
	author = {Pollard, David},
	journal = {The Annals of Statistics},
	pages = {135--140},
	publisher = {JSTOR},
	title = {Strong consistency of k-means clustering},
	year = {1981}}

@article{arellano1987computing,
	author = {Arellano, Manuel},
	journal = {Oxford Bulletin of Economics and Statistics},
	pages = {431--434},
	title = {Computing Robust Standard Errors for Within-groups Estimators},
	volume = {49},
	year = {1987}}

@article{hansen2007asymptotic,
	author = {Hansen, Christian B.},
	journal = {Journal of Econometrics},
	pages = {597--620},
	title = {Asymptotic properties of a robust variance matrix estimator for panel data when {$T$} is large},
	volume = {141},
	year = {2007}}

@article{driscoll1998consistent,
	author = {Driscoll, John C. and Kraay, Aart C.},
	journal = {The Review of Economics and Statistics},
	pages = {549--560},
	title = {Consistent Covariance Matrix Estimation with Spatially Dependent Panel Data},
	volume = {80},
	year = {1998}}

\end{document}